\documentclass[draftclsnofoot,journal,onecolumn, 12pt]{IEEEtran}

\usepackage{graphics, color, amsmath, amssymb}

\usepackage{fancybox}

\newtheorem{lem}{Lemma}
\newtheorem{thm}{Theorem}
\newtheorem{prop}{Proposition}
\newtheorem{defn}{Definition}
\newtheorem{cor}{Corollary}

\begin{document}
%

\title{Is SP BP?}
%
%
%

\author{{\bf \large Ronghui Tu},
 {\bf \large Yongyi Mao}
        and
{\bf \large Jiying Zhao}\\
School of Information Technology and Engineering\\
University of Ottawa\\
800 King Edward Avenue\\
 Ottawa, Ontario, K1N 6N5, Canada\\
Email: \{rtu, yymao, jyzhao\}@site.uottawa.ca
}

%


\maketitle

\begin{abstract}
The Survey Propagation (SP) 
algorithm for solving $k$-SAT problems has been shown recently
 as an instance of the Belief Propagation (BP) algorithm. In this paper, we
show that for general constraint-satisfaction problems, SP may not be reducible from BP. We also establish the conditions under which such a reduction is 
possible. Along our development, we present a unification of the existing SP algorithms in terms of a probabilistically interpretable iterative procedure --- 
weighted Probabilistic Token Passing.  
\end{abstract}


\vfill

\begin{center}
\scalebox{0.3}{\begin{picture}(0,0)%
\includegraphics{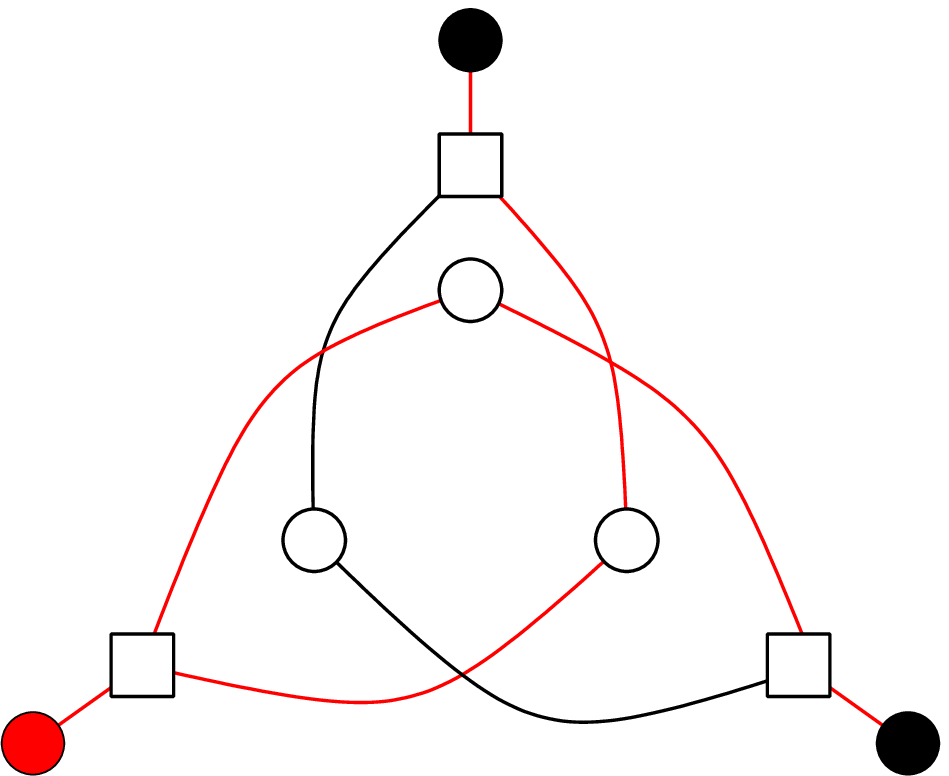}%
\end{picture}%
\setlength{\unitlength}{3947sp}%
\begingroup\makeatletter\ifx\SetFigFont\undefined%
\gdef\SetFigFont#1#2#3#4#5{%
  \reset@font\fontsize{#1}{#2pt}%
  \fontfamily{#3}\fontseries{#4}\fontshape{#5}%
  \selectfont}%
\fi\endgroup%
\begin{picture}(4516,3689)(3893,-4043)
\end{picture}%
}
\end{center}

\vfill
\vfill

\mbox{}

\clearpage

\section{Introduction}
\label{sec:intro}

Survey Propagation (SP) \cite{Mezard} is a recent algorithmic
breakthrough in solving certain hard families of constraint
satisfaction problems (CSPs). Derived from statistical physics, SP
first demonstrated its power in solving classic prototypical
NP-complete problems, the $k$-SAT problems \cite{cook}.
--- For random instances of these problems in the hard regime,
SP is shown to be the first efficient solver \cite{Mezard}. Recently, SP has also been applied
to other CSPs, including other NP-complete problem families
such
as graph coloring (or $q$-COL) problems \cite{Braunstein:3col}, as
well as problems arising in communications and data compressions,
some examples being coding for Blackwell channels
\cite{Yu:SP_Blackwell} and quantization of Bernoulli sequences
\cite{Wainwright:SP_compression}.  In all these cases, great
successes have been demonstrated.

Powerful as it appears, SP however largely remains as a heuristic
algorithm to date, where analytic understanding of its algorithmic
nature and rigorous characterization of its performance are widely
open and of great curiosity and research importance.

Similar to the well-known Belief Propagation (BP) algorithm used in
iterative decoding \cite{Richardson:98} and statistical inference
\cite{Pearl:BP}, SP operates by iteratively passing ``messages'' in
a factor graph representation \cite{Frank:factor} of the problem
instance, where each variable vertex corresponds to a variable whose
value is to be decided and each function vertex corresponds to a
local constraint imposed on the variables. This observation has
inspired a recent research effort in understanding whether SP may be
viewed as a special case of BP.  --- The significance of questions
of such a kind has been witnessed repeatedly in the history of
communication research, for example, in understanding the Viterbi
algorithm as a dynamic programming algorithm \cite{Forney:viterbi},
in understanding the turbo decoding algorithm \cite{Turbo} as an
instance of Belief Propagation \cite{turbo_as_BP},  and in unifying
the BCJR algorithm \cite{BCJR} and the Viterbi algorithm under the
umbrella of the generalized distributive law \cite{GDL}, etc.  These
unified frameworks have on one hand provided additional insights
into the nature of the algorithms, and on the other hand allowed an
easier access of the algorithm by much wider research communities.
Specific to the question ``is SP BP'', if SP may be understood as an
instance of BP, then the existing analytic techniques of BP are
readily applicable to analyzing SP; if SP can not be characterized
as a special case of BP, one is then motivated to seek a different
algorithmic framework to which SP belongs or to discover the unique
algorithmic nature of SP.

The first result reporting that SP is an instance of BP is the work
of \cite{Braunstein:SP_BP} in the context of $k$-SAT problems.  This
result is generalized in \cite{SP:newlook} to an extended version of
SP for solving $k$-SAT
problems. Briefly, the authors of \cite{SP:newlook} present a Markov
Random Field (MRF) \cite{MRF} formalism for $k$-SAT problems; a
parameter, denoted by $\gamma$ in this paper, is used to parametrize
the MRF. When the BP algorithm is derived on such an MRF, the BP
message-update equations result in a {\em family} of SP algorithms,
referred to as {\em weighted SP} or $\mbox{SP}(\gamma)$ in this paper,
parametrized by $\gamma\in [0,1]$; and when $\gamma=1$,
$\mbox{SP}(\gamma)$ is the original (non-weighted) SP. In addition
to extending SP --- in the context of $k$-SAT problems --- to a
family of SP algorithms with tunable performance, another
significance of this result is a conclusive answer to the titular
question in that context, namely that SP is BP for the $k$-SAT problem
family. This result was re-developed in our earlier work
\cite{rtu:sp_isit_06} where a simpler MRF formalism using Forney
graphs \cite{Forney:normalGraph} is presented and a more transparent
reduction of BP messages to weighted SP messages is given.

The objective of this paper is to answer the question whether SP and
more generally weighted SP are special cases of  BP for
arbitrary CSPs beyond $k$-SAT problems. It is
worth noting that weighted SP has only been presented for $k$-SAT
problems, although its principle may be extended to designing other
CSPs involving {\em binary} variables (see, e.g., \cite{Wainwright:SP_compression}).
Furthermore, resulting from BP on a properly defined MRF,
weighted SP, unlike the original (non-weighted) SP, does not have a
probabilistic interpretation that {\em does not} rely on the MRF
constructed in the style of \cite{SP:newlook} or \cite{rtu:sp_isit_06}
and the derived BP algorithm
thereby. Thus to answer the question whether weighted SP is BP for
general CSPs, it is necessary to formulate weighted SP for arbitrary
CSPs that generalizes non-weighted SP without relying on any MRF and
BP formalism. For this reason, this research and hence the structure
of this paper roughly split into two parts. The first part answers
the question what SP and weighted SP exactly are by presenting a
probabilistically interpretable formulation of both non-weighted and
weighted SP for arbitrary CSPs. The second part presents a MRF
formalism for general CSPs in the style of \cite{SP:newlook} or
\cite{rtu:sp_isit_06},
derives the BP update equations, and answers the question whether and how
BP under such MRF formalism may be reduced to SP, if at all.

Although this paper focuses on the second part, namely, on answering whether SP algorithms are instances of BP on a properly defined MRF, our effort in
establishing what SP algorithms are and how to formulate these algorithms for
general CSPs is noteworthy.

First, the notion of
weighted SP, as noted earlier, has only been presented for $k$-SAT
problems as in \cite{SP:newlook} and in sporadic example applications
involving only {\em binary} variables such as in \cite{Wainwright:SP_compression}.  As will become
clear in this paper, the design philosophy of weighted SP for CSPs
involving binary variables (such as in \cite{SP:newlook} and
\cite{Wainwright:SP_compression}) is not
readily extendable to arbitrary CSPs with arbitrary variable
alphabets, since an important notion underlying SP,  namely, an
{\em appropriate} extension of variable alphabets, is blurred in the binary special
cases.

Second, for non-weighted SP, we note that its formulation in the context of
 general CSPs primarily exists
in the literature of statistical physics (see, e.g.,
\cite{Braunstein}). Although its design recipe has been laid out for
arbitrary CSPs, its exposition in statistical physics language has
made it rather difficult for readers with primarily engineering or
computer science background.

Thus, in addition to serving as the basis for
the investigation of BP-to-SP reduction, the first part of the paper
also aims at providing a clean, transparent and easily accessible
formulation of SP algorithms
in its most general form for arbitrary CSPs,
without resorting to statistical physics
concepts.

\section{Main Results and Paper Organization}
\label{sec:intro_summary}

The main results of this paper are summarized as follows.

In the first part, we formulate SP and weighted SP for general CSPs
as what we call ``probabilistic token passing'' (PTP) and ``weighted
probabilistic token passing'' (weighted PTP) respectively, where a
message is a distribution (or non-negative function) on the set of
``tokens'' associated with a variable. Here a ``token'' is a
non-empty {\em subset} of the variable's alphabet\footnote{In fact more rigorously,
a token is a non-empty {\em subset} of all possible {\em assignments}
of a variable -- In this paper, for more mathematical rigor and 
clarity, we make a distinction between the alphabet of a variable and the set of all assignments to the variable, where an assignment to variable $x_v$
is treated as a function mapping the singleton set $\{v\}$ to the alphabet of $x_v$.  Nevertheless, one may always identify the set of all assignments to $x_v$ with the alphabet of $x_v$ via a one-to-one correspondence and loosely refer to the set of all assignments of 
a variable as the alphabet of the variable.}. It has been previously
observed in SP applied to various problems that a ``joker'' symbol
is added to the original variable alphabet. Here we point out that
extending the alphabet by simply adding a joker symbol is not
sufficient for general CSPs, particularly for those involving
non-binary variables. We stress that the {\em right} extension of
the variable alphabet is to replace it with the set of all non-empty
subsets of the original alphabet. Although an equivalent treatment
has been described in some previous literature for non-weighted SP
\cite{Braunstein}, this perspective is for the first time made
explicit beyond statistical physics context and for both
non-weighted and weighted SP. Based on this notion of alphabet
extension, we generalize 
weighted SP for arbitrary CSPs in the form
of weighted PTP. In other words, the weighted PTP formulation
presented in this paper serves as a recipe for designing weighted SP
algorithm for arbitrary CSPs.

In the second part, we present an MRF formalism --- which we refer to as
``normally realized MRF'' ---  for arbitrary CSPs using Forney
graphs, generalizing the MRF construction in the style of \cite{SP:newlook}
and \cite{rtu:sp_isit_06} presented for $k$-SAT problems. States, each consisting of a left
state and a right state, are introduced in the MRF, where the left
state corresponds to the token passed from the variable and the
right state corresponds to the token passed from the constraint. For
any given CSP, the MRF is parametrized by a collection of weighting
functions, each corresponding to a variable in the CSP; in the
$k$-SAT special case, these weighting functions may reduce to a
single parameter, $\gamma$. Noting the combinatorial importance of
such MRF in the context of $k$-SAT problems \cite{SP:newlook}, one expects
that this general formulation of MRF for arbitrary CSP may serve a
similar role, namely
providing a combinatorial framework describing the topology of the
solution space \cite{SP:newlook}. This direction, clearly deserving further
investigation, is however out of the scope of this paper.

On the normally realized MRF formalism, we then proceed to derive the BP
update equations and investigate the reduction of BP to weighted PTP
(noting that weighted PTP {\em is} weighted SP and that non-weighted
SP is a special case of weighted SP). Primarily re-developing the
results of \cite{SP:newlook} and \cite{rtu:sp_isit_06} on BP-to-SP
reduction, we show that for $k$-SAT problems, BP is readily
reducible to weighted PTP as long as a condition --- which we refer
to as the {\em state-decoupling condition} --- is imposed on the BP
messages in initialization. An interesting
fact about this condition in the context of $k$-SAT problems is that
as long as the condition is satisfied in the first BP iteration, it
will continue to be satisfied in all iterations after. This forms
the basis on which BP messages may be simplified to the form of
weighted PTP messages. This condition, also arising in
\cite{SP:newlook} and \cite{rtu:sp_isit_06} as a peculiar and
curious construction, had not been explained prior to this work. In
this paper, we argue that the state-decoupling condition serves a
critical role in the reduction of the weighted PTP messages from the
BP messages derived from the MRF formalism in the style of
\cite{SP:newlook} and \cite{rtu:sp_isit_06}, or from the normally
realized MRF presented in this paper. Using the example of $3$-COL
problems, we show that such a condition is also needed in all BP
iterations so as for BP to reduce to PTP. However, in that case, we
show that this condition can not be made satisfied in every BP
iteration (except for the trivial cases in which the BP messages contain no
useful information) and one must manually impose this condition by
manipulating the BP messages in each iteration. This result on one
hand justifies the important role of the state-decoupling condition
in the reduction of BP to PTP and on the other hand asserts that BP
is {\em not} PTP and hence {\em not} SP for $3$-COL problems!

At that point, one is ready to conclude that weighted PTP or
weighted SP is not a special case of BP for general CSPs. 
The manual manipulation of BP messages in $3$-COL problems, which results in
what we call {\em state-decoupled BP} brings up a further question, namely, for 
general CSPs, whether PTP and weighted PTP are readily expressed as state-decoupled BP. We proceed to show that for general CSPs, the reduction of weighted PTP
from BP requires yet another condition pertaining to the structure
of the CSP. Briefly, this additional condition demands that the
constraints in the CSP be ``locally compatible'' with each other in
some sense. We show that the local compatibility condition of the CSP is the
necessary and sufficient condition for state-decoupled 
BP to reduce to weighted PTP or weighted SP. At that end, we complete
the answer to the titular question ``is SP BP?''.

As mentioned earlier, in addition to answering whether SP is BP,
another objective of this paper is to explain SP as simply as
possible. For this purpose, we have made an effort in presenting
this paper in a pedagogical manner and carrying along the examples of
$k$-SAT and $3$-COL problems throughout the paper.

The remainder of this paper is organized as follows. In Section
\ref{sec:csp}, we present a generic formulation of CSPs while also
introducing various notations that will be used in later parts of
the paper. In Section \ref{sec:sp}, we introduce the existing SP 
algorithms using the
examples of $k$-SAT problems and $3$-COL problems, where we
purposefully avoid SP formulations in statistical physics languages.
We then proceed in Section \ref{sec:token} to present a general
formulation of SP algorithms in terms of PTP and weighted PTP. In
Section \ref{sec:bp}, we present the normally realized MRF formalism
and present results concerning the reduction of BP messages to SP
messages. At this time, how SP algorithms behave over iterations and
how they solve a CSP are important open problems. Although such
questions are not of particular importance for the purpose of this
paper, completely ignoring them appears not satisfactory to us and
perhaps also to some readers. For this reason, we present some
preliminary results along those lines for understanding the dynamics of
 PTP. ---
These results are included in the Appendix so as to maintain the
focus of this paper.  The paper is briefly concluded in Section
\ref{sec:conclude}.


\section{A Generic Formulation of Constraint Satisfaction Problems}
\label{sec:csp}

Let $V$ be a finite set, in which each element will be referred to
as a {\em coordinate}.
Associated with each
$v\in V$, there is a
finite {\em alphabet} $\chi_v$.  For each $v\in V$, we will assume throughout
of this paper that
every $\chi_v$ is identical to each other, and is therefore denoted by $\chi$.
We note that this slight loss of generality is made only for lightening the
upcoming notations, and that there is no difficulty to extend the results of
this paper to more general cases where $\chi_v$'s are different from each other.
For any subset $U\subseteq V$, a {\em $\chi$-assignment}
 $x_U$ on $U$
is a function mapping $U$ into the
set $\chi$.  That is, a  $\chi$-assignment $x_U$ specifies a way to assign
each coordinate $u\in U$ a value in $\chi$. 
The set of all $\chi$-assignments on $U$
will be denoted by $\chi^U$.
When $U$ is a singleton set
$\{u\}$, which contains a single coordinate $u$,
 we will call $\chi$-assignment $x_{\{u\}}$ on $\{u\}$ an 
{\em elementary ($\chi$-)assignment} and
write it
as $x_u$ for simplicity. Clearly, any given elementary $\chi$-assignment
$x_u$ is uniquely specified by
a value $r\in \chi$, which is the assigned value in $\chi$ to coordinate $u$.
In this case, this assignment is denoted by $r_{u}$, for example,
if $\chi:=\{0, 1\}$, then the only possible $\chi$-assignments
on $\{u\}$ are $0_u$ and $1_u$, which are
the elementary
assignments assigning $0$ and $1$ to coordinate $u$, respectively.

Suppose that $U\subset W\subseteq V$ and that $x_W$ is a $\chi$-assignment
on $W$. We
will use $x_{W:U}$ to denote the (function) restriction of $x_W$ on $U$.
For any subset of $\chi$-assignments $\Omega\subseteq \chi^W$ on $W$,
we denote the projection of $\Omega$ on $U$ by $\Omega_{:U}$. That is,
\begin{equation*}
\Omega_{:U}:=\{x_{W:U}:x_W\in \Omega\}.
\end{equation*}

If coordinate set
$U$ can be partitioned into disjoint subsets $A$ and $B$,
then it is obvious that
assignment $x_U$ decomposes into assignments $x_{U:A}$ and $x_{U:B}$,
and $x_U$ may be
written as $(x_{U:A}, x_{U:B})$ (in any order).  Evidently, $x_U$ may be
decomposed according to any partition of $U$, not necessarily two-fold partitions. In particular, if a collection of sets $\{U_i:i\in {\cal I}\}$,
for some ${\cal I}$, form a partition of $U$,
then we may assignment $x_U$ as $\langle x_{U:U_i} \rangle_{i\in {\cal I}}$.



For simplicity, we
will write
$(x_A, x_B)$ and $\langle x_{U_i}\rangle_{i\in {\cal I}}$ in place of
$(x_{U:A}, x_{U:B})$ and $\langle x_{U:U_i} \rangle_{i\in {\cal I}}$
respectively. In fact, unless
some particular clarity is needed, we will always write $x_{W:U}$ simply as
$x_U$, making the underlying $x_W$ implicit. Furthermore, when $U$ is a
singleton set $\{u\}$, as mentioned earlier, we will simply denote it
by $x_u$, which reduces to the conventional ``variable'' notation 
standard literatures of graphical models.

Given $\chi$ and $V$, the objective of a constraint satisfaction
problem (CSP) is to find a global $\chi$-assignment $x_V$ that
satisfies a given set of constraints or to conclude that no such
assignment exists.  Formally, we will use set $C$ to index the set
of constraints $\{\Gamma_c:c\in C\}$. Each constraint $\Gamma_c$,
$c\in C$, applies to a subset of the coordinates $V$, which will be
denoted by $V(c)$. Specifically, each constraint $\Gamma_c$ is
identified with a subset of $\chi^{V(c)}$, and the constraint is
satisfied by global $\chi$-assignment $x_V$ if $x_{V:V(c)}\in
\Gamma_c$. Then any CSP may be formulated via specifying $V$, $C$,
$\chi$, $\{V(c):c\in C\}$ and $\{\Gamma_c:c\in C\}$, where the
objective of the CSP is to find a $\chi$-assignment $x_V$ such that
\begin{equation}
\label{equ:csp} \prod_{c\in C}[x_{V:V(c)}\in \Gamma_c] = 1,
\end{equation}
or to conclude that no such assignment exists.
Here the notation $[P]$, for any Boolean proposition $P$, is the
Iverson's convention \cite{Frank:factor}, namely, evaluating to 1 if
$P$, and to 0 otherwise.

Now it is easy to verify that the factorization structure of
(\ref{equ:csp}) can be represented by a factor graph
\cite{Frank:factor}: in the factor graph,
``variable vertices'' are
indexed by $V$, where the ``variable'' indexed by $v\in V$ represents an
elementary assignment $x_{V:\{v\}}$
on $\{v\}$, or simply $x_v$; ``function vertices'' are indexed by $C$,
where the function indexed by
$c\in C$ is $[x_{V:V(c)}\in \Gamma_c]$, which, with a slight overloading of notation,
will also be denoted by $\Gamma_c(x_{V(c)})$;
there is an edge connecting variable vertex
$x_v$ with function vertex $\Gamma_c$ if and only if $v\in V(c)$.
Inspired by its correspondence (to an edge) in the factor graph,
we will use $(v-c)$ to denote a coordinate-constraint pair $(v,c)$
where coordinate $v$ is involved in constraint $\Gamma_c$ in the CSP.



For notational symmetry, we denote the set $\{c:v\in
V(c)\}$ by $C(v)$, namely, $C(v)$ indexes the set of all constraints
involving  coordinate $v$, or the set of all function vertices connecting
to variable vertex $x_v$. We will assume
that $|C(v)|\ge 2$ for all $v\in V$.  It is clear that such an assumption is
without loss of generality, since if a variable $x_v$ is
involved in only one constraint, one may always modify the  constraint
and remove the variable from the problem. Similarly, we will assume that
$|V(c)|\ge 2$ for every $c\in C$. This is also without loss of generality
since if a constraint $\Gamma_c$ only involves a single variable $x_v$, it is
always possible to ``absorb'' this constraint in other constraints
involving $x_v$ (noting that $x_v$ must have another constraint since $|C(v)|
\ge 2|$).

\subsection{$k$-SAT}
The $k$-SAT problems are a classic family of CSPs, known to be
NP-complete for $k\ge 3$ \cite{cook}. An instance of $k$-SAT
problems consists of a set of variables $\{x_v:v\in V\}$, each of
which takes on values from the set $\chi:=\{0,1\}$, and a set of
constraints $\{\Gamma_c:c\in C\}$, each of which involves exactly
$k$ variables. For each constraint $\Gamma_c$ and every $v\in V(c)$,
there is a value $L_{v, c}\in \{0, 1\}$ which we will refer to as
the {\em preferred value} on $v$ in constraint $\Gamma_c$. The
$k$-SAT problem is then to decide on an assignment $x_V$ such that
for each constraint $\Gamma_c$, at least one of its involved
coordinate is assigned its preferred value in $\Gamma_c$. To map
back to the afore-mentioned set-theoretic formulation of
constraints, in a $k$-SAT problem, for each $c\in C$, let $l^c$
denote the $\chi$-assignment on $V(c)$ in which every coordinate
$v\in V(c)$ is assigned the negated value $\bar{L}_{v, c}$ of its
preferred value $L_{v, c}$ in $\Gamma_c$, namely that
$l^c_{:\{v\}}=\bar{L}_{v,c}$ for every $(v-c)$,  then constraint
$\Gamma_c$ is defined as $\Gamma_c:=\chi^{V(c)}\setminus \{l^c\}$.


The factor-graph representation of a toy 3-SAT problem is shown in
Fig. \ref{fig:ksat_fg}. For $k$-SAT problems, it is convenient to
treat each preferred value $L_{v,c}$ as the label for edge
$(x_v,\Gamma_c)$ on the factor graph, and use dashed edge to
represent label 0 and solid edge to represent label 1.

We note that it is customary in this paper that variable vertices in
a factor graph are listed on the left side and function (constraint) vertices
listed on the right side.

\begin{figure}[htb]
\begin{center}
  \scalebox{0.5}{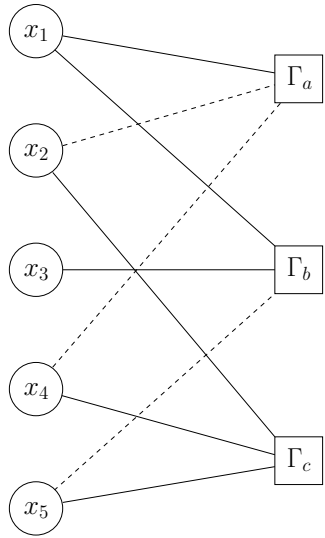}
  \caption{
A factor graph for $3$-SAT problem specified by
formula
$(x_{1}\vee\overline{x}_{2}\vee\overline{x}_{4})\wedge(x_{1}\vee
x_{3}\vee\overline{x}_{5})\wedge(x_{2}\vee x_{4}\vee x_{5})$. Logic
operation notations  are used here to define the problem, where
$\vee$ denotes logic {\tt OR}, $\wedge$ denotes logic {\tt AND}, and
the horizontal bar on a variable denotes the negation of the
variable. The function represented by the factor graph is
$[(x_1,x_2,x_4)\in\Gamma_a]\cdot[(x_1,x_3,x_5)\in\Gamma_b]\cdot[(x_2,x_4,x_5)\in\Gamma_c]$,
where $\Gamma_a=\chi^{\{1, 2, 4\}}\setminus \{(0_1,1_2,1_4)\}$,
$\Gamma_b=\chi^{\{1, 3, 5\}}\setminus\{(0_1,0_3,1_5)\}$, and
$\Gamma_c=\chi^{\{2, 4, 5\}}\setminus\{(0_2,0_4,0_5)\}$.} \label{fig:ksat_fg}
\end{center}
\end{figure}

\subsection{Graph Coloring}

 Graph coloring  or $q$-COL problems are another family of
NP-complete problems.  Given an undirected
graph $(\Delta, \Xi)$ with vertex set $\Delta$ and edge set $\Xi$, the objective of the $q$-COL
problem on $(\Delta, \Xi)$ is
to assign each vertex in $\Delta$ a color from $q$ different colors
such that every pair of adjacent vertices have different colors. To
use the above generic formulation of CSPs, we will denote the set of
all $q$ colors by set $\chi:=\{1, 2, \ldots, q\}$. We will denote
every undirected edge in $\Xi$, say the edge connecting vertices $u$
and $v$, by set $\{u,v\}$. The set $V$ of all coordinates is then
identified with set $\Delta$, and the set $C$ indexing all constraints is
identified with $\Xi$. Specifically note that every $c\in C$ is then
identified with some $\{u, v\}\in \Xi$, and $V(c)$ is identified
with $c$, or the corresponding set $\{u, v\}$. Suppose that
$c=\{u,v\}\in \Xi$, then constraint $\Gamma_c$ is identified with
$\chi^{\{u, v\}}\setminus \{(1_u, 1_v), (2_u, 2_v), \ldots, (q_u,
q_v)\}$.
Fig. \ref{fig:qcol_fg}(b) shows
the factor-graph representation of a $q$-COL problem on the
undirected graph shown in Fig. \ref{fig:qcol_fg}(a).

\begin{figure}[htb]
\begin{center}
    \begin{minipage}[t]{4.0cm}
    \begin{center}
        \scalebox{0.5}{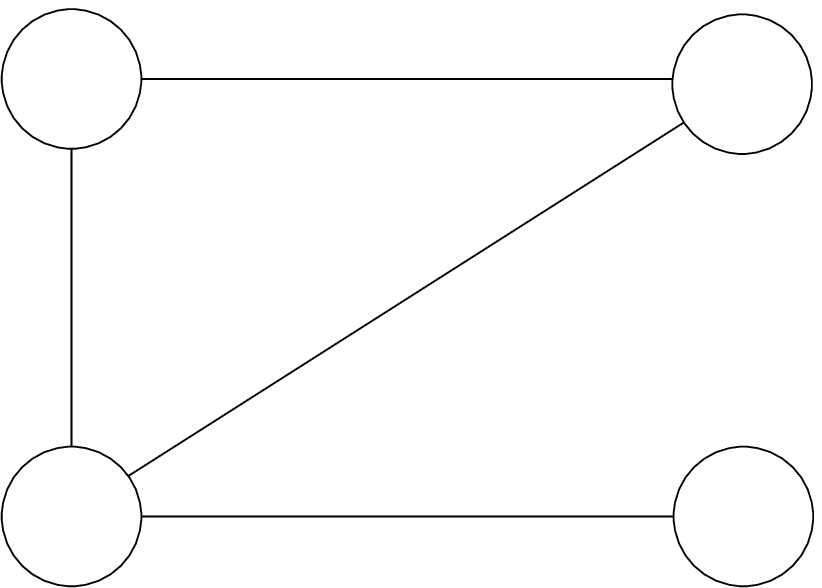}
        \scriptsize{(a)}
    \end{center}
    \end{minipage}
    \hspace{1.5cm}
    \begin{minipage}[t]{4.0cm}
    \begin{center}
        \scalebox{0.5}{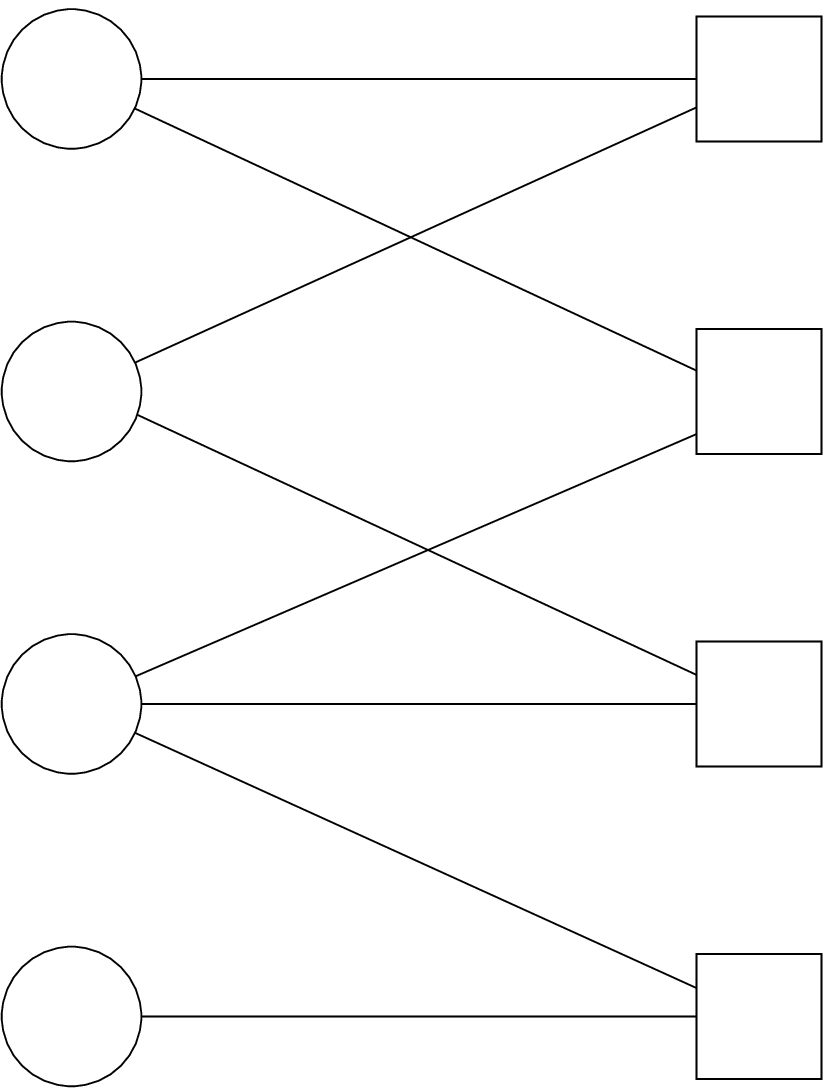}
        \scriptsize{(b)}
    \end{center}
    \end{minipage}
  \caption{(a) An undirected graph. (b) The factor graph for a $q$-COL problem on graph (a). The global function represented by the factor graph
is $[(x_1,x_2)\in \Gamma_{\{1, 2\}}]\cdot
[(x_1, x_3)\in \Gamma_{\{1,3\}}]\cdot
[(x_2, x_3)\in \Gamma_{\{2,3\}}]\cdot
[(x_3, x_4)\in \Gamma_{\{3,4\}}]$, where $\Gamma_{\{u,v\}}:=
\chi^{\{u,v\}}\setminus \{(1_u, 1_v), (2_u, 2_v), \ldots, (q_u, q_v)\}$.
  \label{fig:qcol_fg}}
\end{center}
\end{figure}

\section{Survey Propagation Algorithms}
\label{sec:sp}

\subsection{Survey Propagation for $k$-SAT Problems}

Extensive study has been carried out to understand the hardness of
$k$-SAT problems (for $k\ge 3$) and to develop efficient solvers. A
parameter $\alpha:=|C|/|V|$ is observed to be critically related to
the hardness of random $k$-SAT problems. There appear two thresholds
of $\alpha$, denoted by $\alpha_d$ and $\alpha_c$,
($\alpha_d<\alpha_c$), marking two ``phase transitions''
\cite{Mezard}. When $\alpha>\alpha_c$, random $k$-SAT problems are
unsatisfiable (i.e., having no satisfying assignment) with high
probability; when $\alpha_d<\alpha<\alpha_c$, the satisfying
assignments form exponentially many disjoint ``clusters'', making
the problem extremely difficult; when $\alpha<\alpha_d$, the
satisfying assignments merge into one huge cluster and problems are
easier. In the regime of $\alpha<\alpha_d$, local search algorithms,
such as BP, may find a satisfying assignment. In the regime of
$\alpha_d<\alpha<\alpha_c$, local search algorithms usually fail.

The discovery and first application of survey propagation (SP) are
in solving the $k$-SAT problems in the hard regime, where messages
are passed on the above-defined factor graphs \cite{Mezard}. In SP,
a ``joker'' symbol ``$*$'' is introduced to variable alphabet $\chi$
of the $k$-SAT problem, where $x_v$ equal to the ``joker'' indicates
that it is free to take any value from its original alphabet, and
that $x_v$ equals a non-joker symbol indicates that it is
constrained to taking the designated value. Briefly, SP on $k$-SAT
problems may be viewed as an iterative method for estimating the
``biases'' of each variable $x_v$ on $0, 1$ and $*$ respectively and
a variable that is highly biased on $0$ or $1$ can be fixed to that
value whereby simplifying the problem. It is shown that in the hard
regime of random $k$-SAT problems, the ``joker'' symbol connects the
disconnected clusters, making SP remain very effective
 even for $\alpha$ very close to
$\alpha_c$ \cite{SP:newlook}. For $k$-SAT problems, the original
version of SP \cite{Mezard} is generalized in \cite{SP:newlook} to
what we call the {\em weighted SP} \footnote{In \cite{SP:newlook},
weighted SP is referred to as generalized SP. In this paper, we
would like to reserve the term ``generalized SP'' to refer to SP
algorithms generalized for arbitrary CSPs beyond $k$-SAT problems. }
or $\mbox{SP}(\gamma)$ in this paper. $\mbox{SP}(\gamma)$ is a
family of algorithms parametrized by a real number $\gamma\in [0,
1]$, where $\mbox{SP}(1)$ is the original SP and for some judicious
choice of $\gamma\in (0, 1)$, $\mbox{SP}(\gamma)$ may have further
improved performance.

We note that generalizing SP to the family of weighted SP algorithms
has only been reported for $k$-SAT problems to date, and one of the
objectives of this paper is to extend such a generalization to
arbitrary CSPs.

Similar to BP, in the SP algorithms, messages are passed between
variable vertices and function vertices. For the purpose of
describing the SP message-update rule for $k$-SAT problems, we
introduce the following notations. For any $(v-c)$, $C_c^{\tt u}(v)$
denotes the set $\{b\in C(v)\setminus \{c\}:L_{v,b}\neq L_{v, c}\}$,
and $C_c^{\tt s}(v)$ denotes the set $\{b\in C(v)\setminus
\{c\}:L_{v,b} = L_{v, c}\}$.

Following \cite{SP:newlook}, the message-update rule of
$\mbox{SP}(\gamma)$ is described as follows.

The message passed from variable vertex $x_v$ to function vertex
$\Gamma_c$ --- also referred as a {\em left message} ---
is a triplet of real numbers $(\Pi^{\tt u}_{v\rightarrow
c},\Pi^{\tt s}_{v\rightarrow c},\Pi^{*}_{v\rightarrow c}$), and the
message passed from function vertex $\Gamma_c$ to variable vertex
$x_v$ --- also referred to as a {\em right message} ---
is a real number $\eta_{c\rightarrow v}\in [0,1]$. These
messages are updated respectively according to the following
equations.
\begin{eqnarray}
\label{equ:spia_1}
    \Pi^{\tt u}_{v\rightarrow c}& := & \left(1-\gamma\prod_{b\in C^{\tt u}_{c}(v)}(1-\eta_{b\rightarrow v})\right)
\prod_{b\in C^{\tt s}_{c}(v)}(1-\eta_{b\rightarrow v})\\
\label{equ:spia_2}
    \Pi^{\tt s}_{v\rightarrow c}& := &
\left(1-\prod_{b\in C^{\tt s}_{c}(v)}(1-\eta_{b\rightarrow
v})\right)
\prod_{b\in C^{\tt u}_{c}(v)}(1-\eta_{b\rightarrow v})\\
\label{equ:spia_3}
    \Pi^{*}_{v\rightarrow c}& :=& \prod_{b\in C^{\tt s}_{c}(v)}
(1-\eta_{b\rightarrow v})\prod_{b\in C^{\tt
u}_{c}(v)}(1-\eta_{b\rightarrow v})\\
\label{equ:spai}
    \eta_{c\rightarrow v}& := & \prod_{u\in V(c)\setminus \{v\}}
\frac{\Pi^{\tt u}_{u\rightarrow c}}{\Pi^{\tt u}_{u\rightarrow
c}+\Pi^{\tt s}_{u\rightarrow c}+\Pi^{*}_{u\rightarrow c}}.
\end{eqnarray}

The initialization of SP messages is usually random, and
message-passing schedule is typically similar to the {\em flooding
schedule} \cite{Frank:factor} in BP message passing, namely, that
each iteration may be defined by all variable vertices passing
messages followed by all function vertices passing messages. We note
that throughout this paper all message-passing schedules are
restricted to the flooding schedule for convenience, where each
iteration is defined as first updating all ``left messages'' and
then updating all ``right messages'' \footnote{An iteration may also
include updating all summary messages after updating the right
messages; see the description of summary messages.}

Similar to BP, at the end of an iteration, SP may compute a
``summary message'' at each variable vertex. For any $v\in V$,
define $C^{\rm 1}(v):=\{b\in C(v):L_{v, b}=1\}$ and $C^{\rm
0}(v):=\{b\in C(v):L_{v, b}=0\}$, then the ``summary message'' at
$x_v$ is a triplet $(\zeta_{v}^1,\zeta_{v}^0,\zeta_{v}^*)$ of real
numbers, computed by

\begin{eqnarray}
\label{eq:sum_mssg_1}
\zeta_{v}^1& := & \left( 1-\gamma\prod_{b\in
C^{\rm 1}(v)}(1-\eta_{b\rightarrow v}) \right)
\prod_{b\in C^{\rm 0}(v)}(1-\eta_{b\rightarrow v})\\
\label{eq:sum_mssg_0}
\zeta_{v}^0& := & \left(1-\gamma\prod_{b\in
C^{\rm 0}(v)}(1-\eta_{b\rightarrow v}) \right)
\prod_{b\in C^{\rm 1}(v)}(1-\eta_{b\rightarrow v})\\
\label{eq:sum_mssg_*}
 \zeta_{v}^*& := &
\gamma\prod_{b\in C^{\rm 1}(v)}(1-\eta_{b\rightarrow v}) \prod_{b\in
C^{\rm 0}(v)}(1-\eta_{b\rightarrow v})
\end{eqnarray}
where summary message $(\zeta^1, \zeta^0, \zeta^*)$ is typically normalized
to a scaled version $({\zeta^1}^{\rm norm}, {\zeta^0}^{\rm norm}, {\zeta^*}^{\rm norm})$ such that
\[{\zeta^1}^{\rm norm}+ {\zeta^0}^{\rm norm}+ {\zeta^*}^{\rm norm}=1.\]

Equations (\ref{equ:spia_1}) to (\ref{eq:sum_mssg_*}) and the
normalization procedure after completely specify the message-update
rule of $\mbox{SP}(\gamma)$.

Usually, SP is applied in conjunction with a heuristic
``decimation'' procedure, which is carried out after SP converges or
after a certain number of SP iterations. In the
decimation procedure, the ``polarity'' $B(v):={\zeta_v^0}^{\rm
norm}-{\zeta_v^1}^{\rm norm}$ at each $v\in V$ is calculated, and
the most polarized variable (namely, one
 having the highest $|B(v)|$)
is fixed to $0$ or $1$ according to
the sign of $B(v)$: $x_v$ is set to $0$ if
$B(v)>0$, and to $1$ otherwise.
The $k$-SAT problem is
then simplified and SP is applied again. This process iterates until
the reduced problem is simple enough for a local search algorithm.

When $\gamma=1$, it is shown in \cite{Braunstein} and \cite{SP:newlook}
that
the passed messages as in (\ref{equ:spia_1}) through (\ref{equ:spai})
can be interpreted probabilistically, namely,
$\eta_{c\rightarrow v}$ may be interpreted as
the probability that a ``warning'' symbol is sent
from $\Gamma_c$ to $x_v$, and $\Pi_{v\rightarrow c}^{\tt u}$,
$\Pi_{v\rightarrow c}^{\tt s}$ and $\Pi_{v\rightarrow c}^*$ are respectively
the probabilities that $x_v$ sends to $\Gamma_c$ symbol $\bar{L}_{v,c}$,
symbol $L_{v,c}$ and symbol $*$.

When $\gamma<1$, $\mbox{SP}(\gamma)$ 
however can no longer be interpreted
probabilistically. We now present a slightly modified formulation of
$\mbox{SP}(\gamma)$, referred to as $\mbox{SP}^*(\gamma)$, which is
completely equivalent to $\mbox{SP}(\gamma)$ defined in
\cite{SP:newlook}, and which will be shown in a later section to
have a natural probabilistic interpretation.


In $\mbox{SP}^*(\gamma)$, the left message $(\Pi^{\tt
u}_{v\rightarrow c},\Pi^{\tt s}_{v\rightarrow c},
\Pi^{*}_{v\rightarrow c})$ passed from variable vertex $x_v$ to
function vertex $\Gamma_c$ is modified to the equations given in
(\ref{equ:spiam_1}) to (\ref{equ:spiam_3}), and the right message
$\eta_{c\rightarrow v}$ passed from function vertex $\Gamma_c$ to
variable vertex $x_v$ and the summary message $(\zeta_v^1,
\zeta_v^0, \zeta_v^*)$ at variable $x_v$ stay unchanged.

\begin{eqnarray}
\label{equ:spiam_1}
    \Pi^{\tt u}_{v\rightarrow c}& := & \left(1-\gamma\prod_{b\in C^{\tt u}_{c}(v)}(1-\eta_{b\rightarrow v})\right)
\prod_{b\in C^{\tt s}_{c}(v)}(1-\eta_{b\rightarrow v})\\
\label{equ:spiam_2}
    \Pi^{\tt s}_{v\rightarrow c}& := &
\left(1-\gamma\prod_{b\in C^{\tt s}_{c}(v)}(1-\eta_{b\rightarrow
v})\right)
\prod_{b\in C^{\tt u}_{c}(v)}(1-\eta_{b\rightarrow v})\\
\label{equ:spiam_3}
    \Pi^{*}_{v\rightarrow c}& :=& \gamma\prod_{b\in C^{\tt s}_{c}(v)}
(1-\eta_{b\rightarrow v})\prod_{b\in C^{\tt
u}_{c}(v)}(1-\eta_{b\rightarrow v})
\end{eqnarray}

The following lemma shows that $\mbox{SP}(\gamma)$ and $\mbox{
SP}^*(\gamma)$ are equivalent.

\begin{lem}
\label{lem:sp*gamma} For the same initialization of
$\{\eta_{c\rightarrow v}: \forall (v-c)\}$, at any given iteration,
$\mbox{SP}^*(\gamma)$ and $\mbox{SP}(\gamma)$ give rise to identical
results in $\eta_{c\rightarrow v}$ for every $(v-c)$, and in
$(\zeta_{v}^1,\zeta_v^0, \zeta_v^*)$ for every $v\in V$.
\end{lem}
\begin{proof}
The lemma follows from that in the computation of
$\eta_{c\rightarrow v}$ and hence of $(\zeta_v^1, \zeta_v^0,
\zeta_v^*)$, $\Pi_{v\rightarrow c}^{\tt s}$ and $\Pi_{v\rightarrow
c}^*$ always appear together in the form of $\Pi_{v\rightarrow
c}^{\tt s}+\Pi_{v\rightarrow c}^*$. But it is easy to see that in
$\mbox{SP}(\gamma)$ and in $\mbox{SP}^*(\gamma)$, $\Pi_{v\rightarrow
c}^{\tt s}+\Pi_{v\rightarrow c}^*$ has the same parametric form,
both equal to $\prod_{b\in C^{\tt u}_c(v)}(1-\eta_{b\rightarrow
v})$.
\end{proof}


We conclude this subsection by remarking that it is possible to
verify that all results concerning $\mbox{SP}(\gamma)$ in
\cite{SP:newlook} hold for $\mbox{SP}^*(\gamma)$
\footnote{Specifically, we note that BP on the MRF formulated in
\cite{SP:newlook} will also reduce to  $\mbox{SP}^*(\gamma)$. We
leave this for the interested readers to verify.}. As such, in the
rest of this paper, $\mbox{SP}^*(\gamma)$ rather than
$\mbox{SP}(\gamma)$ will be taken as the weighted SP for $k$-SAT
problems.

\subsection{Survey Propagation for $q$-COL Problems}

Similar to SP developed for $k$-SAT problems, in $q$-COL problems,
SP passes messages between the variable
vertices and the function (constraint) vertices in the factor-graph
representation of the problem. Some notable differences however exist.

First, weighted SP has not been developed for $q$-COL problems to
date, and it is not even clear whether such algorithm family, if
existing, can be developed in a similar manner as that for $k$-SAT
in \cite{SP:newlook}, namely, via reducing the BP algorithm derived
from a properly defined MRF.  Answering this question in a later
section, we here therefore only review the original version of SP
applied to $3$-COL problems following the formulation in
\cite{Braunstein:3col}, which is analogous to $\mbox{SP}(1)$, or the
non-weighted SP, in the context of $k$-SAT.

Second, the SP messages for $q$-COL problems can be expressed more compactly,
due to a specific nature of the problem, on which we now elaborate.

For $q$-COL problems, each constraint vertex has degree $2$. This
allows the combination of the message passed from variable $x_u$ to
a neighboring  constraint, say $\Gamma_c$, with the message passed
from constraint $\Gamma_c$ to the other neighbor, say $x_v$, of
$\Gamma_c$. As a consequence, $\Gamma_c$ may be suppressed in the
factor graph, and messages are directly passed between variable
vertices that are distance $2$ apart
\footnote{Still implementing the flooding schedule, the SP message-update
rule for $3$-COL problems however suppresses the passing of one set of messages
(say, for example, the right messages) by including the computation of these
messages in updating the other set of messages.}
(or equivalently, messages are
passed on graph $(\Delta, \Xi)$). Following \cite{Braunstein:3col},
a compact version of SP message-passing rule for $3$-COL problems is
given as follows, where the message passed from variable $x_u$ to
variable $x_v$ is
 a quadruplet of real numbers $(\eta^1_{u \rightarrow v},\eta^2_{u
\rightarrow v},\eta^3_{u\rightarrow v},\eta^*_{u\rightarrow v})$.
For $i=1, 2, 3$,
\begin{equation}
\label{equ:sp_3_col} \eta^i_{u\rightarrow v} :=
\frac{\prod\limits_{w\in N(u)\setminus \{v\}}(1-\eta^i_{w\rightarrow
u})-\sum\limits_{j\neq i}\prod\limits_{w\in
N(u)\setminus\{v\}}(\eta^*_{w\rightarrow u} + \eta^j_{w\rightarrow
u})+\prod\limits_{w\in N(u)\setminus\{v\}}\eta^*_{w\rightarrow
u}}{\sum\limits_{j=1,2,3}\prod\limits_{w\in
N(u)\setminus\{v\}}(1-\eta^j_{w\rightarrow
u})-\sum\limits_{j=1,2,3}\prod\limits_{w\in
N(u)\setminus\{v\}}(\eta^*_{w\rightarrow u}+\eta^j_{w\rightarrow
u})+\prod\limits_{w\in N(u)\setminus\{v\}}\eta^*_{w\rightarrow u}}
\end{equation}
where $N(u)$ is the set $\{v:v\in V,
\{u, v\}\in \Xi\}$, namely, the set of neighboring vertices of
vertex $u$ on graph $\{\Delta,\Xi\}$; and
\begin{equation}
\label{equ:sp_3_col_star} \eta^*_{u\rightarrow v}:=1-\sum_{j=1,2,3}\eta^j_{u\rightarrow v}.
\end{equation}

For $3$-COL problems, the ``summary message'' computed at each
variable vertex $x_v$ is a quadruplet of real numbers, denoted by
$(\zeta^1_v,\zeta^2_v,\zeta^3_v,\zeta^*_v)$, where for $i=1, 2, 3,$
\begin{equation*}
\zeta^i_v:= \frac{\prod\limits_{u\in N(v)}(1-\eta^i_{u\rightarrow
v})-\sum\limits_{j\neq i}\prod\limits_{u\in
N(v)}(\eta^*_{u\rightarrow v} + \eta^j_{u\rightarrow
v})+\prod\limits_{u\in N(v)}\eta^*_{u\rightarrow
v}}{\sum\limits_{j=1,2,3}\prod\limits_{u\in
N(v)}(1-\eta^j_{u\rightarrow
v})-\sum\limits_{j=1,2,3}\prod\limits_{u\in
N(v)}(\eta^*_{u\rightarrow v}+\eta^j_{u\rightarrow
v})+\prod\limits_{u\in N(v)}\eta^*_{u\rightarrow v}}
\end{equation*}
and
\begin{equation*}
\zeta^*_v := 1-\sum\limits_{j=1,2,3}\zeta^j_v.
\end{equation*}

Similar to that for $k$-SAT problems, the summary message for a
$3$-COL problem at variable $x_v$ may indicate the ``bias'' of
variable $x_v$ to each letter in $\{1, 2, 3, *\}$. In the decimation
procedure for $3$-COL problems -- carried out in a similar way to
that for $k$-SAT problems, a variable is fixed to a color
$i\in\{1,2,3\}$ if it is highly biased to that color. The reader is
referred to \cite{Braunstein:3col} for a detailed account of a
heuristic decimation rule used in solving $3$-COL problems using SP.

We note that this paper primarily focuses on SP update equations, where the
decimation aspect of SP is largely ignored.

\section{SP as Probabilistic Token Passing}
\label{sec:token}

To date, SP algorithms have been applied to various other CSPs,
for example,
in coding for Blackwell channels \cite{Yu:SP_Blackwell}, in quantization of
Bernoulli sources \cite{Wainwright:SP_compression},
and in solving graph coloring problems \cite{Braunstein:3col}, etc.. However,
a general formulation of SP, particularly that of
weighted SP, for solving arbitrary non-binary CSPs, has been
largely missing. Specifically, we note the following milestones in
the formulation of SP algorithms.
\begin{itemize}
\item The work of \cite{Braunstein} presents non-weighted version
of SP formulas for general CSPs beyond those involving only binary variables.
However, the exposition of \cite{Braunstein}
uses the language of statistical physics, rather remote
to the engineering community, and a cleaner and more friendly
formulation of SP, and particularly of weighted SP, is desirable
for general problems.
\item The work of \cite{SP:newlook} presents
weighted SP for $k$-SAT problems, in which weighted SP is treated as
a special case of BP in a properly defined MRF. This treatment of SP
and the corresponding principle for developing weighted SP are conceivably
applicable to all binary CSPs.
However, it has remained open, prior to this work, whether such an approach
to understanding and developing weighted SP is applicable to arbitrary
non-binary CSPs.
\end{itemize}

The line of development in this section is summarized below.

We will first present an understanding of non-weighted SP for
arbitrary CSPs (namely, that formulated in \cite{Braunstein}) in
terms of ``probabilistic token passing (PTP)''. Although similar
understanding has been previously reported in various contexts, we
here stress the role of extending the variable alphabet in SP
algorithms, and explicitly point out that the alphabet extension is {\em not}
to simply include an extra joker symbol, but to {\em replace} the
variable alphabet with its {\em power set} (excluding the empty-set
element). To make the PTP procedure more intuitively sensible, prior
to defining PTP, we will introduce a precursor of PTP, which we call
``deterministic token passing'' (DTP).

After introducing PTP, we then show that the probabilistic interpretation of non-weighted
SP in terms of PTP makes it naturally generalizable to a weighted
version, which we call weighted PTP. For a brief preview, the
generalization of PTP to weighted PTP essentially involves
generalizing a {\em functional dependency} in PTP message-update
rule to a {\em probabilistic dependency}. Interestingly as we will
show, it turns out that for $k$-SAT problems, weighted PTP precisely
coincides with weighted SP of \cite{SP:newlook}.  This should
convincingly demonstrate that weighted PTP is a generalization of
weighted SP for arbitrary CSPs.

The outline of this section is given as follows. Subsection
\ref{subsec:alpha} introduces the notion of alphabet extension and related
concepts. Subsection \ref{subsec:dtp} defines DTP as a precursor of PTP.
In Subsection \ref{subsec:ptp}, we introduce PTP. In Subsection \ref{subsec:sp_as_ptp}, we show that PTP is equivalent to SP, using $3$-COL problem as
an example. In Subsection \ref{subsec:w_ptp}, we introduce weighted PTP.
In Subsection \ref{subsec:w_ptp_gen_w_sp}, we show that weighted PTP generalize
weighted SP using $k$-SAT problems as an example.

\subsection{Alphabet Extension}
\label{subsec:alpha}

For a given CSP with variable alphabet $\chi$, we define the {\em
extended alphabet} $\chi^*$ as the power set of $\chi$  excluding
the empty set $\emptyset$. That is, $\chi^*=\{t: t\subseteq \chi,
t\neq \emptyset\}$). The extended alphabet $\chi^*$ of $k$-SAT
problems is then the set $\{ \{0\},\{1\},\{0,1\}\}$. For $3$-COL
problems, $\chi^*$ is the set
$\{\{1\},\{2\},\{3\},\{1,2\},\{1,3\},\{2,3\},\{1,2,3\}\}$. Each
element $t$ of $\chi^*$ will be written as a string -- in bold font
-- containing the elements of $t$. For example, we may write
$\{1,2\}$ as ${\bf 12}$, $\{1, 2, 3\}$ as ${\bf 123}$ and $\{1\}$
simply as ${\bf 1}$.

Given any subset $U\subseteq V$, a $\chi^*$-assignment $y_U$ on $U$
is referred to as a {\em rectangle}
on $U$. The set of all rectangles on $U$ is denoted by
$(\chi^*)^U$.
Given rectangle $y_U\in (\chi^*)^U$, for every $v\in U$, $y_{U:\{v\}}$, or
simply written as $y_v$ --- following an earlier convention of this paper ---
is referred to as the $v$-{\em side} of $y_U$. Apparently, rectangle $y_U$ has
$|U|$ sides, and may also be written as the concatenation of all its sides,
namely, as $\langle y_v\rangle_{v\in U}$.

For any $v\in V$, an elementary $\chi^*$-assignment
$t_v\in (\chi^*)^{\{v\}}$
will be referred to as a {\em token} on $v$. Using
this nomenclature, the $v$-side of any rectangle is a token on $v$.
We note that a token $t_v$ may be interpreted as a set of elementary
$\chi$-assignments on $\{v\}$, which is in fact the set of all
elementary $\chi$-assignments on $\{v\}$ that assign $v$ a value in set
$t_v(v)\subseteq \chi$. For example, suppose that $\chi:=\{1, 2,
3\}$, then token ${\bf 12}_v$ may be identified with the set $\{1_v,
2_v\}$ of elementary $\chi$-assignments on $\{v\}$.

It is worth noting that when a token $t_v$ is identified with a set
of elementary $\chi$-assignments on $v$, a rectangle $\langle t_v
\rangle_{v\in U}$ may be identified with the {\em Cartesian product}
of all its sides. For example, rectangle $({\bf 12}_v, {\bf 23}_u)$
may be interpreted as the following set of $\chi$-assignments on
$\{v, u\}$: $\{(1_v, 2_u), (1_v, 3_u), (2_v, 2_u), (2_v, 3_u)\}$.
Under this interpretation, we will also make frequent uses of the
Cartesian product notation, writing rectangle $({\bf 12}_v, {\bf
23}_u)$ as ${\bf 12}_v \times {\bf 23}_u$, and rectangle $\langle
t_v \rangle_{v\in U}$ as $\prod_{v\in U}t_v$. We note that this
interpretation is in fact the reason for which we choose the
terminologies ``rectangle'' and ``side''.

For simplicity, from here on, we shall reserve the term
``assignment'' to referring to a $\chi$-assignment only, and a
$\chi^*$-assignment will be referred to as a ``rectangle'', ``side''
or ``token''.

We say that an assignment $x_U$ on $U$ is
 {\em contained} in rectangle $y_U$ if $x_{U:\{v\}}(v) \in
y_{U:\{v\}}(v)$ for every $v\in U$. For example, assignment $(1_v, 2_u)$
is contained in rectangle $({\bf 13}_v, {\bf 23}_u)$
We will use $x_U\in y_U$ to
denote this containedness relationship, since this notation is
precise when the rectangle $y_U$ is interpreted as a {\em set} of
assignments on $U$.

Given a CSP and a $(v-c)$ pair, we define function ${\tt F}_c^v:
\left(\chi^*\right)^{V(c)\setminus \{v\}} \rightarrow \left( \chi^*
\right)^{\{v\}}$ as follows: for every rectangle $\prod_{u\in
V(c)\setminus \{v\}} t_u$ on $V(c)\setminus \{v\}$,
\begin{equation*}
{\tt F}^v_c\left(
\prod_{u\in V(c)\setminus \{v\}} t_u
\right):=
\left(
\left(
\chi^{\{v\}}\times \prod\limits_{u\in V(c)\setminus\{v\}}t_u
\right)
\cap \Gamma_c\right)_{:\{v\}}.
\end{equation*}
We often write ${\tt F}_c^v$ in short as ${\tt F}_c$ since the domain
and co-domain of the function may be recovered from the form of its argument.
Given rectangle
$\prod_{u\in V(c)\setminus \{v\}} t_u$ on $V(c)\setminus \{v\}$, we call
${\tt F}_c\left(
\prod_{u\in V(c)\setminus \{v\}} t_u
\right)$ the {\em forced token} by rectangle $\prod_{u\in V(c)\setminus \{v\}} t_u$ via constraint $\Gamma_c$.
It is easy to verify that the forced  token
${\tt F}_c\left(\prod\limits_{u \in V(c)\setminus \{v\}} t_u\right)$
is simply the set of all (elementary) assignments on $\{v\}$ which, when concatenated
with an assignment on
$V(c)\setminus \{v\}$ contained in rectangle
$\prod\limits_{u \in V(c)\setminus \{v\}} t_u$,
make local constraint
$\Gamma_c$ satisfied. We now give some examples
using the toy $3$-SAT problem shown in Fig. \ref{fig:ksat_fg} to
illustrate this
definition. Consider
constraint $\Gamma_a$, if rectangle $t_{\{1,2\}}$ on $\{1, 2\}$
is defined as $({\bf 1}_{1}, {\bf 01}_{2})$,
then forced token ${\tt F}_a(t_{\{1, 2\}})={\bf 01}_4$,
since when assigning variable $x_4$ either value $0$ or $1$, it is possible to
find an assignment
of variables $x_1$ and $x_2$ in rectangle $t_{\{1, 2\}}$
that makes $\Gamma_a$ satisfied;
on the other hand, if $t_{\{1, 2\}}=({\bf 0}_1, {\bf 1}_2)$, then
forced token
${\tt F}_a(t_{\{1, 2\}})={\bf 0}_4$, since rectangle
$t_{\{1,2\}}$ contains a single
assignment of $x_1$ and $x_2$ (namely $(0_1, 1_2)$),
and the only assignment of $x_4$ that
will make constraint $\Gamma_a$ satisfied is the one
assigning $0$ to $x_4$, namely $0_4$.

A ``monotonicity property'' of function ${\tt F}_c$, stated in the following
lemma, follows immediately from the definition of the function.

\begin{lem}\label{lem:forcedToken_monotonic}
Suppose that $x_v$ and $\Gamma_c$ are a pair of neighboring variable and
constraint vertices in the factor graph, and that $y_{V(c)\setminus \{v\}}$
and $y'_{V(c)\setminus \{v\}}$ are two rectangles on $V(c)\setminus \{v\}$.
Then 
$y_{V(c)\setminus \{v\}} \subset y'_{V(c)\setminus \{v\}}$ impliese that
$
{\tt F}_c\left(
y_{V(c)\setminus \{v\}}
\right)
\subseteq
{\tt F}_c\left(
 y'_{V(c)\setminus {\{v\}}}
\right).
$
\end{lem}

\subsection{Deterministic Token Passing (DTP)}
\label{subsec:dtp}

As we will introduce --- for arbitrary CSPs --- a probabilistic
interpretation of non-weighted SP (namely, PTP) and generalize it to
a weighted version (namely, weighted PTP), in this subsection, we
first introduce an algorithmic procedure, which we call {\em
deterministic token passing} or DTP. We note that the purpose of
introducing DTP is to provide an easier access to PTP, a procedure
to be introduced in the next subsection.


In DTP, messages are tokens passed along
the edges of the factor graph representing the CSP of interest.
Specifically, the token passed from and to each
variable $x_v$ is a token on $v$, or equivalently, a set of (elementary)
assignments on $\{v\}$.
For any pair of
neighboring vertices $x_v$ and $\Gamma_c$ on the factor graph, the
token, or left message,
$t_{v\rightarrow c}$ passed from variable $x_v$ to constraint
$\Gamma_c$ depends on all incoming tokens (right messages)
passed to $x_v$ except
that  passed from $\Gamma_c$.
 Similarly, the token, or right message, $t_{c\rightarrow
v}$ passed from  constraint $\Gamma_c$ to variable $x_v$ depends on
all incoming tokens (left messages)
passed to $\Gamma_c$ except that passed from
$x_v$. Each iteration of token passing in DTP is defined by
every variable passing a token on each of its edges followed by every
constraint passing a token on each of its edges.
Within any iteration, the token-passing rule of DTP is given as follows.
\begin{eqnarray}
\label{equ:tokenvtoc} t_{v\rightarrow c} &:= &\bigcap_{b\in
C(v)\setminus\{c\}}t_{b\rightarrow v}\\
\label{equ:tokenctov} t_{c\rightarrow v} &:= &{\tt
F}_c\left(\prod_{u\in V(c)\setminus\{v\}}t_{u\rightarrow c}\right).
\end{eqnarray}
That is, the token passed from a variable is the {\em intersection} of its
incoming tokens from the upstream, whereas the token passed from
a constraint is the forced token via the constraint by the rectangle formed
by the upstream incoming tokens as sides.

It is intuitive to illuminate this message-passing rule using the following
analogy. We may view the token sent from a variable as the ``intention''
of the variable, indicating the possible values that the variable
intends to take. On the other hand, we may
view the token sent from a constraint as the
``command'' from the constraint, indicating the possible values that the constraint allows the destination variable to take.
 If $a$ is an intention and $b$ is a command,
where both are tokens on the same coordinate, then the relationship
$a\subseteq b$ may be viewed as that ``intention $a$ obeys command $b$''.
Under this perspective, the token sent from a variable is the ``maximal''
intention of the variable that obeys all incoming commands from the upstream
constraints; on the other hand, the token sent from a constraint
is the ``maximal'' command that is ``compatible'' with all incoming
intentions from the upstream variables. Here ``maximality'' is in the sense
of maximizing the cardinality of
the subset of assignments, and ``compatibility'' is in the sense of
satisfying the local constraint.

Examples of token passing for a 3-COL problem are illustrated in
Fig. \ref{fig:token}.

\begin{figure}[htb]
  \begin{center}
    \begin{minipage}[h]{3cm}
      \begin{center}
        \scalebox{0.5}{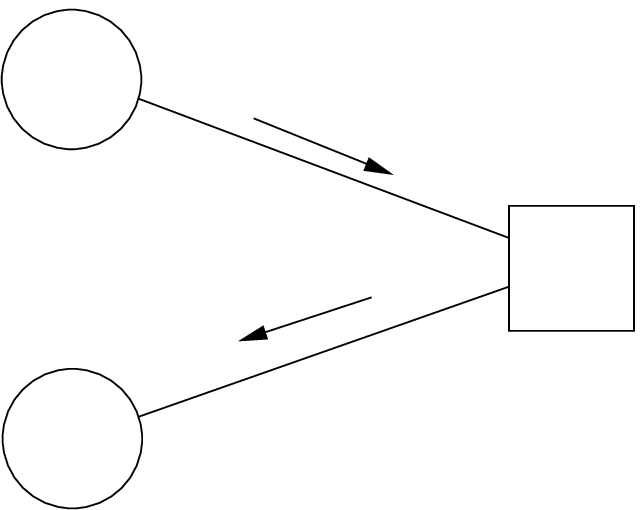}
        \scriptsize{(a)}
      \end{center}
    \end{minipage}
    \hspace{1.5cm}
    \begin{minipage}[h]{3cm}
      \begin{center}
        \scalebox{0.5}{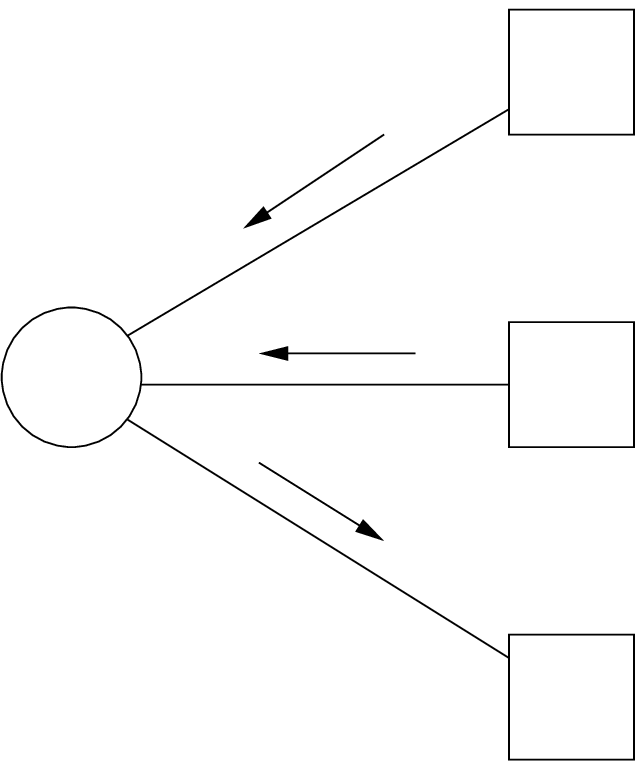}
        \scriptsize{(b)}
      \end{center}
    \end{minipage}
   \caption{Examples of deterministic token passing for a 3-COL problem.
   (a) Token $t_{c\rightarrow v}$ passed from constraint $\Gamma_c$
   to variable $x_v$. (b) Token $t_{v\rightarrow c}$ passed from
   variable $x_v$ to constraint $\Gamma_c$.}
   \label{fig:token}
  \end{center}
\end{figure}

A summary message or ``summary token''
at variable vertex $x_v$ may be computed, according
to the rule in (\ref{equ:sum_token}) for each
$v\in V$ at any iteration after the all constraint vertices have passed
tokens.
\begin{equation}
\label{equ:sum_token} t_v: =\bigcap_{b\in C(v)}t_{b\rightarrow v}.
\end{equation}

Using the ``intention/command'' analogy, the summary token at a variable
is the ``maximal'' intention of the variable that obeys the
incoming commands from {\em all} directions.

Some caution is needed on the well-definedness of the updating rule
of passed tokens and summary tokens.  That is, in
(\ref{equ:tokenvtoc}), (\ref{equ:tokenctov}) and
(\ref{equ:sum_token}) the right-hand side can be equal to the empty
set $\emptyset$, which is not a well-defined token. Whenever in an
iteration a not-well-defined token (i.e., the empty set) arises from
the updating rule, we may
force DTP to terminate. --- As we will see later in the ``random'' version of
DTP (i.e., PTP and weighted PTP), we will eventually condition on the case in 
which these events do not happen.

At any iteration, one may read out the summary tokens at all variable vertices
and form a rectangle on $V$ using these tokens as its sides. It is clear that at any
given iteration, the resulting rectangle formed by the summary tokens depends
on the initialization of DTP. 

Although our primary purpose of introducing DTP is to make smoother
the transition to understanding PTP, in Appendix A, we present some
elementary results concerning the dynamics of DTP. We note that those
results will also be used
to derive some insights on the dynamics of PTP --- an algorithmic
procedure that we introduce next as a simple formulation of SP.

\subsection{Probabilistic Token Passing (PTP)}
\label{subsec:ptp}

We now introduce the ``probabilistic token passing'' (or PTP)
procedure. The key distinction between PTP and DTP is that on each
edge and along each direction, PTP passes a {\em random} token and
the messages being updated in PTP are the {\em distributions} of the
random tokens.

Specifically, PTP message-update rule can be constructed by
considering the following mechanism of passing random tokens.
\begin{enumerate}
\item On each edge connecting variable $x_v$ and constraint $\Gamma_c$ in
the factor graph, the token $t_{v\rightarrow c}$ passed to
constraint $\Gamma_c$ and the token $t_{c\rightarrow v}$ passed to
variable $x_v$ are both {\em random variables}, distributed over
$\left(\chi^*\right)^{\{v\}}$.
\item For any given vertex in the factor graph,
all of its incoming random tokens are assumed to be independent.
\item For any given vertex in the factor graph,
the outgoing random token sent along any edge is a function of all
the incoming random tokens from the upstream, where the functional
dependency is precisely that specified in DTP, namely,
(\ref{equ:tokenvtoc}) or (\ref{equ:tokenctov}), depending on whether
the vertex is a variable vertex or a function (constraint) vertex.
\item The summary (random) token $t_v$ at each variable vertex $x_v$
is a function of all incoming random tokens, where the functional dependency
is precisely that specified in DTP, namely, (\ref{equ:sum_token}).
\end{enumerate}

Building on this mechanism, we will then define each PTP (passed or
summary) message as the distribution of the corresponding random
token {\em conditioned} on that the token is well defined (namely,
not equal to the empty set). We note that such a ``conditioning''
merely involves a normalization (namely, scaling) of each message so
that it sums to $1$ over all valid tokens. We will use
$\lambda_{v\rightarrow c}$ to denote the message sent from $x_v$ to
$\Gamma_c$ --- also referred to as a left message,
$\rho_{c\rightarrow v}$ to denote the message sent from $\Gamma_c$
to $x_v$ --- also referred to as a right message, and $\mu_v$ to
denote the summary message at variable vertex $x_v$. It is then
straight-forward to derive the message-update rule of PTP as
follows, where the superscript ``${\tt norm}$'' on a message
indicates that the message has been normalized.

\begin{center}
\shadowbox{
\begin{minipage}{6.3in}
{\bf \large PTP Message-Update Rule}\\
\begin{eqnarray}
\label{equ:ptpvtoc} \lambda_{v\rightarrow c}(t_{v\rightarrow c})& :=
& \sum_{ \langle t_{b\rightarrow v} \rangle_{b\in
C(v)\setminus\{c\}} } \left[ t_{v\rightarrow c}=\bigcap\limits_{b\in
C(v)\setminus\{c\}}t_{b\rightarrow v} \right] \prod_{b\in
C(v)\setminus\{c\}}\rho^{\rm norm}_{b\rightarrow v}(t_{b\rightarrow v})\\
\label{equ:ptpctov} \rho_{c\rightarrow v}(t_{c\rightarrow
v})&: =&\sum_{ \langle t_{u\rightarrow c} \rangle_{u\in
V(c)\setminus\{v\}} } \left[t_{c\rightarrow v}={\tt F}_c\left(
\prod\limits_{u\in V(c)\setminus \{v\}} t_{u\rightarrow c} \right)
\right] \prod_{u\in
V(c)\setminus\{v\}}\lambda^{\rm norm}_{u\rightarrow c}(t_{u\rightarrow c})\\
\label{equ:ptp_summary} \mu_v(t_v) & := &\sum_{ \langle
t_{c\rightarrow v} \rangle_{c\in C(v)} } \left[
t_v=\bigcap\limits_{c\in C(v)}t_{c\rightarrow v} \right] \prod_{c\in
C(v)}\rho^{\rm norm}_{c\rightarrow v}(t_{c\rightarrow v}),
\end{eqnarray}
\\
and the normalized messages are defined as
\begin{eqnarray}
\lambda_{v\rightarrow c}^{\rm norm}(t_{v\rightarrow c})
&:= &
 \lambda_{v\rightarrow c}(t_{v\rightarrow c})
/\sum\limits_{t\in \left(\chi^*\right)^{\{v\}}}
\lambda_{v\rightarrow c}(t)\\
\rho_{c\rightarrow v}^{\rm norm}(t_{c\rightarrow v})
&:= &
 \rho_{c\rightarrow v}(t_{c\rightarrow v})
/\sum\limits_{t\in \left(\chi^*\right)^{\{v\}}}
\rho_{c\rightarrow v}(t)\\
\mu_{v}^{\rm norm}(t_{v})
&:= &
 \mu_{v}(t_{v})
/\sum\limits_{t\in \left(\chi^*\right)^{\{v\}}} \mu_{v}(t).
\end{eqnarray}
\end{minipage}
}
\end{center}

We note that the update of messages in each PTP iteration
is proceeded by first computing the un-normalized messages and then computing
their normalized version.

\subsection{SP as PTP}
\label{subsec:sp_as_ptp}

We now show that SP is precisely PTP using the example of $3$-COL
problems. Here we note that it is possible (and entails little
additional difficulty) to show the
equivalence between PTP
and the {\em general} formulation of
non-weighted SP \cite{Braunstein} for arbitrary
CSPs. However, as we feel it unnecessary to distract the readers
with the additional statistical physics terminologies presented in
\cite{Braunstein}, we choose not to repeat the exposition of SP in
\cite{Braunstein} and only show that SP is PTP for the special case of
$3$-COL problems.

In the factor graph representing a $3$-COL problem, noting that each
constraint vertex has degree 2, we will make a slight abuse of
notation: for any $(v-c)$ pair, we will use $V(c)\setminus \{v\}$ to
also denote the {\em index} of the unique other variable vertex
(besides $x_v$) connecting to $\Gamma_c$, although $V(c)\setminus
\{v\}$ originally refers to the singleton set containing that index.
Whether $V(c)\setminus\{v\}$ should be treated as the index of a
variable or as the singleton set containing the index should be
clear from the context.

For notational simplicity, from here on, for every element in the
token set $\left(\chi^*\right)^{\{v\}}$, when no ambiguity is
resulted, we will suppress the subscript indicating the coordinate
of the element. For example, we will write 
${\bf 12}_v$ as ${\bf 12}$, when the subscript can be recovered from
the context. Additionally, we will use $i,j,$ and $k$
to denote the three distinct colors $1,2,$ and $3$ in the $3$-COL
problem, so that token ${\bf i}$ can refer to any token that is a
singleton set, token ${\bf ij}$ can refer to any token that contains
a pair of assignments, and token ${\bf ijk}$ refers to the token
containing all three assignments.

Using these notations,
the PTP message-update rule for $3$-COL problems can be easily derived, which is presented in the following lemma.

\begin{lem}
For $3$-COL problems, the PTP message-update rule is:

    \begin{eqnarray}
    \label{equ:ptp_3col_l_i}
        \lambda_{v\rightarrow c}({\bf i})
& := & \prod_{b\in C(v)\setminus\{c\}}
        \left(\rho^{\rm norm}_{b\rightarrow v}({\bf ij})+\rho^{\rm norm}_{b\rightarrow v}({\bf ik})+
        \rho^{\rm norm}_{b\rightarrow v}({\bf ijk})\right)-\prod_{b\in
        C(v)\setminus\{c\}}\left(\rho^{\rm norm}_{b\rightarrow v}({\bf ij})+\rho^{\rm norm}_{b\rightarrow
        v}({\bf ijk})\right)   \nonumber \\
       &  &-\prod_{b\in C(v)\setminus\{c\}}\left(\rho^{\rm norm}_{b\rightarrow v}
        ({\bf ik})+\rho^{\rm norm}_{b\rightarrow v}({\bf ijk})\right)+\prod_{b\in
        C(v)\setminus\{c\}}\rho^{\rm norm}_{b\rightarrow v}({\bf ijk})\\
        \label{equ:ptp_3col_l_ij}
 \lambda_{v\rightarrow c}({\bf ij})
 & := &\prod_{b\in C(v)\setminus\{c\}}\left(\rho^{\rm norm}_{b\rightarrow
        v}({\bf ij})+\rho^{\rm norm}_{b\rightarrow v}({\bf ijk})\right)-\prod_{b\in
        C(v)\setminus\{c\}}\rho^{\rm norm}_{b\rightarrow v}({\bf ijk})\\
        \label{equ:ptp_3col_l_ijk}
\lambda_{v\rightarrow c}({\bf ijk}) & := &
    \prod_{b\in
        C(v)\setminus\{c\}}\rho^{\rm norm}_{b\rightarrow v}({\bf ijk})\\
        \label{equ:ptp_3col_r_ij}
\rho_{c\rightarrow v}({\bf ij}) & := & \lambda^{\rm
norm}_{V(c)\setminus\{v\}\rightarrow c}
({\bf k})\\
        \label{equ:ptp_3col_r_ijk}\rho_{c\rightarrow v}({\bf ijk})
&:=&\lambda^{\rm norm}_{V(c)\setminus\{v\}\rightarrow c}({\bf
ij})+\lambda^{\rm norm}_{V(c)\setminus\{v\}\rightarrow
        c}({\bf ik})+\lambda^{\rm norm}_{V(c)\setminus\{v\}\rightarrow c}({\bf jk})+\lambda^{\rm norm}_{V(c)\setminus\{v\}\rightarrow c}({\bf ijk})\\
        \label{equ:ptp_3col_s_i}\mu_v({\bf i}) & := &\prod_{c\in C(v)}\left(\rho^{\rm norm}_{c\rightarrow v}({\bf
ij})+\rho^{\rm norm}_{c\rightarrow
        v}({\bf ik})+\rho^{\rm norm}_{c\rightarrow v}({\bf ijk})\right)-\prod_{c\in
        C(v)}\left(\rho^{\rm norm}_{c\rightarrow v}({\bf ij})+\rho^{\rm norm}_{c\rightarrow
        v}({\bf ijk})\right)  \nonumber \\
       & &-\prod_{c\in C(v)}\left(\rho^{\rm norm}_{c\rightarrow
        v}({\bf ik})+\rho^{\rm norm}_{c\rightarrow v}({\bf ijk})\right)+\prod_{c\in
        C(v)}\rho^{\rm norm}_{c\rightarrow v}({\bf ijk})\\
        \label{equ:ptp_3col_s_ij}\mu_v({\bf ij})& := & \prod_{c\in C(v)}\left(\rho^{\rm norm}_{c\rightarrow v}({\bf ij})+\rho^{\rm norm}_{c\rightarrow
        v}({\bf ijk})\right)-\prod_{c\in C(v)}\rho^{\rm norm}_{c\rightarrow v}({\bf ijk})\\
        \label{equ:ptp_3col_s_ijk}\mu_v({\bf ijk})& := & \prod_{c\in C(v)}\rho^{\rm norm}_{c\rightarrow v}({\bf ijk}).
    \end{eqnarray}
\end{lem}


It is then possible to relate the PTP messages and the (non-weighted) SP messages for $3$-COL problems, and show their equivalence.

\begin{thm}
\label{thm:PTPisSP_3col}
    For $3$-COL problems, the correspondence between SP and PTP message-update
rules is
    \begin{eqnarray}
\label{equ:eqv_3col}
    \eta_{u\rightarrow v}^i
& \leftrightarrow &
\lambda^{\rm norm}_{u\rightarrow\{u,v\}}({\bf i}) \nonumber \\
        \eta_{u\rightarrow v}^*
& \leftrightarrow &
1-\sum_{{\bf i}={\bf 1, 2, 3}} \lambda^{\rm norm}_{u\rightarrow\{u,v\}}
({\bf i}) \nonumber\\
& &=
\lambda^{\rm norm}_{u\rightarrow
        \{u,v\}}({\bf ij})+\lambda^{\rm norm}_{u\rightarrow\{u,v\}}({\bf ik})+\lambda^{\rm norm}_{u\rightarrow
        \{u,v\}}({\bf jk})+\lambda^{\rm norm}_{u\rightarrow\{u,v\}}({\bf ijk})
\nonumber
\\
\eta_{u}^i & \leftrightarrow &\mu_{u}^{\rm norm}({\bf i}) \nonumber\\
\eta_{u}^* & \leftrightarrow & 1-\sum_{{\bf i}={\bf 1, 2, 3}} \mu^{\rm norm}_{u}({\bf i}) \nonumber\\
    \end{eqnarray}
\end{thm}

\begin{proof}
    First we will identify $c$ in the subscript of
    $\lambda^{\rm norm}_{u\rightarrow c}$ with
    $\{u, v\}$ in which $v$ indexes the destination vertex in the subscript
    of  $\eta_{u\rightarrow v}$.

    For any $c=\{u,v\}$, let $\alpha_{u,v}=\lambda^{\rm norm}_{u\rightarrow
c}({\bf ij})+ \lambda^{\rm norm}_{u\rightarrow
    c}({\bf ik})+\lambda^{\rm norm}_{u\rightarrow c}({\bf jk})+\lambda^{\rm
    norm}_{u\rightarrow c}({\bf ijk})$. When applying PTP update equations (\ref{equ:ptp_3col_r_ij}) and
    (\ref{equ:ptp_3col_r_ijk}) to equations (\ref{equ:ptp_3col_l_i})
    to (\ref{equ:ptp_3col_l_ijk}) and re-writing the
update rule in terms of left messages only, the un-normalized left
messages are updated as follows.
    \begin{eqnarray}
    \lambda_{u\rightarrow c}({\bf i})& = &\prod_{b\in C(u)\setminus
    \{c\}}\left(1-\lambda^{\rm norm}_{V(b)\setminus\{u\}\rightarrow b}({\bf i})\right)-
    \prod_{b\in C(u)\setminus \{c\}}\left(\lambda^{\rm norm}_{V(b)\setminus\{u\}
    \rightarrow b}({\bf j})+\alpha_{V(b)\setminus\{u\},u}\right)\\
    & &-\prod_{b\in C(u)\setminus \{c\}}\left(\lambda^{\rm norm}_{V(b)\setminus\{u\}
    \rightarrow b}({\bf k})+\alpha_{V(b)\setminus\{u\},u}\right)
    +\prod_{b\in C(u)\setminus \{c\}}
    \alpha_{V(b)\setminus\{u\},u}\nonumber\\
    \lambda_{u\rightarrow c}({\bf ij}) & = & \prod_{b\in
    C(u)\setminus\{c\}}\left(\lambda^{\rm norm}_{V(b)\setminus\{u\}\rightarrow
    b}({\bf k})+\alpha_{V(b)\setminus\{u\},u}\right)-\prod_{b\in
    C(u)\setminus\{c\}}\alpha_{V(b)\setminus\{u\},u}\\
    \lambda_{u\rightarrow c}({\bf ijk}) & = & \prod_{b\in
    C(u)\setminus\{c\}}\alpha_{V(b)\setminus\{u\},u}.
    \end{eqnarray}

    After normalization, we have
    \begin{eqnarray}
    \lambda^{\rm norm}_{u\rightarrow c}({\bf i})& = &\frac{1}{\beta}\cdot\left(\prod_{b\in C(u)\setminus
    \{c\}}\left(1-\lambda^{\rm norm}_{V(b)\setminus\{u\}\rightarrow b}({\bf i})\right)-
    \prod_{b\in C(u)\setminus \{c\}}\left(\lambda^{\rm norm}_{V(b)\setminus\{u\}
    \rightarrow b}({\bf j})+\alpha_{V(b)\setminus\{u\},u}\right)\right.\nonumber\\
    & &\left. -\prod_{b\in C(u)\setminus \{c\}}\left(\lambda^{\rm norm}_{V(b)\setminus\{u\}
    \rightarrow b}({\bf k})+\alpha_{V(b)\setminus\{u\},u}\right)
    +\prod_{b\in C(u)\setminus \{c\}}
    \alpha_{V(b)\setminus\{u\},u}\right)\\
    \lambda^{\rm norm}_{u\rightarrow c}({\bf ij}) & = &\frac{1}{\beta}\cdot\left(\prod_{b\in
    C(u)\setminus\{c\}}\left(\lambda^{\rm norm}_{V(b)\setminus\{u\}\rightarrow
    b}({\bf k})+\alpha_{V(b)\setminus\{u\},u}\right)-\prod_{b\in
    C(u)\setminus\{c\}}\alpha_{V(b)\setminus\{u\},u}\right)\\
    \lambda^{\rm norm}_{u\rightarrow c}({\bf ijk}) & = &\frac{1}{\beta}\cdot\prod_{b\in
    C(u)\setminus\{c\}}\alpha_{V(b)\setminus\{u\},u},
    \end{eqnarray}
    where $\beta:=\sum_{t\in \left( \chi^* \right)^{\{u\}}}
\lambda_{u\rightarrow c}(t)$.

It is easy to see that
\begin{equation*}
\begin{split}
\beta & =
\sum_{{\bf i}={\bf 1,2,3}}\prod_{b\in C(u)\setminus\{c\}}\left(1-\lambda^{\rm norm}_{V(b)
    \setminus\{u\}\rightarrow b}({\bf i)}\right)
-\sum_{{\bf i}={\bf 1,2,3}}
    \prod_{b\in C(u)\setminus\{c\}}\left(\lambda^{\rm
    norm}_{V(b)\setminus\{u\}\rightarrow b}({\bf
    i})+\alpha_{V(b)\setminus\{u\},u}\right)\\
&+\prod_{b\in
    C(u)\setminus\{c\}}\alpha_{V(b)\setminus\{u\},u}.
\end{split}
\end{equation*}
    For any $c=\{u,v\}$, it is clear that
when identifying $\lambda^{\rm norm}_{u\rightarrow c}({\bf i})$
with    $\eta^{i}_{u\rightarrow v}$ and identifying
 $\alpha_{\{u,v\}}=
1-\sum_{{\bf i}={\bf 1, 2, 3}} \lambda^{\rm norm}_{u\rightarrow\{u,v\}}
({\bf i})$ with $\eta^*_{u\rightarrow v}$, the update rule for passed message
$(\eta_{u\rightarrow v}^1, \eta_{u\rightarrow v}^2, \eta_{u\rightarrow v}^3,
\eta_{u\rightarrow v}^*)$
in SP is resulted.

To prove the equivalence of PTP and SP summary messages, we can
follow the same procedure as we did for proving the equivalence of
PTP left messages and SP left messages.  When applying message update equations
(\ref{equ:ptp_3col_r_ij}) and (\ref{equ:ptp_3col_r_ijk}) to
equations (\ref{equ:ptp_3col_s_i}) to (\ref{equ:ptp_3col_s_ijk}) and
re-write summary messages in terms of left messages, the PTP summary
messages are updated as follows.
    \begin{eqnarray}
    \mu_{u}({\bf i})& = &\prod_{c\in C(u)}\left(1-\lambda^{\rm norm}_{V(c)\setminus\{u\}\rightarrow c}({\bf i})\right)-
    \prod_{c\in C(u)}\left(\lambda^{\rm norm}_{V(c)\setminus\{u\}
    \rightarrow c}({\bf j})+\alpha_{V(c)\setminus\{u\},u}\right)\\
    & &-\prod_{c\in C(u)}\left(\lambda^{\rm norm}_{V(c)\setminus\{u\}
    \rightarrow c}({\bf k})+\alpha_{V(c)\setminus\{u\},u}\right)
    +\prod_{c\in C(u)}
    \alpha_{V(c)\setminus\{u\},u}\nonumber\\
    \mu_{u}({\bf ij}) & = & \prod_{c\in
    C(u)}\left(\lambda^{\rm norm}_{V(c)\setminus\{u\}\rightarrow
    c}({\bf k})+\alpha_{V(c)\setminus\{u\},u}\right)-\prod_{c\in
    C(u)}\alpha_{V(c)\setminus\{u\},u}\\
    \mu_{u}({\bf ijk}) & = & \prod_{c\in
    C(u)}\alpha_{V(c)\setminus\{u\},u}.
    \end{eqnarray}

 After normalization, we have
    \begin{eqnarray}
    \mu^{\rm norm}_{u}({\bf i})& = &\frac{1}{\beta'}\cdot\left(\prod_{c\in C(u)}\left(1-\lambda^{\rm norm}_{V(c)\setminus\{u\}\rightarrow c}({\bf i})\right)-
    \prod_{c\in C(u)}\left(\lambda^{\rm norm}_{V(c)\setminus\{u\}
    \rightarrow c}({\bf j})+\alpha_{V(c)\setminus\{u\},u}\right)\right.\nonumber\\
    & &\left. -\prod_{c\in C(u)}\left(\lambda^{\rm norm}_{V(c)\setminus\{u\}
    \rightarrow c}({\bf k})+\alpha_{V(c)\setminus\{u\},u}\right)
    +\prod_{c\in C(u)}
    \alpha_{V(c)\setminus\{u\},u}\right)\\
    \mu^{\rm norm}_{u}({\bf ij}) & = &\frac{1}{\beta'}\cdot\left(\prod_{c\in
    C(u)}\left(\lambda^{\rm norm}_{V(c)\setminus\{u\}\rightarrow
    c}({\bf k})+\alpha_{V(c)\setminus\{u\},u}\right)-\prod_{c\in
    C(u)}\alpha_{V(c)\setminus\{u\},u}\right)\\
    \mu^{\rm norm}_{u}({\bf ijk}) & = &\frac{1}{\beta'}\cdot\prod_{c\in
    C(u)}\alpha_{V(c)\setminus\{u\},u},
    \end{eqnarray}
    where $\beta':=\sum_{t\in \left( \chi^* \right)^{\{u\}}}
\mu_{u}(t)$.

It is easy to show that
\begin{equation*}
\begin{split}
\beta' & = \sum_{{\bf i}={\bf 1,2,3}}\prod_{c\in
C(u)}\left(1-\lambda^{\rm norm}_{V(c)
    \setminus\{u\}\rightarrow c}({\bf i})\right)
-\sum_{{\bf i}={\bf 1,2,3}}
    \prod_{c\in C(u)}\left(\lambda^{\rm
    norm}_{V(c)\setminus\{u\}\rightarrow c}({\bf
    i})+\alpha_{V(c)\setminus\{u\},u}\right)\\
&+\prod_{c\in
    C(u)}\alpha_{V(c)\setminus\{u\},u}.
\end{split}
\end{equation*}

For any $u\in V$, it is clear that when identifying $\mu^{\rm
norm}_{u}({\bf i})$ with    $\eta^{i}_{u}$ and identifying
 $1-\sum_{{\bf i}={\bf 1, 2, 3}} \mu^{\rm
norm}_{u} ({\bf i})$ with $\eta^*_{u}$, the update rule for summary
message $(\eta_{u}^1, \eta_{u}^2, \eta_{u}^3, \eta_{u}^*)$ in SP is
resulted.

\end{proof}

This theorem suggests that for $3$-COL problems, SP {\em is} PTP.
Similar results can be shown for $k$-SAT problems --- instead of
showing this result, we will in a later section, show a more general
result, namely that weighted SP is weighted PTP for $k$-SAT
problems. It should be convincing then that the general principle of
designing SP algorithm for arbitrary CSPs is the recipe specified in
the PTP message-update rule.

In the correspondence between SP and PTP for $3$-COL problems
established in this theorem, it is worth noting that symbol $i$ in
the SP messages corresponds to the singleton token ${\bf i}$ that
contains the single element $i$, and symbol $*$ in the SP messages
corresponds to the group of all non-singleton tokens.  We note that
the fact that all non-singleton tokens can be represented by a
single symbol $*$ is rather a coincidence, intrinsically related to
the structure of $3$-COL problems, and should not be understood as a
general principle. Specifically, for $3$-COL problems, each
constraint vertex has degree $2$, and as long as a non-singleton
token is passed to a constraint vertex, the outgoing token from the
constraint vertex will be token ${\bf 123}$. It is precisely due to
this fact that all non-singleton tokens can be represented by the
same symbol --- the joker symbol $*$, as is conventionally termed.
This observation then implies that for general CSPs with non-binary
alphabet, SP, or equivalently PTP, may be expected to contain more
than one ``joker'' symbols, each corresponding to one or several
non-singleton tokens. In other words, this suggests that the notion
of ``joker'' symbol in SP messages is {\em not} a fundamental one,
and that the rather fundamental perspective of SP is the extension
of the variable alphabet to its power set with empty set excluded
--- or equivalently via a one-to-one correspondence, the set of all
tokens associated with the variable.

Finally, we remark that there can be a caveat on whether SP and PTP
are exactly equivalent, when taking into account the decimation
procedure associated with the SP algorithms. Specifically, we note
that decimation is performed based on summary messages in SP. For
$3$-COL problems, each SP summary message contains ``biases'' on
four different symbols, but each PTP summary message contains
``biases'' on seven different tokens. The natural decimation
procedure for PTP is then to fix one ``highly biased'' variable to
one of the seven tokens, rather than to one of the four symbols.
Although it is not clear at this point whether this finer procedure
may provide gains in algorithm performance, it nevertheless suggests
that PTP is slightly more general than SP. Investigation on possible
benefit of this slight generality can be an interesting direction of
research.

\subsection{Weighted PTP}
\label{subsec:w_ptp}

In the mechanism of passing random tokens that underlies the PTP
message passing rule, the outgoing token sent from a variable vertex
is a {\em function} of all incoming tokens from its upstream. A
natural angle to generalize the dependency of these outgoing tokens
on the incoming tokens is to generalize this {\em functional
dependency} to a {\em probabilistic dependency}. Specifically, using
the ``intention-command'' analogy, this probabilistic dependency
will allow the intention of a variable, conditioned on all incoming
commands from the upstream, to take any set of the values --- not
necessarily the maximal set --- that obeys by the commands, and this
probabilistic dependency is specified via the probability of each
allowed intention. This result in what we call {\em weighted PTP}.

In weighted PTP, we assume that the token $t_{v\rightarrow c}$
passed from variable vertex $x_v$ to constraint vertex $\Gamma_c$
may be any subset of the intersection of all incoming tokens passed
to $x_v$ except that passed from $\Gamma_c$, and the probability
that token $t_{v\rightarrow c}$ equals to each subset is specified
via a non-negative function $\omega_v(a|b)$ defined on
$\left(\chi^*\right)^{\{v\}}\times \left(
\left(\chi^*\right)^{\{v\}}\cup \{\emptyset_{v}\} \right) $ for each
$v\in V$. We will restrict $\omega_v(a|b)$ to an {\em obedience
conditional} on $\left(\chi^*\right)^{\{v\}}$, the definition of
which is given as follows.

\begin{defn}[Obedience Conditional]

A non-negative function $h(a|b)$ on $\left(\chi^*\right)^{\{v\}}
\times
\left(\left(\chi^*\right)^{\{v\}}\cup\{\emptyset_{v}\}\right)$ is
said to be an obedience conditional on $\left(\chi^*\right)^{\{v\}}$
if $h(a|\emptyset_{v})=0$ for all $a\in \left(\chi^*\right)^{\{v\}}$
and $h(a|b)=0$ for any $a, b \in \left(\chi^*\right)^{\{v\}}$ with
$a\not\subseteq b$.
\end{defn}

First we note that in the definition, variable $a$ in $h(\cdot)$ is
intended to refer to an ``intention'', variable $b$ is intended to
refer to a ``command'', and function $h$ is evaluated to zero if the
command is null or if the intention does not obey the command. This
is the reason for which we name such a function an ``obedience''
conditional. Second, it is also worth noting that an obedience
conditional $h$ as defined above is not a true conditional
distribution, since it is not the case that
$\sum\limits_{a}h(a|b)=1$ for all $b$. However, it is a minor
technicality to modify the definition of $h$ (without impacting the
development of any result in this paper) so that it is indeed a
conditional distribution \footnote{Given an obedience conditional
$h$, we may define a conditional distribution $\tilde{h}(a|b)$. Let
$Z$ be $\max\limits_{b\in
\left(\chi^*\right)^{\{v\}}}\sum\limits_{a\in
\left(\chi^*\right)^{\{v\}}} h(a|b)$. Let non-negative function
$\tilde{h}(a|b)$ on $\left(
\left(\chi^*\right)^{\{v\}}\cup\{\emptyset_{v}\} \right) \times
\left( \left(\chi^*\right)^{\{v\}}\cup\{\emptyset_{v}\} \right)$ be
defined as follows: $\tilde{h}(a|\emptyset_v):=[a=\emptyset_v]$;
$\tilde{h}(\emptyset_v|b):=1- \sum\limits_{a\in
\left(\chi^*\right)^{\{v\}}} h(a|b)/Z$ for all $b\neq \emptyset_v$;
and for all other $(a, b)$, $\tilde{h}(a|b):=h(a|b)/Z$. It is easy
to see that $\tilde{h}(a|b)$ is a conditional distribution. Since
eventually we will condition on that $a\neq \emptyset$, it is
straight-forward to verify that the role of $h$ is equivalent to
$\tilde{h}$.}. Thus for the purpose of this paper, one may always
regard an obedience conditional as a conditional distribution of an
intention given a command.

Apparently, function $[a=b]$ is a special case of obedience
conditional, characterizing a special {\em functional} dependency of
intention $a$ on command $b$, namely that the intention set
$a$ is exactly the command set $b$.

We now give the precise message-update rule of weighted PTP
where the only difference with PTP is in left message and summary message.

\begin{center}
\shadowbox{
\begin{minipage}{6.3in}
{\bf \large Weighted PTP Message-Update Rule}\\
\begin{eqnarray}
\label{equ:wptpvtoc}
\lambda_{v\rightarrow c}(t_{v\rightarrow c})
&:=&
\sum_{
\langle t_{b\rightarrow v} \rangle_{b\in C(v)\setminus\{c\}}
}
\omega_v\left(t_{v\rightarrow
c}\left|\bigcap_{b\in {C(v)\setminus\{c\}}}t_{b\rightarrow v}\right.\right)\prod_{b\in
C(v)\setminus\{c\}}\rho^{\rm norm}_{b\rightarrow v}(t_{b\rightarrow v})\\
\rho_{c\rightarrow v}(t_{c\rightarrow v})
&:=&
\sum_{
\langle t_{u\rightarrow c} \rangle_{u\in V(c)\setminus\{v\}}
}
\left[t_{c\rightarrow v}={\tt F}_c\left(
\prod\limits_{u\in V(c)\setminus \{v\}}
t_{u\rightarrow c}
\right)
\right]
\prod_{u\in
V(c)\setminus\{v\}}\lambda^{\rm norm}_{u\rightarrow c}(t_{u\rightarrow c})\\
\label{equ:ptp_gamma_summary} \mu_v(t_v) & := &\sum_{ \langle
t_{c\rightarrow v} \rangle_{c\in C(v)} }
\omega_v\left(t_v\left|\bigcap_{c\in C(v)}t_{c\rightarrow
v}\right.\right)\prod_{c\in C(v)}\rho^{\rm norm}_{c\rightarrow
v}(t_{c\rightarrow v}),
\end{eqnarray}
\\
and the normalized messages are defined as
\begin{eqnarray}
\lambda_{v\rightarrow c}^{\rm norm}(t_{v\rightarrow c})
&:= &
 \lambda_{v\rightarrow c}(t_{v\rightarrow c})
/\sum\limits_{t\in \left(\chi^*\right)^{\{v\}}}
\lambda_{v\rightarrow c}(t)\\
\rho_{c\rightarrow v}^{\rm norm}(t_{c\rightarrow v})
&:= &
 \rho_{c\rightarrow v}(t_{c\rightarrow v})
/\sum\limits_{t\in \left(\chi^*\right)^{\{v\}}}
\rho_{c\rightarrow v}(t)\\
\mu_{v}^{\rm norm}(t_{v})
&:= &
 \mu_{v}(t_{v})
/\sum\limits_{t\in \left(\chi^*\right)^{\{v\}}} \mu_{v}(t).
\end{eqnarray}
\end{minipage}
}
\end{center}

It is easily seen that weighted PTP is a family of algorithms,
parametrized by a collection of obedience conditionals,
$\{\omega_v:v\in V\}$, each for a coordinate. The fact that conditional
distribution $\omega_v(a|b)$ generalizes indicator function $[a=b]$
immediately implies that weighted PTP generalizes PTP, as stated in
the following lemma.

\begin{lem}
If $\omega_v(a|b):=[a=b]$ for all $v\in V$, then weighted PTP is PTP.
\end{lem}

\subsection{Weighted PTP Generalizes Weighted SP}
\label{subsec:w_ptp_gen_w_sp}

Now we will show that the weighted SP developed for $k$-SAT problems
\cite{SP:newlook} is a special case of weighted PTP. That is, for
$k$-SAT problems, when setting functions $\{\omega_v:v\in V\}$ in
weighted PTP to a particular form, weighted SP, or
$\mbox{SP}^*(\gamma)$ is resulted.

For a $k$-SAT problem, let function $\omega_v(a|b)$ for every $v\in V$
in weighted PTP be defined via a single real number $\gamma\in [0, 1]$ as
follows.

\begin{equation}
\label{eq:omega4ksat}
\omega_v(a|b):=\left\{
\begin{array}{lll}
\gamma, & {\rm if} ~a=b={\bf 01} \\
1-\gamma,& {\rm if} ~a\subset b={\bf 01} \\
1, & {\rm if} ~a=b\neq {\bf 01} \\
0, & {\rm otherwise}
\end{array}
\right.
\end{equation}

\begin{lem}
\label{lem:ptp_sat_msg} Let $\{\omega_v:v\in V\}$ in $k$-SAT be
defined as in (\ref{eq:omega4ksat}). The
message-update rule of weighted PTP is then:
\begin{eqnarray}
\label{equ:ptp_sat_l_0}
    \lambda_{v\rightarrow c}({\bf 0}) &:=& \prod\limits_{b\in C(v)\setminus \{c\}}
(\rho^{\rm norm}_{b\rightarrow v}({\bf 0})+\rho^{\rm
norm}_{b\rightarrow v}({\bf 01}))-\gamma\prod\limits_{b\in
C(v)\setminus\{c\}}\rho^{\rm norm}_{b\rightarrow v}({\bf 01})\\
\label{equ:ptp_sat_l_1}
    \lambda_{v\rightarrow c}({\bf 1}) &:=& \prod\limits_{b\in C(v)\setminus \{c\}}
(\rho^{\rm norm}_{b\rightarrow v}({\bf 1})+\rho^{\rm
norm}_{b\rightarrow v}({\bf 01}))-\gamma\prod\limits_{b\in
C(v)\setminus\{c\}}\rho^{\rm norm}_{b\rightarrow v}({\bf 01})\\
\label{equ:ptp_sat_l_*}
    \lambda_{v\rightarrow c}({\bf 01}) &:=& \gamma\prod\limits_{b\in C(v)\setminus\{c\}}
\rho^{\rm norm}_{b\rightarrow v}({\bf 01})\\
\label{equ:ptp_sat_r_0}
    \rho_{c\rightarrow v}({\bf 0}) &:=& [L_{v,c}=0]\cdot\prod\limits_{u\in V(c)\setminus
\{v\}:L_{u,c}=1}\lambda^{\rm norm}_{u\rightarrow c}({\bf
0})\cdot\prod\limits_{u\in
V(c)\setminus\{v\}:L_{u,c}=0}\lambda^{\rm norm}_{u\rightarrow c}({\bf 1})\\
\label{equ:ptp_sat_r_1}
    \rho_{c\rightarrow v}({\bf 1}) &:=& [L_{v,c}=1]\cdot\prod\limits_{u\in V(c)\setminus
\{v\}:L_{u,c}=1}\lambda^{\rm norm}_{u\rightarrow c}({\bf
0})\cdot\prod\limits_{u\in
V(c)\setminus\{v\}:L_{u,c}=0}\lambda^{\rm norm}_{u\rightarrow c}({\bf 1})\\
\label{equ:ptp_sat_r_*}
    \rho_{c\rightarrow v}({\bf 01}) &:=&
1 -\prod\limits_{{u\in V(c)\setminus
\{v\}:}\atop{L_{u,c}=1}}\lambda^{\rm norm}_{u\rightarrow c}({\bf
0})\cdot\prod\limits_{{u\in V(c)\setminus
\{v\}:}\atop{L_{u,c}=0}}\lambda^{\rm norm}_{u\rightarrow c}({\bf 1})\\
\label{equ:ptp_sat_s_0} \mu_v({\bf 0}) &:=& \prod\limits_{c\in
C(v)}(\rho^{\rm norm}_{c\rightarrow v}({\bf 0})+\rho^{\rm
norm}_{c\rightarrow v}({\bf 01}))-\gamma\prod\limits_{c\in
C(v)}\rho^{\rm norm}_{c\rightarrow v}({\bf
01})\\
\label{equ:ptp_sat_s_1} \mu_v({\bf 1}) &:=& \prod\limits_{c\in
C(v)}(\rho^{\rm norm}_{c\rightarrow v}({\bf 1})+\rho^{\rm
norm}_{c\rightarrow v}({\bf 01}))-\gamma\prod\limits_{c\in
C(v)}\rho^{\rm norm}_{c\rightarrow v}({\bf
01})\\
\label{equ:ptp_sat_s_*} \mu_v({\bf 01}) &:=&
\gamma\prod\limits_{c\in C(v)}\rho^{\rm norm}_{c\rightarrow v}({\bf
01}).
\end{eqnarray}
\end{lem}

\begin{proof}
These update equations can be immediately obtained from weighted
PTP message update equations (\ref{equ:wptpvtoc}) to
(\ref{equ:ptp_gamma_summary}), where (\ref{equ:ptp_sat_r_*}) follows from
\begin{eqnarray*}
\rho_{c\rightarrow v}({\bf 01}) &=&
\prod\limits_{u\in V(c)\setminus\{v\}}
(\lambda^{\rm norm}_{u\rightarrow c}({\bf 0})+\lambda^{\rm
norm}_{u\rightarrow c}({\bf 1})+\lambda^{\rm norm}_{u\rightarrow
c}({\bf 01}))-\prod\limits_{{u\in V(c)\setminus
\{v\}:}\atop{L_{u,c}=1}}\lambda^{\rm norm}_{u\rightarrow c}({\bf
0})\cdot\prod\limits_{{u\in V(c)\setminus
\{v\}:}\atop{L_{u,c}=0}}\lambda^{\rm norm}_{u\rightarrow c}({\bf 1})\\
& = &
1-\prod\limits_{{u\in V(c)\setminus
\{v\}:}\atop{L_{u,c}=1}}\lambda^{\rm norm}_{u\rightarrow c}({\bf
0})\cdot\prod\limits_{{u\in V(c)\setminus
\{v\}:}\atop{L_{u,c}=0}}\lambda^{\rm norm}_{u\rightarrow c}({\bf 1})
\end{eqnarray*}.
\end{proof}

%
%

\begin{thm}
\label{thm:ptp_sp_equiv_ksat} Let $\{\omega_v:v\in V\}$ in a $k$-SAT
problem be defined as in (\ref{eq:omega4ksat}). Denote by
$(\Pi^{{\tt s}~{\rm norm}}_{v\rightarrow c}, \Pi^{{\tt u}~{\rm
norm}}_{v\rightarrow c}, \Pi^{*~{\rm norm}}_{v\rightarrow c})$ the
normalized version of SP message $(\Pi^{\tt s}_{v\rightarrow c},
\Pi^{\tt u}_{v\rightarrow c}, \Pi^*_{v\rightarrow c})$, namely that
$\Pi^{{\tt s}~{\rm norm}}_{v\rightarrow c}=\Pi^{\tt s}_{v\rightarrow
c}/\left(\Pi^{\tt s}_{v\rightarrow c}+\Pi^{\tt u}_{v\rightarrow
c}+\Pi^*_{v\rightarrow c}\right)$, $\Pi^{{\tt u}~{\rm
norm}}_{v\rightarrow c}=\Pi^{\tt u}_{v\rightarrow c}/\left(\Pi^{\tt
s}_{v\rightarrow c}+\Pi^{\tt u}_{v\rightarrow c}+\Pi^*_{v\rightarrow
c}\right)$, and $\Pi^{*~{\rm norm}}_{v\rightarrow c}=\Pi^{\tt
s}_{v\rightarrow c}/\left(\Pi^{\tt s}_{v\rightarrow c}+\Pi^{\tt
u}_{v\rightarrow c}+\Pi^*_{v\rightarrow c}\right)$. Then the
correspondence between $\mbox{SP}^*(\gamma)$ message-update rule and
weighted PTP message-update rule is
\begin{eqnarray}
\label{equ:pi_s_eq} \Pi^{{\tt s}~{\rm norm}}_{v\rightarrow c}
&\leftrightarrow& [L_{v,c}=0]\cdot\lambda^{\rm norm}_{v\rightarrow
c}({\bf 0})+[L_{v,c}=1]\cdot\lambda^{\rm norm}_{v\rightarrow c}({\bf 1})\\
\label{equ:pi_u_eq}\Pi^{{\tt u}~{\rm norm}}_{v\rightarrow c}
&\leftrightarrow& [L_{v,c}=0]\cdot\lambda^{\rm norm}_{v\rightarrow
c}({\bf 1})+[L_{v,c}=1]\cdot\lambda^{\rm norm}_{v\rightarrow c}({\bf 0})\\
\label{equ:pi_*_eq}\Pi^{*~{\rm norm}}_{v\rightarrow c}
&\leftrightarrow& \lambda^{\rm norm}_{v\rightarrow
c}({\bf 01})\\
\label{equ:eta_eq}\eta_{c\rightarrow v} &\leftrightarrow& \rho^{\rm
norm}_{c\rightarrow v}({\bf
0})+\rho^{\rm norm}_{c\rightarrow v}({\bf 1})\\
\label{equ:eta_0_eq}\zeta_v^0 &\leftrightarrow& \mu_v({\bf 0})\\
\label{equ:eta_1_eq}\zeta_v^1 &\leftrightarrow& \mu_v({\bf 1})\\
\label{equ:eta_*_eq}\zeta_v^* &\leftrightarrow& \mu_v({\bf 01}).
\end{eqnarray}
\end{thm}

Prior to proving the theorem, we will introduce some notations and a
simple lemma which will be useful in the proof. For any neighboring
variable vertex $x_v$ and constraint vertex $\Gamma_c$, we will
denote by ${\bf L_{v,c}}$ the singleton token containing the single
elementary assignment that assigns coordinate $v$ the edge label
$L_{v,c}$. Similarly, we will denote by
 ${\bf \bar{L}_{v,c}}$ the singleton token
containing the single elementary assignment that assigns coordinate
$v$ the negated edge label $\bar{L}_{v,c}$. With these notations,
the following lemma immediately follows from Lemma
\ref{lem:ptp_sat_msg}.

\begin{lem} For any $(v-c)$ pair in a $k$-SAT problem, the right message $\rho^{\rm norm}_{c\rightarrow v}$ satisfies:
\begin{eqnarray}
\label{eq:simple_lem_a}
\rho_{c\rightarrow v}^{\rm norm}({\bf L_{v,c}})+ \rho_{c\rightarrow v}^{\rm norm}({\bf 01})& = & 1\\
\label{eq:simple_lem_b}
\rho_{c\rightarrow v}^{\rm norm}({\bf \bar{L}_{v,c}})+ \rho_{c\rightarrow v}^{\rm norm}({\bf 01})& = & \rho_{c\rightarrow v}^{\rm norm}({\bf 01}).
\end{eqnarray}
\end{lem}

Now we are ready to prove Theorem \ref{thm:ptp_sp_equiv_ksat}.

\begin{proof}
We will refer to the
 message correspondence in Equations (\ref{equ:pi_s_eq}) to
(\ref{equ:pi_*_eq}) as the ``left correspondence'',  the correspondence in
(\ref{equ:eta_eq}) as the ``right correspondence'', and the correspondence in
Equations
(\ref{equ:eta_0_eq}) to (\ref{equ:eta_*_eq}) as the ``summary
correspondence''.

We will prove the theorem by first showing that if the left
correspondence holds, then the right correspondence holds, and
conversely that if the right correspondence holds, then the left
correspondence holds. This should prove that correspondence between
$\mbox{SP}^*(\gamma)$ and weighted PTP in their passed messages.  We
will then complete the proof by showing the summary correspondence.

First suppose that the left correspondence holds, namely that $\Pi^{{\tt
s}~{\rm norm}}_{v\rightarrow c}= [L_{v,c}=0]\cdot\lambda^{\rm
norm}_{v\rightarrow c}({\bf 0})+[L_{v,c}=1]\cdot\lambda^{\rm
norm}_{v\rightarrow c}({\bf 1}),$ $\Pi^{{\tt u}~{\rm
norm}}_{v\rightarrow c} = [L_{v,c}=0]\cdot\lambda^{\rm
norm}_{v\rightarrow c}({\bf 1})+[L_{v,c}=1]\cdot\lambda^{\rm
norm}_{v\rightarrow c}({\bf 0}),$ and $\Pi^{*~{\rm
norm}}_{v\rightarrow c} = \lambda^{\rm norm}_{v\rightarrow c}({\bf
01}).$

In each iteration, by Lemma \ref{lem:ptp_sat_msg}
and the fact $[L_{v,c}=1]+[L_{v,c}=0]=1$ for every $(v-c)$ pair,
the right messages satisfy
\begin{eqnarray*}
\rho_{c\rightarrow v}({\bf 0})+\rho_{c\rightarrow v}({\bf 1}) +
\rho_{c\rightarrow v}({\bf 0}1)
& & =[L_{v,c}=0]\cdot\prod\limits_{u\in
V(c)\setminus\{v\}:L_{u,c}=1}\lambda^{\rm norm}_{u\rightarrow
c}({\bf 0})\cdot\prod\limits_{u\in
V(c)\setminus\{v\}:L_{u,c}=0}\lambda^{\rm norm}_{u\rightarrow c}({\bf 1})\\
& & + [L_{v,c}=1]\cdot\prod\limits_{u\in
V(c)\setminus\{v\}:L_{u,c}=1}\lambda^{\rm norm}_{u\rightarrow
c}({\bf 0})\cdot\prod_{u\in
V(c)\setminus\{v\}:L_{u,c}=0}\lambda^{\rm norm}_{u\rightarrow c}({\bf 1})\\
& & + 1 -\prod\limits_{u\in V(c)\setminus\{v\}:L_{u,c}=1}\lambda^{\rm
norm}_{u\rightarrow c}({\bf 0})\cdot\prod\limits_{u\in
V(c)\setminus\{v\}:L_{u,c}=0}\lambda^{\rm norm}_{u\rightarrow c}({\bf 1})\\
&&=1.
\end{eqnarray*}
That is, each right message $\rho_{c\rightarrow v}$ is already
normalized, or $\rho_{c\rightarrow v}= \rho^{\rm norm}_{c\rightarrow
v}.$ Then
\begin{eqnarray*}
\rho^{\rm norm}_{c\rightarrow v}({\bf 0}) +
\rho^{\rm norm}_{c\rightarrow v}({\bf 1}) & = &
\rho_{c\rightarrow v}({\bf 0}) + \rho_{c\rightarrow v}({\bf 1})\\
&=& [L_{v,c}=0]\cdot\prod_{u\in
V(c)\setminus\{v\}:L_{u,c}=1}\lambda^{\rm norm}_{u\rightarrow
c}({\bf 0})\cdot\prod_{u\in
V(c)\setminus\{v\}:L_{u,c}=0}\lambda^{\rm norm}_{u\rightarrow c}({\bf 1})\\
&& + [L_{v,c}=1]\cdot\prod_{u\in
V(c)\setminus\{v\}:L_{u,c}=1}\lambda^{\rm norm}_{u\rightarrow
c}({\bf 0})\cdot\prod_{u\in
V(c)\setminus\{v\}:L_{u,c}=0}\lambda^{\rm norm}_{u\rightarrow c}({\bf 1})\\
& = & \prod_{u\in V(c)\setminus\{v\}:L_{u,c}=1}\lambda^{\rm
norm}_{u\rightarrow c}({\bf 0})\cdot\prod_{u\in
V(c)\setminus\{v\}:L_{u,c}=0}\lambda^{\rm norm}_{u\rightarrow
c}({\bf
1})\\
& = & \prod_{u\in
V(c)\setminus\{v\}:L_{u,c}=1}([L_{u,c}=1]\cdot\lambda^{\rm
norm}_{u\rightarrow c}({\bf
0})+[L_{u,c}=0]\cdot\lambda^{\rm norm}_{u\rightarrow c}({\bf 1}))\\
&& \cdot\prod_{u\in
V(c)\setminus\{v\}:L_{u,c}=0}([L_{u,c}=1]\cdot\lambda^{\rm
norm}_{u\rightarrow c}({\bf
0})+[L_{u,c}=0]\cdot\lambda^{\rm norm}_{u\rightarrow c}({\bf 1}))\\
& = & \prod_{u\in V(c)\setminus\{v\}}([L_{u,c}=1]\cdot\lambda^{\rm
norm}_{u\rightarrow c}({\bf 0})+[L_{u,c}=0]\cdot\lambda^{\rm
norm}_{u\rightarrow c}({\bf 1}))\\
& \stackrel{(a)}{=} &\prod_{u\in V(c)\setminus\{v\}} \Pi^{{\tt
u}~{\rm norm}}_{u\rightarrow c}\\
& \stackrel{(b)}{=} & \prod_{u\in V(c)\setminus \{v\}}
\frac{\Pi^{\tt u}_{u\rightarrow c}}{\Pi^{\tt u}_{u\rightarrow
c}+\Pi^{\tt s}_{u\rightarrow c}+\Pi^{*}_{u\rightarrow c}}\\
& = &
\eta_{c\rightarrow v},
\end{eqnarray*}
where equality $(a)$ is  due to the assumed left correspondence, and
equality $(b)$ follows from the definition of $\Pi^{{\tt u}~{\rm
norm}}_{v\rightarrow c}$. Thus we have shown that if the left
correspondence holds, then the right correspondence holds.

Now suppose that the right correspondence holds, namely that
$\eta_{c\rightarrow v} = \rho^{\rm norm}_{c\rightarrow v}({\bf
0})+\rho^{\rm norm}_{c\rightarrow v}({\bf 1})$ for every $(v-c)$
pair. Following the PTP message-update equations (\ref{equ:ptp_sat_l_0})
to (\ref{equ:ptp_sat_l_*}), we have
\begin{eqnarray*}
&& [L_{v,c}=0]\cdot\lambda_{v\rightarrow
c}({\bf 0})+[L_{v,c}=1]\cdot\lambda_{v\rightarrow c}({\bf 1})\\
&=& [L_{v,c}=0]\cdot\left(\prod_{b\in C(v)\setminus\{c\}}(\rho^{\rm
norm}_{b\rightarrow v}({\bf 0})+\rho^{\rm norm}_{b\rightarrow
v}({\bf 01}))-\gamma\prod_{b\in C(v)\setminus\{c\}}\rho^{\rm
norm}_{b\rightarrow
v}({\bf 01})\right)\\
&& +[L_{v,c}=1]\cdot\left(\prod_{b\in C(v)\setminus\{c\}}(\rho^{\rm
norm}_{b\rightarrow v}({\bf 1})+\rho^{\rm norm}_{b\rightarrow
v}({\bf 01}))-\gamma\prod_{b\in C(v)\setminus\{c\}}\rho^{\rm
norm}_{b\rightarrow
v}({\bf 01})\right)\\
&=& [L_{v,c}=0]\cdot\prod_{b\in C(v)\setminus\{c\}}(\rho^{\rm
norm}_{b\rightarrow v}({\bf 0})+\rho^{\rm norm}_{b\rightarrow
v}({\bf 01}))+[L_{v,c}=1]\cdot\prod_{b\in
C(v)\setminus\{c\}}(\rho^{\rm norm}_{b\rightarrow v}({\bf
1})+\rho^{\rm norm}_{b\rightarrow
v}({\bf 01}))\\
&& - \gamma \prod_{b\in C(v)\setminus\{c\}}\rho^{\rm
norm}_{b\rightarrow
v}({\bf 01})\\
&\stackrel{(\ref{eq:simple_lem_b})}{=}& [L_{v,c}=0]\cdot\prod_{b\in C_{c}^{\tt s}(v)}(\rho^{\rm
norm}_{b\rightarrow v}({\bf 0})+\rho^{\rm norm}_{b\rightarrow
v}({\bf 01}))\cdot\prod_{b\in C_{c}^{\tt
u}(v)}\rho^{\rm norm}_{b\rightarrow v}({\bf 01})\\
&& + [L_{v,c}=1]\cdot\prod_{b\in C_c^{\tt s}(v)}(\rho^{\rm
norm}_{b\rightarrow v}({\bf 1})+\rho^{\rm norm}_{b\rightarrow
v}({\bf 01}))\cdot\prod_{b\in C_c^{\tt u}(v)}\rho^{\rm
norm}_{b\rightarrow v}({\bf 01})-\gamma\prod_{b\in
C(v)\setminus\{c\}}\rho^{\rm norm}_{b\rightarrow v}({\bf 01})\\
&\stackrel{(\ref{eq:simple_lem_a})}{=}& [L_{v,c}=0]\cdot\prod_{b\in C_c^{\tt u}(v)}\rho^{\rm
norm}_{b\rightarrow v}({\bf 01}) + [L_{v,c}=1]\cdot\prod_{b\in
C_c^{\tt u}(v)}\rho^{\rm norm}_{b\rightarrow v}({\bf 01}) -
\gamma\prod_{b\in
C(v)\setminus\{c\}}\rho^{\rm norm}_{b\rightarrow v}({\bf 01})\\
&=& \prod_{b\in C_c^{\tt u}(v)}\rho^{\rm norm}_{b\rightarrow v}({\bf
01}) -
\gamma\prod_{b\in C(v)\setminus\{c\}}\rho^{\rm norm}_{b\rightarrow v}({\bf 01})\\
&=& \prod_{b\in C_c^{\tt u}(v)}\rho^{\rm norm}_{b\rightarrow v}({\bf
01})\cdot\left(1-\gamma\prod_{b\in C_c^{\tt s}(v)}\rho^{\rm
norm}_{b\rightarrow
v}({\bf 01})\right)\\
&=& \prod_{b\in C_c^{\tt u}(v)}(1-\rho^{\rm norm}_{b\rightarrow
v}({\bf 0})-\rho^{\rm norm}_{b\rightarrow v}({\bf
1}))\cdot\left(1-\gamma\prod_{b\in C_c^{\tt s}(v)}(1-\rho^{\rm
norm}_{b\rightarrow v}({\bf 0})-\rho^{\rm norm}_{b\rightarrow
v}({\bf 1}))\right)\\
&\stackrel{(c)}{=}& \prod_{b\in C_c^{\tt u}(v)}(1-\eta_{b\rightarrow v})\cdot\left(1-\gamma\prod_{b\in C_c^{\tt s}(v)}(1-\eta_{b\rightarrow v})\right)\\
&=&\Pi^{\tt s}_{v\rightarrow c},
\end{eqnarray*}
where equality (c) above is due to the assumed right correspondence. We will
denote this result by (A).

Following very similar procedures, it can be shown that
\begin{eqnarray*}
&& [L_{v,c}=0]\cdot\lambda_{v\rightarrow
c}({\bf 1})+[L_{v,c}=1]\cdot\lambda_{v\rightarrow c}({\bf 0})\\
&=& \prod_{b\in C_c^{\tt s}(v)}(1-\rho^{\rm norm}_{b\rightarrow
v}({\bf 0})-\rho^{\rm norm}_{b\rightarrow v}({\bf
1}))\cdot\left(1-\gamma\prod_{b\in C_c^{\tt u}(v)}(1-\rho^{\rm
norm}_{b\rightarrow v}({\bf 0})-\rho^{\rm norm}_{b\rightarrow
v}({\bf 1}))\right)\\
&=& \Pi^{\tt u}_{v\rightarrow c}
\end{eqnarray*}
We will denote this result by (B).

Similarly,
\begin{eqnarray*}
\lambda_{v\rightarrow c}({\bf 01}) &=& \gamma\prod_{b\in C_c^{\tt
s}(v)}(1-\rho^{\rm norm}_{b\rightarrow v}({\bf 0})-\rho^{\rm
norm}_{b\rightarrow v}({\bf 1}))\cdot\prod_{b\in C_c^{\tt
u}(v)}(1-\rho^{\rm norm}_{b\rightarrow v}({\bf 0})-\rho^{\rm
norm}_{b\rightarrow v}({\bf 1}))\\
&=&  \Pi^*_{v \rightarrow c}.
\end{eqnarray*}
We will denote this result by (C).

Combining results (A), (B) and (C), we have
\[
\lambda_{v\rightarrow c}({\bf 0})
+
\lambda_{v\rightarrow c}({\bf 1})
+
\lambda_{v\rightarrow c}({\bf 01})
= \Pi_{v\rightarrow c}^{\tt u}+ \Pi_{v\rightarrow c}^{\tt s}
+ \Pi_{v\rightarrow c}^{*}.
\]
That is, the scaling constant for normalizing
$(\lambda_{v\rightarrow c}({\bf 0}), \lambda_{v\rightarrow c}({\bf 1}),
\lambda_{v\rightarrow c}({\bf 01}))$ and that for
normalizing
$(\Pi_{v\rightarrow c}^{\tt u}, \Pi_{v\rightarrow c}^{\tt s},
\Pi_{v\rightarrow c}^{*})$ are identical. Then results (A), (B) and
(C) respectively translate to
\begin{eqnarray*}
[L_{v,c}=1]\cdot\lambda^{\rm norm}_{v\rightarrow c}({\bf
1})+[L_{v,c}=0]\cdot\lambda^{\rm norm}_{v\rightarrow c}({\bf 0}) & =
& \Pi^{{\tt s}~{\rm norm}}_{v\rightarrow c}\\
{[L_{v,c}=0]\cdot\lambda^{\rm norm}_{v\rightarrow c}({\bf
1})+[L_{v,c}=1]\cdot\lambda^{\rm norm}_{v\rightarrow c}({\bf 0})} &
=& \Pi^{{\tt u}~{\rm norm}}_{v\rightarrow c}\\
\lambda^{\rm norm}_{v\rightarrow c}({\bf 01})&=&\Pi^{*~{\rm
norm}}_{v\rightarrow c}.
\end{eqnarray*}

At this point we have established the correspondence between the
passed messages in weighted PTP and those in weighted SP. We now prove the
summary correspondence.

Starting from Lemma \ref{lem:ptp_sat_msg}, we have
\begin{eqnarray*}
\mu_v({\bf 0}) & = &
\prod_{c\in C(v)} \left(\rho^{\rm norm}_{c\rightarrow v}({\bf 0})+
\rho^{\rm norm}_{c\rightarrow v}({\bf 01})\right)
-\gamma
\prod_{c\in C(v)}\rho^{\rm norm}_{c\rightarrow v}({\bf 01})\\
& = &
\prod_{c\in C^{\rm 1}(v)} \left(\rho^{\rm norm}_{c\rightarrow v}({\bf 0})+
\rho^{\rm norm}_{c\rightarrow v}({\bf 01})\right)
\prod_{c\in C^{\rm 0}(v)} \left(\rho^{\rm norm}_{c\rightarrow v}({\bf 0})+
\rho^{\rm norm}_{c\rightarrow v}({\bf 01})\right)
-\gamma
\prod_{c\in C(v)}\rho^{\rm norm}_{c\rightarrow v}({\bf 01})\\
&
\stackrel{(\ref{eq:simple_lem_a}),(\ref{eq:simple_lem_b})}{=} &
\prod_{c\in C^{\rm 1}(v)} \rho^{\rm norm}_{c\rightarrow v}({\bf 01})
-\gamma
\prod_{c\in C(v)}\rho^{\rm norm}_{c\rightarrow v}({\bf 01})\\
& = &
\left(1-\gamma \prod_{c\in C^{\rm 0}(v)}\rho^{\rm norm}_{c\rightarrow v}({\bf 01})\right)\prod_{c\in C^{\rm 1}(v)} \rho^{\rm norm}_{c\rightarrow v}({\bf 01})\\
&=&
\left(1-\gamma \prod_{c\in C^{\rm 0}(v)}
\left(1-\rho^{\rm norm}_{c\rightarrow v}({\bf 0})-\rho^{\rm norm}_{c\rightarrow
v}({\bf 1})\right)
\right)
\prod_{c\in C^{\rm 1}(v)}\left(
1- \rho^{\rm norm}_{c\rightarrow v}({\bf 0})- \rho^{\rm norm}_{c\rightarrow v}({\bf 1})
\right)\\
& \stackrel{(d)}{=} &
\left(1-\gamma \prod_{c\in C^{\rm 0}(v)}
\left(1-\eta_{c\rightarrow v}\right)
\right)
\prod_{c\in C^{\rm 1}(v)}\left(
1- \eta_{c\rightarrow v})
\right)\\
& = &
\zeta^{\rm 0}_{v}
\end{eqnarray*}
where $(d)$ above is due to the right correspondence that we just
proved.

Symmetrically, it can be shown that
\begin{eqnarray*}
\mu_v({\bf 1})
&=&
\left(1-\gamma \prod_{c\in C^{\rm 1}(v)}
\left(1-\rho^{\rm norm}_{c\rightarrow v}({\bf 0})-\rho^{\rm norm}_{c\rightarrow
v}({\bf 1})\right)
\right)
\prod_{c\in C^{\rm 0}(v)}\left(
1- \rho^{\rm norm}_{c\rightarrow v}({\bf 0})- \rho^{\rm norm}_{c\rightarrow v}({\bf 1})
\right)\\
& = &
\left(1-\gamma \prod_{c\in C^{\rm 1}(v)}
\left(1-\eta_{c\rightarrow v}\right)
\right)
\prod_{c\in C^{\rm 0}(v)}\left(
1- \eta_{c\rightarrow v})
\right)\\
& = &
\zeta^{\rm 1}_{v}.\\
\end{eqnarray*}

Finally, it is straight-forward to see
\begin{eqnarray*}
\mu_v({\bf 01}) & = &
\gamma\prod_{c\in C^{\rm 0}(v)}\left(1-\rho^{\rm
norm}_{c\rightarrow v}({\bf 0})-\rho^{\rm norm}_{c\rightarrow
v}({\bf 1})\right)
\prod_{c\in C^{1}(v)}\left(1-\rho^{\rm norm}_{c\rightarrow
v}({\bf 0})-\rho^{\rm norm}_{c\rightarrow v}({\bf 1})\right)\\
& = &
\gamma\prod_{c\in C^{\rm 0}(v)}\left(1-\eta_{c\rightarrow v}\right)
\prod_{c\in C^{1}(v)}\left(1-\eta_{c\rightarrow v}\right)\\
& = &
\zeta_{v}^{*}.
\\
\end{eqnarray*}
This proves the summary correspondence and completes the proof.
\end{proof}

This theorem asserts that weighted SP developed for $k$-SAT problems is
an instance of weighted PTP that we propose in this paper, or alternatively
phrased, weighted PTP generalizes weighted SP from the context of $k$-SAT problems to arbitrary CSPs with arbitrary
variable alphabets.  When specifying parameter $\gamma$ to be $1$, this
result immediately implies that non-weighted SP is non-weighted PTP for $k$-SAT
problems.

Additionally, we note that in the correspondence between the summary
messages of weighted PTP and weighted SP in the above theorem, it is
clear that symbols $0, 1,$ and $*$ in weighted SP (or SP)
corresponds to tokens (sets) ${\bf 0}$, ${\bf 1}$ and ${\bf 01}$
respectively. In addition, if we use notation ${\bf
L_{v,c}}$, we may re-write the correspondence between the left
messages of weighted SP and those of weighted PTP in the above
theorem as
\begin{eqnarray*}
\Pi_{v\rightarrow c}^{\tt s}
&\leftrightarrow &
\lambda_{v\rightarrow c}({\bf L_{v, c}})\\
\Pi_{v\rightarrow c}^{\tt u}
&\leftrightarrow &
\lambda_{v\rightarrow c}({\bf \bar{L}_{v,c}})\\
\Pi_{v\rightarrow c}^{\rm *}
&\leftrightarrow &
\lambda_{v\rightarrow c}({\bf 01})
\end{eqnarray*}
That is, symbols ``$\tt s$'' and ``$\tt u$'' in SP respectively
correspond to singleton set ${\bf L_{v,c}}$ and ${\bf
\bar{L}_{v,c}}$. These observations suggest that, although blurred
by the addition of single symbol $*$ to the variable alphabet, the
true alphabet used as the support of SP messages is the set of all
tokens associated with the variable, or equivalently, the power set
of the original alphabet with the empty set removed.

At this point, questions may naturally arise pertaining to
what PTP and weighted PTP do towards the goal of solving a CSP. Although
rigorous question 
this question remains largely open at this point, we present some preliminary results in Appendix B.  From Appendix B, intuitively one may view PTP or weighted PTP as essentially updating a {\em random rectangle} whose sides are independently distributed random variables; as PTP iterates, it drives some side of the random rectangle to being deterministically biased towards a singleton that contains the solution of the CSP. The reader is referred to Appendix B for more detailed exposition.

\section{The Reduction of SP from BP}
\label{sec:bp}

At this point, we have identified SP with an equivalent but
 probabilistically interpretable
algorithmic procedure, PTP, and generalized weighted SP from the
special case of $k$-SAT and binary problems to arbitrary CSPs, in
terms of weighted PTP. Now we are in the position to discuss the
reduction of SP from BP, where we will refer to SP exclusively as
PTP,  and weighted SP exclusively as weighted PTP.

As is well known, the derivation of the BP algorithm is based on a
well-defined factoring function, or seen from a probabilistic
perspective, a Markov random field (MRF).  Thus, whether PTP or
weighted PTP may be reduced from BP boils down to whether there is an
MRF formulation on which the derived BP algorithm coincides with PTP
or weighted PTP. In \cite{SP:newlook}, an MRF is
constructed for $k$-SAT problem, on which BP reduces to what we now
call weighted PTP. In \cite{rtu:sp_isit_06}, similar results are
shown using a different MRF formalism, where (generalized) states
are introduced and 
the MRF is represented by
a Forney graph or normal realization\cite{Forney:normalGraph}. Although in some sense, the
normally realized MRF formalism of \cite{rtu:sp_isit_06} is
equivalent to the MRF of \cite{SP:newlook}, the Forney-graph
formalism in \cite{rtu:sp_isit_06} makes the development cleaner and
more transparent, and the explicit introduction of states provides a
better correspondence with the weighted PTP messages.

In this section, we first generalize the MRF formalism, in the style of 
\cite{SP:newlook} or \cite{rtu:sp_isit_06}, to arbitrary CSPs, and derive
the corresponding BP algorithm. We then investigate whether the derived 
BP algorithm may be reduced to 
PTP or weighted PTP. We will begin this investigation with the special case of $k$-SAT problems, and then proceed to the $3$-COL
problems and to general CSPs. Re-developing the results of \cite{SP:newlook}
and \cite{rtu:sp_isit_06} for $k$-SAT problems, we show that the BP algorithm on the normally realized MRF is readily reducible to weighted PTP as long as the BP messages are 
initialized to satisfying certain condition. We note that this condition, when
satisfied in the first BP iteration, will necessarily be satisfied in later
iterations in $k$-SAT problems. 
Identifying the important role of this condition, we call this
condition the {\em state-decoupling condition}. However, as we proceed to show, in $3$-COL problems, it is impossible for the state-decoupling condition 
to hold true non-trivially
across all BP iterations. Nevertheless, if one manually 
manipulate the BP messages to impose this condition in every iteration, 
which results in a modified BP message-update rule referred to as {\em state-decoupled BP} or SDBP in short,  then the (SD)BP messages will 
still reduce to PTP. This on one hand justifies the role of the state-decoupling 
condition in BP-to-PTP reduction, and on the other hand suggests that 
for general CSPs, PTP (or SP) is not a special case of the BP algorithm. We then proceed further by investigating whether the state-decoupling condition is sufficient for BP to reduce to PTP or weighted PTP for general CSPs. To that end, we show that yet another ``local compatibility'' 
condition concerning the structure of the CSP (in terms of the interaction between neighboring constraints) is required for SDBP to reduce to PTP or weighted PTP.

\subsection{Normally Realized Markov Random Field}

Given a CSP  represented by factor graph $G$, we now define its
corresponding {\em normally realized Markov random field}
$\tilde{G}$ using a Forney graph representation
\cite{Forney:normalGraph}. We note that random variables involved in
the probability mass function (PMF) represented by $\tilde{G}$ are
no longer those associated with factor graph (or equivalently MRF)
$G$, but rather a new set of random variables, each distributed over
the set of {\em tokens} associated with a coordinate. Additionally,
as the central component of the Forney graph, another set of random
variables, typically called {\em generalized states} or simply {\em
states}, are also included.

Specifically, as a graph, $\tilde{G}$ can be constructed by adding a
``half-edge'' to each variable vertex of $G$. As a factor graph,
$\tilde{G}$ uses a different notation: edges and half edges are
interpreted as ``variables'' and vertices are interpreted as local
functions; a variable is an argument of the function if and only if
the corresponding edge or half edge is incident on the corresponding
vertex. We now define each variable and local function in
$\tilde{G}$.
\begin{itemize}
\item Each local function (or vertex)
in $\tilde{G}$ corresponding to variable vertex
$x_v$ in $G$ will be denoted by $g_v(\cdot)$, and referred to as a {\em left
function}.
\item Each local function (or vertex) in $\tilde{G}$ corresponding to
function vertex $\Gamma_c$ will be denoted by $f_c(\cdot)$, and
referred to as a {\em right function}.
\item The half edge incident on $g_v$ represents variable $y_v$, referred to
as a {\em side}, taking values from $\left(\chi^*\right)^{\{v\}}$.
\item The edge connecting left function $g_v$ and right function $f_c$ represents variable $s_{v,c}$, referred to as a {\em state}, taking values from
$\left(\chi^*\right)^{\{v\}}\times\left(\chi^*\right)^{\{v\}}$. We will also
write state $s_{v,c}$ as pair $\left(s^L_{v,c}, s^R_{v,c}\right)$ of
{\em left state} $s^L_{v,c}$ and {\em right state} $s^R_{v,c}$.
\item Left function $g_v$ for $v\in V$ is defined as
\begin{equation}
\label{equ:leftfunction}
g_v(y_v,s_{v,C(v)}):=\omega_v\left(y_v\left|\bigcap_{c\in
C(v)}s^R_{v,c}\right.\right)\cdot\prod_{c\in C(v)}[s^L_{v,c}=y_v],
\end{equation}
where $s_{v, C(v)}$ is the short-hand notation for $\langle s_{v,c}\rangle_{c\in C(v)}$ and $\omega_v$ is an obedience conditional on
 $\left(\chi^*\right)^{\{v\}}$.

\item Right function $f_c$ for each $c\in C$ is defined as
\begin{equation}
\label{equ:rightfunction} f_c(s_{V(c),c}):=\prod_{v\in
V(c)}[s^R_{v,c}={\tt F}_c(s^L_{V(c)\setminus\{v\},c})],
\end{equation}
where $s_{V(c), c}$ is the short-hand notation for $\langle s_{v,c}\rangle_{v\in V(c)}$.
\item The global function represented by $\tilde{G}$ is
\begin{equation}
\label{equ:globalfunction} F(y_V,s_{V,C}):=\prod_{v\in
V}g_v(y_v,s_{v,C(v)})\cdot\prod_{c\in C}f_c(s_{V(c),c}),
\end{equation}
where $s_{V,C}$ is the short-hand notation for $\{s_{v,c}: \forall (v-c)\}$.
\end{itemize}

It is clear that upon normalization, function $F$ may represent a PMF and
the factorization of $F$ encoded by $\tilde{G}$ realizes an MRF.  An example of  such normally realized MRF, corresponding to the toy 3-SAT
problem in Figure  \ref{fig:ksat_fg},
 is given in Fig. \ref{fig:ksat_forneyG}.

Using the ``intention-command'' analogy, 
one may view that for any $v$, 
both $y_v$ and each left state $s_{v,c}^L$ stores the intention of variable
$x_v$, and 
that for any given $c$, each right state $s_{v,c}^R$ stores the command
of constraint $\Gamma_c$ sent to 
variable $v$. The intention of variable $x_v$ depends on 
the intersection of all incoming commands probabilistically via the 
obedience conditional $\omega_v$. The command of $\Gamma_c$ sent to each 
variable $x_v$ need to equal the forced token by
the rectangle formed by the intentions from all other neighboring variables.

We say that a configuration of $(y_V, s_{V, C})$ is
{\em valid} under $F$ if it is in the support of function $F$ (namely, if it
gives rise to a non-zero value of function $F$). Further, rectangle
$y_V$ is said to be {\em valid} under $F$
if there exists a configuration of $s_{V,C}$
such that $(y_V, s_{V,C})$ is valid under $F$. Then it immediately follows that
the PMF represented by MRF $\tilde{G}$, upon marginalizing over states $s_{V,C}$, characterizes
the set of all valid rectangles under $F$ (via the support of the marginal of
$F$ on $y_V$). We now give an intuitive explanation of the MRF defining the
distribution of rectangle $y_V$.

\begin{figure}[htb]
\begin{center}
  \scalebox{0.45}{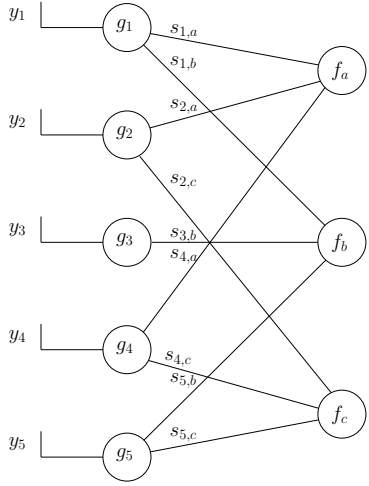}
  \caption{The Forney graph representing the normal realization of the
toy problem in Figure \ref{fig:ksat_fg}.}
  \label{fig:ksat_forneyG}
\end{center}
\end{figure}

A simple property of such MRFs is given in the following lemma, which
immediately follows from the definition of the left functions.

\begin{lem}
\label{lem:s_y_relation} If configuration $(y_V,s_{V,C})$ is  valid
under $F$, then it holds for every $(v-c)$ that
\[
s_{v,c}^L=y_v\subseteq s_{v,c}^R.
\]
\end{lem}

Now we consider applying the BP message-update rule on the
Forney graph $\tilde{G}$ we just defined, where we will use $\rho_{c\rightarrow
v}$  (referred to as a right message) and
$\lambda_{v\rightarrow c}$ (referred to as a left message) to
denote the message passed
from a right function $f_c$ to a left function $g_v$ and the message
passed from left function $g_v$ to right function $f_c$
respectively, and use $\mu_v$ to denote the summary message at
 variable $y_v$. We note that both right message
$\rho_{c\rightarrow v}$ and
left message $\lambda_{v\rightarrow c}$ are functions on the state space
$\left(\chi^*\right)^{\{v\}}\times\left(\chi^*\right)^{\{v\}}$.

\begin{lem}\label{lem:MRF_BP}
The BP message-update rule on Forney graph
$\tilde{G}$  is:
\begin{eqnarray}
\label{equ:bp_l_mssg}\lambda_{v\rightarrow c}(s^L_{v,c},s^R_{v,c})
&:=&
\sum_{s^R_{v,C(v)\setminus\{c\}}}\omega_v\left(s^L_{v,c}\left|\bigcap_{b\in
C(v)}s^R_{v,b}\right.\right)\prod_{b\in
C(v)\setminus\{c\}}\rho_{b\rightarrow v}(s^L_{v,c},s^R_{v,b})\\
\label{equ:bp_r_mssg}\rho_{c\rightarrow v}(s^L_{v,c},s^R_{v,c}) &:=&
\sum_{s^L_{V(c)\setminus\{v\},c}}\left[s^R_{v,c}={\tt
F}_c(s^L_{V(c)\setminus\{v\},c})\right]\prod_{u\in
V(c)\setminus\{v\}}\lambda_{u\rightarrow c}(s^L_{u,c},{\tt
F}_c(s^L_{V(c)\setminus \{u\},c}))\\
\label{equ:bp_s_mssg}\mu_v(y_v) &:=&
\sum_{s^R_{v,C(v)}}\omega_v\left(y_v\left|\bigcap_{c\in
C(v)}s^R_{v,c}\right.\right)\prod_{c\in C(v)}\rho_{c\rightarrow
v}(y_v,s^R_{v,c}).
\end{eqnarray}
\end{lem}

Before proving this lemma, it is useful to note the following elementary
results.
\begin{lem}
\begin{enumerate}
\item For any function $\phi$,
\begin{equation}
\sum_{y}\phi(x, y)[y=z]=\phi(x, z).
\label{eq:sum_of_indicator}
\end{equation}

\item
For any collection of functions $\phi_1, \phi_2, \ldots, \phi_m$,
\begin{equation}
\sum_{x_1, x_2, \ldots, x_n}\prod_{i=1}^n\phi_i(x_i) =
\prod_{i=1}^n \sum_{x_i}\phi_i(x_i).
\label{eq:sum_prod}
\end{equation}
\end{enumerate}
\end{lem}

We now prove Lemma \ref{lem:MRF_BP}.

\begin{proof}
{\small 
\begin{eqnarray*}
\lambda_{v\rightarrow c}(s^L_{v,c},s^R_{v,c})
&=&
\sum_{y_v}\sum_{s_{v,C(v)\setminus\{c\}}}\omega_v\left(y_v\left|\bigcap_{b\in
C(v)}s^R_{v,b}\right.\right)\prod_{b\in
C(v)}[s^L_{v,b}=y_v]\prod_{b\in
C(v)\setminus\{c\}}\rho_{b\rightarrow v}(s^L_{v,b},s^R_{v,b})\\
&=&
\sum_{y_v}[s^L_{v,c}=y_v]\sum_{s^R_{v,C(v)\setminus\{c\}}}\omega_v\left(y_v\left|\bigcap_{b\in
C(v)}s^R_{v,b}\right.\right)\sum_{s^L_{v,C(v)\setminus\{c\}}}\prod_{b\in
C(v)\setminus\{c\}}\left(\rho_{b\rightarrow
v}(s^L_{v,b},s^R_{v,b})\cdot[s^L_{v,b}=y_v]\right)\\
&\stackrel{(\ref{eq:sum_prod})}{=}&
\sum_{y_v}[s^L_{v,c}=y_v]\sum_{s^R_{v,C(v)\setminus\{c\}}}\omega_v\left(y_v\left|\bigcap_{b\in
C(v)}s^R_{v,b}\right.\right)\prod_{b\in
C(v)\setminus\{c\}}\sum_{s^L_{v,b}}\left(\rho_{b\rightarrow
v}(s^L_{v,b},s^R_{v,b})\cdot[s^L_{v,b}=y_v]\right)\\
&\stackrel{(\ref{eq:sum_of_indicator})}{=}&
\sum_{y_v}[s^L_{v,c}=y_v]\sum_{s^R_{v,C(v)\setminus\{c\}}}\omega_v\left(y_v\left|\bigcap_{b\in
C(v)}s^R_{v,b}\right.\right)\prod_{b\in
C(v)\setminus\{c\}}\rho_{b\rightarrow v}(y_v,s^R_{v,b})\\
&=&
\sum_{s^R_{v,C(v)\setminus\{c\}}}\omega_v\left(s^L_{v,c}\left|\bigcap_{b\in
C(v)}s^R_{v,b}\right.\right)\prod_{b\in
C(v)\setminus\{c\}}\rho_{b\rightarrow v}(s^L_{v,c},s^R_{v,b}).
\end{eqnarray*}
}
{\small
\begin{eqnarray*}
\rho_{c\rightarrow v}(s^L_{v,c},s^R_{v,c}) &=&
\sum_{s_{V(c)\setminus\{v\},c}}\prod_{u\in V(c)}[s^R_{u,c}={\tt
F}_c(s^L_{V(c)\setminus\{u\},c})]\prod_{u\in
V(c)\setminus\{v\}}\lambda_{u\rightarrow c}(s^L_{u,c},s^R_{u,c})\\
&=& \sum_{s^L_{V(c)\setminus\{v\},c}}[s^R_{v,c}={\tt
F}_c(s^L_{V(c)\setminus\{v\},c})]\sum_{s^R_{V(c)\setminus\{v\},c}}\prod_{u\in
V(c)\setminus\{v\}}\left([s^R_{u,c}={\tt
F}_c(s^L_{V(c)\setminus\{u\},c})]\cdot \lambda_{u\rightarrow
c}(s^L_{u,c},s^R_{u,c})\right)\\
&\stackrel{(\ref{eq:sum_prod})}{=}& \sum_{s^L_{V(c)\setminus\{v\},c}}[s^R_{v,c}={\tt
F}_c(s^L_{V(c)\setminus\{v\},c})]\prod_{u\in
V(c)\setminus\{v\}}\sum_{s^R_{u,c}}\left([s^R_{u,c}={\tt
F}_c(s^L_{V(c)\setminus\{u\},c})]\cdot\lambda_{u\rightarrow
c}(s^L_{u,c},s^R_{u,c})\right)\\
&\stackrel{(\ref{eq:sum_of_indicator})}{=}& \sum_{s^L_{V(c)\setminus\{v\},c}}[s^R_{v,c}={\tt
F}_c(s^L_{V(c)\setminus\{v\},c})]\prod_{u\in
V(c)\setminus\{v\}}\lambda_{u\rightarrow c}(s^L_{u,c},{\tt
F}_c(s^L_{V(c)\setminus\{u\},c})).
\end{eqnarray*}
}
{\small 
\begin{eqnarray*}
\mu_v(y_v) &=& \sum_{s_{v,C(v)}}\omega_v\left(y_v\left|\bigcap_{c\in
C(v)}s^R_{v,c}\right.\right)\prod_{c\in
C(v)}\left[s^L_{v,c}=y_v\right]\prod_{c\in C(v)}\rho_{c\rightarrow
v}(s^L_{v,c},s^R_{v,c})\\
&=& \sum_{s^R_{v,C(v)}}\omega_v\left(y_v\left|\bigcap_{c\in
C(v)}s^R_{v,c}\right.\right)\sum_{s^L_{v,C(v)}}\prod_{c\in
C(v)}\left([s^L_{v,c}=y_v]\cdot\rho_{c\rightarrow
v}(s^L_{v,c},s^R_{v,c})\right)\\
&\stackrel{(\ref{eq:sum_prod})}{=}&
\sum_{s^R_{v,C(v)}}\omega_v\left(y_v\left|\bigcap_{c\in
C(v)}s^R_{v,c}\right.\right)\prod_{c\in
C(v)}\sum_{s^L_{v,c}}\left([s^L_{v,c}=y_v]\cdot\rho_{c\rightarrow
v}(s^L_{v,c},s^R_{v,c})\right)\\
&\stackrel{(\ref{eq:sum_of_indicator})}{=}&
\sum_{s^R_{v,C(v)}}\omega_v\left(y_v\left|\bigcap_{c\in
C(v)}s^R_{v,c}\right.\right)\prod_{c\in C(v)}\rho_{c\rightarrow
v}(y_v,s^R_{v,c}).
\end{eqnarray*}
}
\end{proof}

\subsection{Weighted PTP as BP for $k$-SAT}

Now we show that for $k$-SAT problems, weighted PTP {\em is} an instance
of BP when the parametrization of weighted PTP is consistent with
the parametrization of the normally realized MRF from which BP is derived.

We begin with introducing a simplification of notations. For any $(v-c)$ and
edge label $L_{v,c}$, we will write
${\bf L_{v,c}}$ as ${\bf L}$, and
${\bf\bar{L}_{v,c}}$ as ${\bf{\bar L}}$. This
suppression of the dependency of ${\bf L_{v,c}}$
and ${\bf\bar{L}_{v,c}}$ on their subscripts
should not result in any ambiguity, when the context clearly
indicates the subscript $(v,c)$ or the edge to which the edge label
$L_{v,c}$ refers.  Additionally, for any $v\in V$, we will write
${\bf 01}_v$ as $*$. Thus, each left or right state will take
configurations from set $\{{\bf L}, {\bf \bar{L}}, *\}$, where the
interpretation of ${\bf L}$ and ${\bf \bar{L}}$ depends on the edge
with which the state is associated. For any given configuration of a
state $(s^L_{v,c}, s^R_{v,c})$, we will suppress the comma between
the left-state configuration and the right-state configuration. For
example, state configurations $({\bf L}, *)$, $({\bf \bar{L}}, *)$,
$(*, *)$  and $({\bf L}, {\bf L})$ will be written respectively as
${\bf L}*, {\bf \bar{L}}*, **$ and ${\bf LL}$.

\begin{lem}
\label{lem:PMF} Let $F$ be defined via (\ref{equ:leftfunction}),
(\ref{equ:rightfunction}) and ($\ref{equ:globalfunction}$), where
each weighting function $\omega_ v$ is defined in
(\ref{eq:omega4ksat}). If $(y_V, s_{V, C})$ is valid under $F$, then
\begin{enumerate}
\item
for every $(v-c)$, it holds that
$s^R_{v,c}\neq {\bf\bar{L}}$, $s_{v, c}\neq {\bf \bar{L}L}$ and that
$s_{v, c}\neq *{\bf L}$, and
\item  $F(y_V, s_{V, C})=\gamma^{n_{*|*}(y_V, s_{V, C})}
\cdot (1-\gamma)^{n_{\cdot|*}(y_V, s_{V, C})}$, where $n_{*|*}(y_V,
s_{V, C})$ and $n_{\cdot|*}(y_V, s_{V, C})$ are respectively the
cardinalities of set $\{v\in V: y_v=\bigcap_{c\in
C(v)}s^R_{v,c}=*\}$ and set $\{v\in V: y_v\subset\bigcap_{c\in
C(v)}s^R_{v,c}=*\}$.
\end{enumerate}
\end{lem}

\begin{proof}
For part 1, first we observe
that $s^R_{v,c}\neq {\bf\bar{L}}$, directly following from the
definition of the right function (\ref{equ:rightfunction}).
Then by Lemma \ref{lem:s_y_relation}, it
is easy to see that $s_{v,c}\neq {\bf\bar{L}}{\bf L}$ and that
$s_{v,c}\neq *{\bf L}$.

For part 2, we may proceed as follows.
\begin{eqnarray*}
F(y_V,s_{V,C}) &=& \prod_{v\in V}g_v(y_v,s_{v,C(v)})\cdot\prod_{c\in
C}f_c(s_{V(c),c})\\
&\stackrel{(\ref{equ:leftfunction}),(\ref{equ:rightfunction})}{=}&
\prod_{v\in V}\left(\omega_v\left(y_v\left|\bigcap_{c\in
C(v)}s^R_{v,c}\right.\right)\cdot\prod_{c\in
C(v)}[s^L_{v,c}=y_v]\right)\cdot\prod_{c\in C}\prod_{v\in
V(c)}[s^R_{v,c}={\tt F}_c(s^L_{V(c)\setminus\{v\},c})]\\
&\stackrel{(a)}{=}& \prod_{v\in
V}\omega_v\left(y_v\left|\bigcap_{c\in
C(v)}s^R_{v,c}\right.\right)\\
&\stackrel{(b)}{=}&
\gamma^{n_{*|*}(y_V,s_{V,C})}\cdot (1-\gamma)^{n_{\cdot|*}(y_V,s_{V,C})},
\end{eqnarray*}
where equality $(a)$ is due to the fact that $(y_V,s_{V,C})$ is valid
under $F$, and equality $(b)$
follows from the definition of the weighting function
$\omega\left(y_v\left|\bigcap_{c\in C(v)}s^R_{v,c}\right.\right)$
in (\ref{eq:omega4ksat}).
\end{proof}

The second part of this lemma, as a slight digression, suggests that
the PMF under this MRF model is identical to that of
\cite{SP:newlook}, since an equivalent result is shown for the MRF
in \cite{SP:newlook}. We note that the MRF in \cite{SP:newlook}
serves as a combinatorial framework for the study of $k$-SAT
problems, which leads to further insights of SP for $k$-SAT problems
(the reader is referred to \cite{SP:newlook} for additional
results). To a certain extent, one may expect that the
normally realized MRF presented here may serve similar purposes for
general CSPs.

 The first part of this lemma suggests that although each state
takes on values from $\{{\bf L}, {\bf \bar{L}}, *\}\times \{{\bf L},
{\bf \bar{L}}, *\}$, there are in fact only four possible state
configurations that contribute to defining a valid
rectangle. When applying the BP message-update rule on the Forney
graph representation of a $k$-SAT problem, this
implies that messages
$\lambda_{v\rightarrow c}$, $\rho_{c\rightarrow v}$ and $\mu_v$
are all supported by $\{{\bf LL, \bar{L}*, L*, **}\}$.

The BP message-update rule is given in Lemma \ref{lem:bp_mssg}, which directly
follows from equations (\ref{equ:bp_l_mssg}) to
(\ref{equ:bp_s_mssg}).

\begin{lem}
\label{lem:bp_mssg} The BP message-update rule applied on
Forney graph $\tilde{G}$ of a $k$-SAT problem gives rise to:

\begin{eqnarray}
\label{equ:ksat_l_LL}\lambda_{v\rightarrow c}({\bf LL}) &:=&
\prod_{b\in C_{c}^{\tt u}(v)}\rho_{b\rightarrow v}({\bf \bar{L}}*)
\prod_{b\in C_{c}^{\tt s}(v)}(\rho_{b\rightarrow v}({\bf
LL})+ \rho_{b\rightarrow v}({\bf L}*)) \\
\label{equ:ksat_l_L_bar_*}\lambda_{v\rightarrow c}({\bf \bar{L}}*)
&:=& \prod_{b\in C_{c}^{\tt s}(v)}\rho_{b\rightarrow v}
({\bf\bar{L}}*)\left(\prod_{b\in C_{c}^{\tt
u}(v)}\left(\rho_{b\rightarrow v}({\bf L}*)+ \rho_{b\rightarrow
v}({\bf LL})\right)-\gamma\prod_{b\in
C_{c}^{\tt u}(v)}\rho_{b\rightarrow v}({\bf L}*)\right) \\
\label{equ:ksat_l_L*}\lambda_{v\rightarrow c}({\bf L}*) &:=&
\prod_{b\in C_{c}^{\tt u}(v)}\rho_{b\rightarrow
v}({\bf\bar{L}}*)\left(\prod_{b\in C_{c}^{\tt
s}(v)}(\rho_{b\rightarrow v}({\bf L}*)+\rho_{b\rightarrow v}({\bf
LL}))-\gamma\prod_{b\in C_{c}^{\tt s}(v)}\rho_{b\rightarrow v}({\bf
L}*)
\right) \\
\label{equ:ksat_l_**}\lambda_{v\rightarrow c}(**) &:=&
\gamma\prod_{b\in C_{c}^{\tt u}(v)\cup C_{c}^{\tt
s}(v)}\rho_{b\rightarrow v}(**)\\
\label{equ:ksat_r_LL}\rho_{c\rightarrow v}({\bf LL}) &:=&
\prod_{u\in V(c) \setminus
\{v\}}\lambda_{u\rightarrow c}({\bf \bar{L}}*)\\
\label{equ:ksat_r_L_bar_*}\rho_{c\rightarrow v} ({\bf \bar{L}}*)
&:=& \prod_{u\in V(c) \setminus \{v\}} (\lambda_{u\rightarrow
c}({\bf L}*)+\lambda_{u\rightarrow c}(**)
+\lambda_{u\rightarrow c}({\bf \bar{L}}*))\nonumber\\
&&+\sum_{u\in V(c)\setminus \{v\}} (\lambda_{u\rightarrow c}({\bf
LL})-\lambda_{u\rightarrow c}({\bf L}*) - \lambda_{u\rightarrow
c}(**))\prod_{w\in V(c)\setminus \{u,
v\}}\lambda_{w\rightarrow c}({\bf \bar{L}}*)\nonumber\\
&& -\prod_{u\in V(c)\setminus \{v\}} \lambda_{u\rightarrow c}({\bf
\bar{L}}*)\\
\label{equ:ksat_r_L*}\rho_{c\rightarrow v}({\bf L}*) &:=&
\prod_{u\in V(c) \setminus \{v\}} (\lambda_{u\rightarrow c} ({\bf
L}*)+\lambda_{u\rightarrow c}(**)+\lambda_{u\rightarrow c}({\bf
\bar{L}}*)) -\prod_{u\in V(c)\setminus \{v\}}\lambda_{u\rightarrow
c}
({\bf \bar{L}}*)\\
\label{equ:ksat_r_**}\rho_{c\rightarrow v}(**) &:=& \prod_{u\in V(c)
\setminus \{v\}} (\lambda_{u\rightarrow c}({\bf
L}*)+\lambda_{u\rightarrow c}(**)+\lambda_{u\rightarrow c}({\bf
\bar{L}}*)) -\prod_{u\in V(c)\setminus \{v\}}
\lambda_{u\rightarrow c}({\bf \bar{L}}*)\\
\label{equ:ksat_s_0}\mu_v({\bf 0}) &:=& \prod_{c\in
C^1(v)}\rho_{c\rightarrow v}({\bf \bar{L}}*)\left(\prod_{c\in
C^0(v)}(\rho_{c\rightarrow v}({\bf LL})+\rho_{c\rightarrow v}({\bf
L}*))-\gamma\prod_{c\in
C^0(v)}\rho_{c\rightarrow v}({\bf L}*)\right)\\
\label{equ:ksat_s_1}\mu_v({\bf 1}) &:=& \prod_{c\in
C^0(v)}\rho_{c\rightarrow v}({\bf \bar{L}}*)\left(\prod_{c\in
C^1(v)}(\rho_{c\rightarrow v}({\bf LL})+\rho_{c\rightarrow v}({\bf
L}*))-\gamma\prod_{c\in
C^1(v)}\rho_{c\rightarrow v}({\bf L}*)\right)\\
\label{equ:ksat_s_*}\mu_v(*) &:=& \gamma\prod_{c\in
C(v)}\rho_{c\rightarrow v}(**).
\end{eqnarray}
\end{lem}

Now we are ready to investigate how these BP messages may reduced to (weighted)
PTP messages. It turns out that the following condition has a special role
to play in this reduction.

\begin{equation}
\label{eq:reduce_ksat_initial_condition}
 \rho_{c\rightarrow v}({\bf
L*})=\rho_{c\rightarrow v}({\bf {\bar L}*})=\rho_{c\rightarrow
v}({\bf **})
\end{equation}

\begin{prop}
\label{prop:initial_condition_ksat}
In $k$-SAT problems, if the BP messages are initialized to satisfy
condition (\ref{eq:reduce_ksat_initial_condition}), then this condition is satisfied in every BP iteration.
\end{prop}
\begin{proof}
We only need to show that if (\ref{eq:reduce_ksat_initial_condition})
is satisfied during initialization, then it is satisfied in the first iteration
after initialization. -- In fact, noting that $\rho_{c\rightarrow v}({\bf{L*}})=\rho_{c\rightarrow v}(**)$ necessarily holds in each BP iteration
due to  ({\ref{equ:ksat_r_L*}}) and ({\ref{equ:ksat_r_**}}), we only need
to prove that
$\rho_{c\rightarrow v}({\bf
\bar{L}}*)=\rho_{c\rightarrow v}({\bf L}*)$ holds in the first
iteration provided BP messages are initialized to satisfy
(\ref{eq:reduce_ksat_initial_condition}).

Under this initialization condition, we have,
in the first BP iteration after,
\begin{eqnarray*}
\lambda_{v\rightarrow c}({\bf L}*)+\lambda_{v\rightarrow c}(**) &=& \prod_{b\in C_{c}^{\tt
u}(v)}\rho_{b\rightarrow v} ({\bf\bar{L}}*) \times
\left(\prod_{b\in C_{c}^{\tt s}(v)}\left(\rho_{b\rightarrow v}({\bf L}*)+ \rho_{b\rightarrow
v}({\bf LL})\right)-\gamma\prod_{b\in
C_{c}^{\tt s}(v)}\rho_{b\rightarrow v}({\bf L}*)\right)\\
&& +\gamma\prod_{b\in C_{c}^{\tt u}(v)\cup C_{c}^{\tt
s}(v)}\rho_{b\rightarrow v}(**)\\
&=& \prod_{b\in C_{c}^{\tt u}(v)}\rho_{b\rightarrow v}
({\bf\bar{L}}*)\prod_{b\in C_{c}^{\tt s}(v)}\left(\rho_{b\rightarrow v}({\bf L}*)+ \rho_{b\rightarrow
v}({\bf LL})\right)\\
&&-\gamma\prod_{b\in C_{c}^{\tt u}(v)}\rho_{b\rightarrow
v} ({\bf\bar{L}}*)\prod_{b\in C_{c}^{\tt s}(v)}\rho_{b\rightarrow v}({\bf L}*)+\gamma\prod_{b\in C_{c}^{\tt
u}(v)\cup C_{c}^{\tt
s}(v)}\rho_{b\rightarrow v}(**)\\
&\stackrel{(a)}{=}& \prod_{b\in C_{c}^{\tt u}(v)}\rho_{b\rightarrow v} ({\bf\bar{L}}*)\prod_{b\in C_{c}^{\tt
s}(v)}\left(\rho_{b\rightarrow v}({\bf L}*)+ \rho_{b\rightarrow
v}({\bf LL})\right)\\
&=& \lambda_{v\rightarrow c}({\bf LL}),
\end{eqnarray*}
where equality $(a)$ is due to the initialization
condition (\ref{eq:reduce_ksat_initial_condition}).

Then in the subsequent update of the right messages, we have

\begin{eqnarray*}
\rho_{c\rightarrow v} ({\bf \bar{L}}*) &=& \prod_{u\in
V(c) \setminus \{v\}} (\lambda_{u\rightarrow c}({\bf
L}*)+\lambda_{u\rightarrow c}(**)
+\lambda_{u\rightarrow c}({\bf \bar{L}}*))\\
&&+\sum_{u\in V(c)\setminus \{v\}} (\lambda_{u\rightarrow c}({\bf LL})-\lambda_{u\rightarrow c}({\bf L}*) - \lambda_{u\rightarrow c}(**))\prod_{w\in V(c)\setminus
\{u,v\}}\lambda_{w\rightarrow c}({\bf \bar{L}}*)\\
&& -\prod_{u\in V(c)\setminus \{v\}} \lambda_{u\rightarrow c}({\bf \bar{L}}*)\\
&\stackrel{(b)}{=}& \prod_{u\in V(c) \setminus \{v\}} (\lambda_{u\rightarrow c}({\bf L}*)+\lambda_{u\rightarrow c}(**) +\lambda_{u\rightarrow
c}({\bf \bar{L}}*))-\prod_{u\in V(c)\setminus
\{v\}} \lambda_{u\rightarrow c}({\bf \bar{L}}*)\\
&=& \rho_{c\rightarrow v} ({\bf L}*),
\end{eqnarray*}
where equality $(b)$ is due to  the above result $\lambda_{v\rightarrow c}({\bf LL})=\lambda_{v\rightarrow c}({\bf L}*)+\lambda_{v\rightarrow c}(**)$.

\end{proof}

\begin{thm}
\label{thm:sp_as_bp_ksat} In a $k$-SAT problem,
suppose that the following two conditions are
imposed in the BP messages.
\begin{enumerate}
\item For every $(v-c)$, the BP messages are
initialized such that (\ref{eq:reduce_ksat_initial_condition}) is satisfied.
\item In each BP iteration, $\lambda_{v\rightarrow c}$ is scaled to
$\lambda_{v\rightarrow c}^{\rm norm}$
such that
$\lambda^{\rm
norm}_{v\rightarrow c}({\bf L*})+\lambda^{\rm norm}_{v\rightarrow
c}({\bf \bar{L}*})+\lambda^{\rm norm}_{v\rightarrow c}(**)=1$, before it
is passed along the edge; that is,
$\lambda^{\rm norm}_{v\rightarrow
c}(s^L_{v,c},s^R_{v,c}):=\frac{1}{\sum_{s^L_{v,c}}\lambda_{v\rightarrow
c}(s^L_{v,c}, *)}\cdot \lambda_{v\rightarrow c}(s^L_{v,c},s^R_{v,c})$ for
every $(s^L_{v,c},s^R_{v,c})$  in the support of
$\lambda_{v\rightarrow c}$ and the right messages are updated based on
the normalized left messages, namely,
\begin{eqnarray}
\label{equ:ksat_r_LL_n}\rho_{c\rightarrow v}({\bf LL}) &:=&
\prod_{u\in V(c) \setminus
\{v\}}\lambda^{\rm norm}_{u\rightarrow c}({\bf \bar{L}}*)\\
\label{equ:ksat_r_L_bar_*_n}\rho_{c\rightarrow v} ({\bf \bar{L}}*)
&:=& \prod_{u\in V(c) \setminus \{v\}} (\lambda^{\rm
norm}_{u\rightarrow c}({\bf L}*)+\lambda^{\rm norm}_{u\rightarrow
c}(**)
+\lambda^{\rm norm}_{u\rightarrow c}({\bf \bar{L}}*))\nonumber\\
&&+\sum_{u\in V(c)\setminus \{v\}} (\lambda^{\rm norm}_{u\rightarrow
c}({\bf LL})-\lambda^{\rm norm}_{u\rightarrow c}({\bf L}*) -
\lambda^{\rm norm}_{u\rightarrow c}(**))\prod_{w\in V(c)\setminus
\{u,
v\}}\lambda^{\rm norm}_{w\rightarrow c}({\bf \bar{L}}*)\nonumber\\
&& -\prod_{u\in V(c)\setminus \{v\}} \lambda^{\rm
norm}_{u\rightarrow c}({\bf
\bar{L}}*)\\
\label{equ:ksat_r_L*_n}\rho_{c\rightarrow v}({\bf L}*) &:=&
\!\!\!\!\!\!\prod_{u\in V(c) \setminus \{v\}} (\lambda^{\rm
norm}_{u\rightarrow c} ({\bf L}*)+\lambda^{\rm norm}_{u\rightarrow
c}(**)+\lambda^{\rm norm}_{u\rightarrow c}({\bf \bar{L}}*))
-\prod_{u\in V(c)\setminus \{v\}}\lambda^{\rm norm}_{u\rightarrow c}
({\bf \bar{L}}*)\\
\label{equ:ksat_r_**_n}\rho_{c\rightarrow v}(**) &:=&
\!\!\!\!\!\!\prod_{u\in V(c) \setminus \{v\}} (\lambda^{\rm
norm}_{u\rightarrow c}({\bf L}*)+\lambda^{\rm norm}_{u\rightarrow
c}(**)+\lambda^{\rm norm}_{u\rightarrow c}({\bf \bar{L}}*))
-\prod_{u\in V(c)\setminus \{v\}} \lambda^{\rm norm}_{u\rightarrow
c}({\bf \bar{L}}*).
\end{eqnarray}
\end{enumerate}
Then the correspondence between BP messages and weighted PTP
messages is
\begin{eqnarray}
\label{equ:bp_l_L*_ptp}\lambda^{\rm norm(BP)}_{v\rightarrow c}({\bf
L*}) &\leftrightarrow& [L_{v,c}=0]\cdot\lambda^{\rm
norm(PTP)}_{v\rightarrow c}({\bf 0}) +
[L_{v,c}=1]\cdot\lambda^{\rm norm(PTP)}_{v\rightarrow c}({\bf 1})\\
\label{equ:bp_l_L_bar_*_ptp}\lambda^{\rm norm(BP)}_{v\rightarrow
c}({\bf \bar{L}*}) &\leftrightarrow& [L_{v,c}=0]\cdot\lambda^{\rm
norm(PTP)}_{v\rightarrow c}({\bf 1}) +
[L_{v,c}=1]\cdot\lambda^{\rm norm(PTP)}_{v\rightarrow c}({\bf 0})\\
\label{equ:bp_l_**_ptp}\lambda^{\rm norm(BP)}_{v\rightarrow c}(**)
&\leftrightarrow& \lambda^{\rm norm(PTP)}_{v\rightarrow c}(*)\\
\label{equ:bp_r_L*_ptp}\rho^{(\rm BP)}_{c\rightarrow v}({\bf L}*)
&\leftrightarrow& \rho^{\rm norm(PTP)}_{c\rightarrow v}(*)\\
\label{equ:bp_r_LL_ptp}\rho^{(\rm BP)}_{c\rightarrow v}({\bf LL})
&\leftrightarrow& \rho^{\rm norm(PTP)}_{c\rightarrow v}({\bf 0}) +
\rho^{\rm norm(PTP)}_{c\rightarrow v}({\bf 1})\\
\label{equ:bp_s_0_ptp}\mu^{(\rm BP)}_v({\bf 0}) &\leftrightarrow&
\mu^{(\rm PTP)}_v({\bf 0})\\
\label{equ:bp_s_1_ptp}\mu^{(\rm BP)}_v({\bf 1}) &\leftrightarrow&
\mu^{(\rm PTP)}_v({\bf 1})\\
\label{equ:bp_s_01_ptp}\mu^{(\rm BP)}_v(*) &\leftrightarrow&
\mu^{(\rm PTP)}_v(*).
\end{eqnarray}
\end{thm}

\begin{proof}

Note that based on Proposition \ref{prop:initial_condition_ksat}, condition $\rho^{(\rm
BP)}_{c\rightarrow v}({\bf L}*)=\rho^{(\rm BP)}_{c\rightarrow
v}({\bf \bar{L}}*)=\rho^{(\rm BP)}_{c\rightarrow v}(**)$ holds in
every BP iteration.
From the proof of Proposition \ref{prop:initial_condition_ksat}, it also holds
in every BP iteration that
\begin{equation}
\label{eq:B}
\lambda^{{\rm norm (BP)}}_{v\rightarrow c}({\bf L*})+\lambda^{\rm norm
(BP)}_{v\rightarrow c}(**)=\lambda^{\rm norm (BP)}_{v\rightarrow c}({\bf
LL}).
\end{equation}

Now we will prove this theorem by first proving that
the ``left correspondence''
((\ref{equ:bp_l_L*_ptp}) to (\ref{equ:bp_l_**_ptp}))
implies the ``right correspondence''
((\ref{equ:bp_r_L*_ptp}) and (\ref{equ:bp_r_LL_ptp}))
and conversely that the ``right correspondence'' implies
the ``left correspondence'', whereby proving the correspondence in the passed
messages. We then prove the summary correspondence
((\ref{equ:bp_s_0_ptp}) to (\ref{equ:bp_s_01_ptp})).

First suppose that left correspondence holds, namely that
$\lambda^{\rm norm(BP)}_{v\rightarrow c}({\bf
L*})=[L_{v,c}=0]\cdot\lambda^{\rm norm(PTP)}_{v\rightarrow c}({\bf
0}) + [L_{v,c}=1]\cdot\lambda^{\rm norm(PTP)}_{v\rightarrow c}({\bf
1})$, $\lambda^{\rm norm(BP)}_{v\rightarrow c}({\bf
\bar{L}*})=[L_{v,c}=0]\cdot\lambda^{\rm norm(PTP)}_{v\rightarrow
c}({\bf 1}) + [L_{v,c}=1]\cdot\lambda^{\rm norm(PTP)}_{v\rightarrow
c}({\bf 0})$, and $\lambda^{\rm norm(BP)}_{v\rightarrow c}(**)=
\lambda^{\rm norm(PTP)}_{v\rightarrow c}(*)$. Following PTP
message-updating equations (\ref{equ:ptp_sat_r_0}) to
(\ref{equ:ptp_sat_r_*}), we have

\begin{eqnarray*}
\rho^{\rm norm(PTP)}_{c\rightarrow v}({\bf 0}) &+& \rho^{\rm
norm(PTP)}_{c\rightarrow v}({\bf 1}) \stackrel{(a)}{=} \rho^{\rm
(PTP)}_{c\rightarrow v}({\bf 0}) + \rho^{\rm (PTP)}_{c\rightarrow
v}({\bf 1})\\
&=& [L_{v,c}=0]\cdot\prod\limits_{u\in V(c)\setminus
\{v\}:L_{u,c}=1}\lambda^{\rm norm(PTP)}_{u\rightarrow c}({\bf
0})\prod\limits_{u\in V(c)\setminus\{v\}:L_{u,c}=0}\lambda^{\rm
norm(PTP)}_{u\rightarrow c}({\bf 1})\\
&& +[L_{v,c}=1]\cdot\prod\limits_{u\in V(c)\setminus
\{v\}:L_{u,c}=1}\lambda^{\rm norm(PTP)}_{u\rightarrow c}({\bf
0})\prod\limits_{u\in V(c)\setminus\{v\}:L_{u,c}=0}
\lambda^{\rm norm(PTP)}_{u\rightarrow c}({\bf 1})\\
&=& \prod\limits_{u\in V(c)\setminus \{v\}:L_{u,c}=1}\lambda^{\rm
norm(PTP)}_{u\rightarrow c}({\bf 0})\prod\limits_{u\in
V(c)\setminus\{v\}:L_{u,c}=0} \lambda^{\rm norm(PTP)}_{u\rightarrow
c}({\bf 1})\\
&=& \prod\limits_{u\in V(c)\setminus \{v\}:L_{u,c}=1}\left(
[L_{u,c}=0]\cdot\lambda^{\rm norm(PTP)}_{u\rightarrow c}({\bf
1})+[L_{u,c}=1]\cdot\lambda^{\rm norm(PTP)}_{u\rightarrow c}({\bf
0})\right)\\
&& \times\prod\limits_{u\in V(c)\setminus \{v\}:L_{u,c}=0}\left(
[L_{u,c}=0]\cdot\lambda^{\rm norm(PTP)}_{u\rightarrow c}({\bf
1})+[L_{u,c}=1]\cdot\lambda^{\rm norm(PTP)}_{u\rightarrow c}({\bf
0})\right)\\
&=& \prod\limits_{u\in V(c)\setminus \{v\}}\left(
[L_{u,c}=0]\cdot\lambda^{\rm norm(PTP)}_{u\rightarrow c}({\bf
1})+[L_{u,c}=1]\cdot\lambda^{\rm norm(PTP)}_{u\rightarrow c}({\bf
0})\right)\\
&\stackrel{(b)}{=}& \prod\limits_{u\in V(c)\setminus
\{v\}}\lambda^{\rm norm(BP)}_{u\rightarrow c}({\bf \bar{L}*})\\
&=& \rho^{(\rm BP)}_{c\rightarrow v}({\bf LL})
\end{eqnarray*}
where equality $(a)$ is due to the fact that $\rho^{\rm
norm(PTP)}_{c\rightarrow v}=\rho^{(\rm PTP)}_{c\rightarrow v}$ as is
shown in the proof of Theorem \ref{thm:ptp_sp_equiv_ksat}, equality
$(b)$ is due to the assumed left correspondence.

Similarly, we have
\begin{eqnarray*}
\rho^{\rm norm(PTP)}_{c\rightarrow v}(*) &=& \rho^{\rm
(PTP)}_{c\rightarrow v}(*)\\
&=& 1-\prod\limits_{u\in V(c)\setminus
\{v\}:{L_{u,c}=1}}\lambda^{\rm norm(PTP)}_{u\rightarrow c}({\bf
0})\prod\limits_{u\in V(c)\setminus
\{v\}:{L_{u,c}=0}}\lambda^{\rm norm(PTP)}_{u\rightarrow c}({\bf 1})\\
&=& 1-\prod\limits_{u\in V(c)\setminus \{v\}}\lambda^{\rm
norm (BP)}_{u\rightarrow c}({\bf \bar{L}*})\\
&\stackrel{(c)}{=}& \prod\limits_{u\in V(c)\setminus \{v\}}
(\lambda^{\rm norm(BP)}_{u\rightarrow c}({\bf L*})+\lambda^{\rm
norm(BP)}_{u\rightarrow c}({\bf \bar{L}*})+\lambda^{\rm
norm(BP)}_{u\rightarrow c}(**))- \prod\limits_{u\in V(c)\setminus
\{v\}}\lambda^{\rm norm(BP)}_{u\rightarrow c}({\bf \bar{L}*})\\
&=& \rho^{(\rm BP)}_{c\rightarrow v}({\bf L*}),
\end{eqnarray*}
where equality $(c)$ is due to the fact that $\lambda^{\rm
norm(BP)}_{u\rightarrow c}({\bf L*})+\lambda^{\rm
norm(BP)}_{u\rightarrow c}({\bf \bar{L}*})+\lambda^{\rm
norm(BP)}_{u\rightarrow c}(**)=1$.

Thus we proved that if the left correspondence holds, then the right
correspondence holds.

Now suppose that the right correspondence holds, namely that $\rho^{(\rm
BP)}_{c\rightarrow v}({\bf L}*)=\rho^{\rm norm(PTP)}_{c\rightarrow
v}(*)$, and $\rho^{(\rm BP)}_{c\rightarrow v}({\bf LL})=
\rho^{\rm norm(PTP)}_{c\rightarrow v}({\bf 0}) + \rho^{\rm
norm(PTP)}_{c\rightarrow v}({\bf 1})$. We then have
\begin{eqnarray*}
\rho^{(\rm BP)}_{c\rightarrow v}({\bf L}*)+\rho^{(\rm
BP)}_{c\rightarrow v}({\bf LL}) &=& \rho^{\rm
norm(PTP)}_{c\rightarrow v}(*)+ \rho^{\rm
norm(PTP)}_{c\rightarrow v}({\bf 0}) + \rho^{\rm
norm(PTP)}_{c\rightarrow v}({\bf 1})\\
&=& 1.
\end{eqnarray*}

Following PTP message-update equations (\ref{equ:ptp_sat_l_0}) to
(\ref{equ:ptp_sat_l_*}), we have

\begin{eqnarray*}
&&[L_{v,c}=0]\cdot\lambda^{\rm (PTP)}_{v\rightarrow c}({\bf 0})
+ [L_{v,c}=1]\cdot\lambda^{\rm (PTP)}_{v\rightarrow c}({\bf 1})\\
&=& [L_{v,c}=0]\cdot\left(\prod\limits_{b\in C(v)\setminus \{c\}}
(\rho^{\rm norm(PTP)}_{b\rightarrow v}({\bf 0})+\rho^{\rm
norm(PTP)}_{b\rightarrow v}(*))-\gamma\prod\limits_{b\in
C(v)\setminus\{c\}}\rho^{\rm norm(PTP)}_{b\rightarrow v}(*)\right)\\
&& +[L_{v,c}=1]\cdot\left(\prod\limits_{b\in C(v)\setminus \{c\}}
(\rho^{\rm norm(PTP)}_{b\rightarrow v}({\bf 1})+\rho^{\rm
norm(PTP)}_{b\rightarrow v}(*))-\gamma\prod\limits_{b\in
C(v)\setminus\{c\}}\rho^{\rm norm(PTP)}_{b\rightarrow v}(*)\right)\\
&=&[L_{v,c}=0]\cdot\prod\limits_{b\in C(v)\setminus \{c\}}
(\rho^{\rm norm(PTP)}_{b\rightarrow v}({\bf 0})+\rho^{\rm
norm(PTP)}_{b\rightarrow v}(*))\\
&& +[L_{v,c}=1]\cdot\prod\limits_{b\in C(v)\setminus \{c\}}
(\rho^{\rm norm(PTP)}_{b\rightarrow v}({\bf 1})+\rho^{\rm
norm(PTP)}_{b\rightarrow v}(*))-\gamma\prod\limits_{b\in
C(v)\setminus\{c\}}\rho^{\rm
norm(PTP)}_{b\rightarrow v}(*)\\
&\stackrel{(\ref{eq:simple_lem_b})}{=}& [L_{v,c}=0]\cdot\prod_{b\in
C_{c}^{\tt s}(v)}(\rho^{\rm norm(PTP)}_{b\rightarrow v}({\bf
0})+\rho^{\rm norm(PTP)}_{b\rightarrow v}(*))\cdot\prod_{b\in
C_{c}^{\tt
u}(v)}\rho^{\rm norm(PTP)}_{b\rightarrow v}(*)\\
&& + [L_{v,c}=1]\cdot\prod_{b\in C_c^{\tt s}(v)}(\rho^{\rm
norm(PTP)}_{b\rightarrow v}({\bf 1})+\rho^{\rm
norm(PTP)}_{b\rightarrow v}(*))\cdot\prod_{b\in C_c^{\tt
u}(v)}\rho^{\rm norm(PTP)}_{b\rightarrow v}(*)\\
&& -\gamma\prod_{b\in
C(v)\setminus\{c\}}\rho^{\rm norm(PTP)}_{b\rightarrow v}(*)\\
&\stackrel{(\ref{eq:simple_lem_a})}{=}& [L_{v,c}=0]\cdot\prod_{b\in
C_c^{\tt u}(v)}\rho^{\rm norm(PTP)}_{b\rightarrow v}(*) +
[L_{v,c}=1]\cdot\prod_{b\in C_c^{\tt u}(v)}\rho^{\rm
norm(PTP)}_{b\rightarrow v}(*) - \gamma\prod_{b\in
C(v)\setminus\{c\}}\rho^{\rm norm(PTP)}_{b\rightarrow v}(*)\\
&=& \prod_{b\in C_c^{\tt u}(v)}\rho^{\rm norm(PTP)}_{b\rightarrow
v}(*) -
\gamma\prod_{b\in C(v)\setminus\{c\}}\rho^{\rm norm(PTP)}_{b\rightarrow v}(*)\\
&=& \prod_{b\in C_c^{\tt u}(v)}\rho^{\rm norm(PTP)}_{b\rightarrow
v}(*)\left(1-\gamma\prod_{b\in C_c^{\tt s}(v)}\rho^{\rm
norm(PTP)}_{b\rightarrow
v}(*)\right)\\
&\stackrel{(d)}{=}& \prod_{b\in C_c^{\tt u}(v)}\rho^{\rm
(BP)}_{b\rightarrow v}({\bf L*})\left(1-\gamma\prod_{b\in
C_c^{\tt s}(v)}\rho^{\rm (BP)}_{b\rightarrow v}({\bf L*})\right)\\
&\stackrel{(e)}{=}&\prod_{b\in C_c^{\tt u}(v)}\rho^{\rm
(BP)}_{b\rightarrow v}({\bf L*})\left(\prod_{b\in C_c^{\tt
s}(v)}(\rho^{\rm (BP)}_{b\rightarrow v}({\bf L*})+\rho^{\rm
(BP)}_{b\rightarrow v}({\bf LL}))-\gamma\prod_{b\in
C_c^{\tt s}(v)}\rho^{\rm (BP)}_{b\rightarrow v}({\bf L*})\right)\\
&\stackrel{(f)}{=}&\prod_{b\in C_c^{\tt u}(v)}\rho^{\rm
(BP)}_{b\rightarrow v}({\bf \bar{L}*})\left(\prod_{b\in C_c^{\tt
s}(v)}(\rho^{\rm (BP)}_{b\rightarrow v}({\bf L*})+\rho^{\rm
(BP)}_{b\rightarrow v}({\bf LL}))-\gamma\prod_{b\in
C_c^{\tt s}(v)}\rho^{\rm (BP)}_{b\rightarrow v}({\bf L*})\right)\\
&=& \lambda^{\rm (BP)}_{v\rightarrow c}({\bf L*})
\end{eqnarray*}
where equality $(d)$ is due to the assumed right correspondence,
equality $(e)$ is due to the fact that $\rho^{(\rm
BP)}_{c\rightarrow v}({\bf L}*)+\rho^{(\rm BP)}_{c\rightarrow
v}({\bf LL})=1$, and equality $(f)$ is due to that the condition
$\rho^{\rm (BP)}_{b\rightarrow v}({\bf L*})=\rho^{\rm
(BP)}_{b\rightarrow v}({\bf \bar{L}*})$ is satisfied in every iteration.
We will denote this result
by $(A)$.

Similarly, we have
\begin{eqnarray*}
&&[L_{v,c}=0]\cdot\lambda^{\rm (PTP)}_{v\rightarrow c}({\bf 1})
+ [L_{v,c}=1]\cdot\lambda^{\rm (PTP)}_{v\rightarrow c}({\bf 0})\\
&=& [L_{v,c}=0]\cdot\left(\prod\limits_{b\in C(v)\setminus \{c\}}
(\rho^{\rm norm(PTP)}_{b\rightarrow v}({\bf 1})+\rho^{\rm
norm(PTP)}_{b\rightarrow v}(*))-\gamma\prod\limits_{b\in
C(v)\setminus\{c\}}\rho^{\rm norm(PTP)}_{b\rightarrow v}(*)\right)\\
&& +[L_{v,c}=1]\cdot\left(\prod\limits_{b\in C(v)\setminus \{c\}}
(\rho^{\rm norm(PTP)}_{b\rightarrow v}({\bf 0})+\rho^{\rm
norm(PTP)}_{b\rightarrow v}(*))-\gamma\prod\limits_{b\in
C(v)\setminus\{c\}}\rho^{\rm norm(PTP)}_{b\rightarrow v}(*)\right)\\
&=& \lambda^{\rm (BP)}_{v\rightarrow c}({\bf \bar{L}*}).
\end{eqnarray*}
We will denote this result by $(B)$.

Finally, we have
\begin{eqnarray*}
\lambda^{\rm(PTP)}_{v\rightarrow c}(*) &=&
\gamma\prod\limits_{b\in C(v)\setminus\{c\}} \rho^{\rm
norm(PTP)}_{b\rightarrow v}(*)\\
&=& \gamma\prod\limits_{b\in C(v)\setminus\{c\}} \rho^{\rm
(BP)}_{b\rightarrow v}(**)\\
&=& \lambda^{\rm (BP)}_{v\rightarrow c}(**).
\end{eqnarray*}
We will denote this result by $(C)$.

Combining results of $(A)$, $(B)$ and $(C)$, we have
\[
\lambda^{\rm (PTP)}_{v\rightarrow c}({\bf 0})+\lambda^{\rm
(PTP)}_{v\rightarrow c}({\bf 1})+\lambda^{\rm (PTP)}_{v\rightarrow
c}(*)=\lambda^{\rm (BP)}_{v\rightarrow c}({\bf
L*})+\lambda^{\rm (BP)}_{v\rightarrow c}({\bf
\bar{L}*})+\lambda^{\rm (BP)}_{v\rightarrow c}(**).
\]
That is, the scaling constant for normalizing $(\lambda^{\rm
(PTP)}_{v\rightarrow c}({\bf 0}),\lambda^{\rm (PTP)}_{v\rightarrow
c}({\bf 1}),\lambda^{\rm (PTP)}_{v\rightarrow c}(*))$ and
that for normalizing $(\lambda^{\rm (BP)}_{v\rightarrow c}({\bf
L*}),\lambda^{\rm (BP)}_{v\rightarrow c}({\bf
\bar{L}*}),\lambda^{\rm (BP)}_{v\rightarrow c}(**))$ are identical.
Therefore, result $(A)$, $(B)$ and $(C)$ respectively translate to
\begin{eqnarray*}
{[L_{v,c}=0]\cdot\lambda^{\rm norm(PTP)}_{v\rightarrow c}({\bf 0}) +
[L_{v,c}=1]\cdot\lambda^{\rm norm(PTP)}_{v\rightarrow c}({\bf 1})}
&=& \lambda^{\rm norm(BP)}_{v\rightarrow c}({\bf L*})\\
{[L_{v,c}=0]\cdot\lambda^{\rm norm(PTP)}_{v\rightarrow c}({\bf 1}) +
[L_{v,c}=1]\cdot\lambda^{\rm norm (PTP)}_{v\rightarrow c}({\bf 0})}
&=& \lambda^{\rm norm(BP)}_{v\rightarrow c}({\bf \bar{L}*})\\
\lambda^{\rm norm(PTP)}_{v\rightarrow c}(*) &=& \lambda^{\rm
norm( BP)}_{v\rightarrow c}(**).
\end{eqnarray*}

At this point we have proved the correspondence between the passed
messages in BP and those in weighted PTP.

We now prove the summary correspondence. Following the PTP message-update equations (\ref{equ:ptp_sat_s_0}) to
(\ref{equ:ptp_sat_s_*}), we have
\begin{eqnarray*}
\mu^{(\rm PTP)}_v({\bf 0}) & = & \prod_{c\in C(v)} \left(\rho^{\rm
norm(PTP)}_{c\rightarrow v}({\bf 0})+ \rho^{\rm
norm(PTP)}_{c\rightarrow v}(*)\right) -\gamma
\prod_{c\in C(v)}\rho^{\rm norm(PTP)}_{c\rightarrow v}(*)\\
& = & \prod_{c\in C^{\rm 1}(v)} \left(\rho^{\rm
norm(PTP)}_{c\rightarrow v}({\bf 0})+ \rho^{\rm
norm(PTP)}_{c\rightarrow v}(*)\right) \prod_{c\in C^{\rm
0}(v)} \left(\rho^{\rm norm(PTP)}_{c\rightarrow v}({\bf 0})+
\rho^{\rm norm(PTP)}_{c\rightarrow v}(*)\right)\\
&& -\gamma \prod_{c\in C(v)}\rho^{\rm norm(PTP)}_{c\rightarrow
v}(*)\\
&\stackrel{(\ref{eq:simple_lem_a}),(\ref{eq:simple_lem_b})}{=} &
\prod_{c\in C^{\rm 1}(v)} \rho^{\rm norm(PTP)}_{c\rightarrow v}(*) -\gamma\prod_{c\in C(v)}\rho^{\rm norm(PTP)}_{c\rightarrow
v}(*)\\
& = &\prod_{c\in C^{\rm 1}(v)}\rho^{\rm norm(PTP)}_{c\rightarrow
v}(*)\left(1-\gamma \prod_{c\in C^{\rm 0}(v)}\rho^{\rm
norm(PTP)}_{c\rightarrow v}(*)\right)\\
&=&\prod_{c\in C^{\rm 1}(v)}\rho^{\rm(BP)}_{c\rightarrow v}({\bf
\bar{L}*})\left(\prod_{c\in C^{\rm 0}(v)}(\rho^{\rm
(BP)}_{c\rightarrow v}({\bf LL})+\rho^{\rm (BP)}_{c\rightarrow
v}({\bf L*})-\gamma \prod_{c\in C^{\rm 0}(v)}\rho^{\rm
(BP)}_{c\rightarrow
v}({\bf L*})\right)\\
&=& \mu^{(\rm BP)}_v({\bf 0}).
\end{eqnarray*}

Following a similar procedure, we have
\begin{eqnarray*}
\mu_v^{\rm (PTP)}({\bf 1}) &=& \prod\limits_{c\in C(v)}(\rho^{\rm
norm(PTP)}_{c\rightarrow v}({\bf 1})+\rho^{\rm
norm(PTP)}_{c\rightarrow v}(*))-\gamma\prod\limits_{c\in
C(v)}\rho^{\rm norm(PTP)}_{c\rightarrow v}(*)\\
&=&\prod_{c\in C^{\rm 0}(v)}\rho^{\rm(BP)}_{c\rightarrow v}({\bf
\bar{L}*})\left(\prod_{c\in C^{\rm 1}(v)}(\rho^{\rm
(BP)}_{c\rightarrow v}({\bf LL})+\rho^{\rm (BP)}_{c\rightarrow
v}({\bf L*})-\gamma \prod_{c\in C^{\rm 1}(v)}\rho^{\rm
(BP)}_{c\rightarrow
v}({\bf L*})\right)\\
&=& \mu^{(\rm BP)}_v({\bf 1}).
\end{eqnarray*}

Finally, we have
\begin{eqnarray*}
\mu_v^{\rm (PTP)}(*) &=& \gamma\prod\limits_{c\in
C(v)}\rho^{\rm
norm(PTP)}_{c\rightarrow v}(*)\\
&=& \gamma\prod\limits_{c\in C(v)}\rho^{\rm
(BP)}_{c\rightarrow v}(**)\\
&=& \mu^{(\rm BP)}_v(*),
\end{eqnarray*}
which proves the summary correspondence.
\end{proof}

\subsection{State-Decoupled BP}

In this subsection, we will consider reducing PTP from BP
for $3$-COL problems, where we only focus on the non-weighted version of PTP,
namely that each weighting function $\omega_v$ is defined as
\begin{equation}
\label{equ:sd_bp_3col_weight}
\omega_v(a|b):=[a=b].
\end{equation}
This gives the form of BP messages in the form specified in the following lemma, easily obtainable from BP update equations
(\ref{equ:bp_l_mssg}) to (\ref{equ:bp_s_mssg}).

\begin{lem}
\label{lem:BP_3COL}
The BP message-update rule for 3-COL problems is as follow:

\begin{eqnarray}
\label{equ:bp_3col_l_i_ij}\lambda_{v\rightarrow c}({\bf i},{\bf ij})
&:=& \prod_{b\in C(v)\setminus\{c\}}(\rho_{b\rightarrow v}({\bf
i},{\bf ij})+\rho_{b\rightarrow v}({\bf i},{\bf
ik})+\rho_{b\rightarrow
v}({\bf i},{\bf ijk}))\nonumber\\
&& - \prod_{b\in C(v)\setminus\{c\}}(\rho_{b\rightarrow v}({\bf
i},{\bf ij})+\rho_{b\rightarrow v}({\bf i},{\bf ijk}))\\
\label{equ:bp_3col_l_i_ijk}\lambda_{v\rightarrow c}({\bf i},{\bf
ijk}) &:=& \prod_{b\in C(v)\setminus\{c\}}(\rho_{b\rightarrow
v}({\bf i},{\bf ij})+\rho_{b\rightarrow v}({\bf i},{\bf
ik})+\rho_{b\rightarrow
v}({\bf i},{\bf ijk}))\nonumber\\
&& - \prod_{b\in C(v)\setminus\{c\}}(\rho_{b\rightarrow v}({\bf
i},{\bf ij})+\rho_{b\rightarrow v}({\bf i},{\bf ijk}))\nonumber\\
&& - \prod_{b\in C(v)\setminus\{c\}}(\rho_{b\rightarrow v}({\bf
i},{\bf ik})+\rho_{b\rightarrow v}({\bf i},{\bf ijk}))+\prod_{b\in
C(v)\setminus\{c\}}\rho_{b\rightarrow v}({\bf i},{\bf ijk})\\
\label{equ:bp_3col_l_ij_ij}\lambda_{v\rightarrow c}({\bf ij},{\bf
ij}) &:=& \prod_{b\in C(v)\setminus\{c\}}(\rho_{b\rightarrow v}({\bf
ij},{\bf
ij})+\rho_{b\rightarrow v}({\bf ij},{\bf ijk}))\\
\label{equ:bp_3col_l_ij_ijk}\lambda_{v\rightarrow c}({\bf ij},{\bf
ijk}) &:=& \prod_{b\in C(v)\setminus\{c\}}(\rho_{b\rightarrow
v}({\bf ij},{\bf ij})+\rho_{b\rightarrow v}({\bf ij},{\bf
ijk}))-\prod_{b\in
C(v)\setminus\{c\}}\rho_{b\rightarrow v}({\bf ij},{\bf ijk})\\
\label{equ:bp_3col_l_ijk_ijk}\lambda_{v\rightarrow c}({\bf ijk},{\bf
ijk}) &:=& \prod_{b\in
C(v)\setminus\{c\}}\rho_{b\rightarrow v}({\bf ijk},{\bf ijk})
\end{eqnarray}
\begin{eqnarray}
\label{equ:bp_3col_r_i_ij}\rho_{c\rightarrow v}({\bf i},{\bf ij})
&:=& \lambda_{V(c)\setminus\{v\}\rightarrow c}({\bf k},{\bf jk})\\
\label{equ:bp_3col_r_i_ijk}\rho_{c\rightarrow v}({\bf i},{\bf ijk})
&:=& \lambda_{V(c)\setminus\{v\}\rightarrow c}({\bf jk},{\bf jk})\\
\label{equ:bp_3col_r_ij_ij}\rho_{c\rightarrow v}({\bf ij},{\bf ij})
&:=& \lambda_{V(c)\setminus\{v\}\rightarrow c}({\bf k},{\bf ijk})\\
\label{equ:bp_3col_r_ij_ijk}\rho_{c\rightarrow v}({\bf ij},{\bf
ijk}) &:=& \lambda_{V(c)\setminus\{v\}\rightarrow c}({\bf ij},{\bf
ijk})+ \lambda_{V(c)\setminus\{v\}\rightarrow c}({\bf jk},{\bf
ijk})+ \lambda_{V(c)\setminus\{v\}\rightarrow c}({\bf ik},{\bf
ijk})\nonumber\\
&& +\lambda_{V(c)\setminus\{v\}\rightarrow c}({\bf ijk},{\bf ijk})\\
\label{equ:bp_3col_r_ijk_ijk}\rho_{c\rightarrow v}({\bf ijk},{\bf
ijk}) &:=& \lambda_{V(c)\setminus\{v\}\rightarrow c}({\bf ij},{\bf
ijk})+ \lambda_{V(c)\setminus\{v\}\rightarrow c}({\bf jk},{\bf
ijk})+ \lambda_{V(c)\setminus\{v\}\rightarrow c}({\bf ik},{\bf
ijk})\nonumber\\
&& +\lambda_{V(c)\setminus\{v\}\rightarrow c}({\bf ijk},{\bf ijk})
\end{eqnarray}
\begin{eqnarray}
\label{equ:bp_3col_s_i}\mu_v({\bf i}) &:=& \prod_{c\in
C(v)}(\rho_{c\rightarrow v}({\bf i},{\bf ij})+\rho_{c\rightarrow
v}({\bf i},{\bf ik})+\rho_{c\rightarrow
v}({\bf i},{\bf ijk}))\nonumber\\
&& -\prod_{c\in C(v)}(\rho_{c\rightarrow v}({\bf i},{\bf ij})
+\rho_{c\rightarrow v}({\bf i},{\bf ijk}))\nonumber\\
&& -\prod_{c\in C(v)}(\rho_{c\rightarrow v}({\bf i},{\bf ik})
+\rho_{c\rightarrow v}({\bf i},{\bf ijk}))+\prod_{c\in
C(v)}\rho_{c\rightarrow v}({\bf i},{\bf ijk})\\
\label{equ:bp_3col_s_ij}\mu_v({\bf ij}) &:=& \prod_{c\in
C(v)}(\rho_{c\rightarrow v}({\bf ij},{\bf ij})+\rho_{c\rightarrow
v}({\bf ij},{\bf ijk}))-\prod_{c\in C(v)}\rho_{c\rightarrow v}({\bf
ij},{\bf
ijk})\\
\label{equ:bp_3col_s_ijk}\mu_v({\bf ijk}) &:=& \prod_{c\in
C(v)}\rho_{c\rightarrow v}({\bf ijk},{\bf ijk}).
\end{eqnarray}

\end{lem}

Before we begin to consider the BP-to-PTP reduction for $3$-COL
problems, it is helpful to take a closer look at the BP-to-PTP
reduction mechanism for $k$-SAT problems.

In Theorem \ref{thm:sp_as_bp_ksat}, one may notice the two conditions
governing the BP-to-PTP reduction for $k$-SAT problems, namely, the
initialization condition and the
normalization condition.
It is
arguable that the normalization condition imposed on the BP messages,
although serving to simplify the form
of BP messages and possibly to alter the interpretation of the
messages, does not have a critical impact on the message-passing
dynamics. This is because the normalization condition merely
involves a scaling operation, without which BP messages and PTP
messages for $k$-SAT would still be equivalent up to a scaling
factor. On the other hand, the initialization condition in Theorem
\ref{thm:sp_as_bp_ksat} plays an important role on the
message-passing dynamics. In essence, the initialization condition
assures that any right message depends only on the right state it
involves. Using the ``intention-command'' analogy, in which one
views each right state as storing the ``command'' sent from a constraint and
each left state as storing the ``intention'' of
a variable, this condition simply restricts that the
{\em distribution} of the command sent
to any variable does {\em not}
depend on the intention of the variable.
It is remarkable that this
interpretation of the initialization condition in Theorem
\ref{thm:sp_as_bp_ksat} (or (\ref{eq:reduce_ksat_initial_condition}))
is  consistent with the
PTP message-passing rule, in which any right message (i.e., outgoing
distribution of command) sent to a variable
is independent of (or, not a function of,)
the incoming intention from that variable. This is however not the case for
the right messages of BP in general.

We are then motivated to formalize
this condition for general CSPs as what we call the
``state-decoupling'' condition and impose it on the right messages of BP, so
as to achieve a consistency with PTP. It is intuitively sensible
that such a consistency is needed in the reduction of PTP
from BP.

\begin{defn}[State-Decoupling Condition] For an arbitrary CSP and at any
given iteration, the BP messages based on the MRF formalism defined
by (\ref{equ:leftfunction}), (\ref{equ:rightfunction}), and
(\ref{equ:globalfunction}) are said to satisfy the state-decoupling
condition if for every $(v-c)$, the
right message $\rho_{c\rightarrow v}(s_{v,c})$ is only a function of
the right state $s_{v,c}^R$, namely, if for any fixed $s^R_{v,c}\in
\left(\chi^*\right)^{\{v\}}$ and any $s^L_{v,c}\subset s^R_{v,c}$,
$\rho_{c\rightarrow v}(s^L_{v,c}, s^R_{v,c})=\rho_{c\rightarrow
v}(s^R_{v,c}, s^R_{v,c})$.
\end{defn}

It is clear that the initialization condition for BP-to-PTP
reduction for $k$-SAT in Theorem \ref{thm:sp_as_bp_ksat} is
equivalent to this condition, where we note that the condition in
Theorem \ref{thm:sp_as_bp_ksat} only puts restrictions on the right
messages with right state equal to $*$, since for the remaining case
with right state equal to ${\bf L}$ this condition is trivially
satisfied.

It is interesting to observe, as shown in Proposition \ref{prop:initial_condition_ksat}, that
for $k$-SAT problems, as long as
the state-decoupling condition is imposed in the initialization of the BP
messages, the condition is preserved in every iteration. This serves as the
basis for BP to reduce to PTP as shown in Theorem \ref{thm:sp_as_bp_ksat}
and its proof.
For $3$-COL problems, however, the corresponding result to Proposition
\ref{prop:initial_condition_ksat} does not hold.

\begin{lem}
For $3$-COL problems, if the state-decoupling condition holds for BP messages
both in iteration $l$ and in iteration $l+1$, then the right message
in iteration $l$ must satisfy for every $(v-c)$
\[
\rho_{c\rightarrow v}(s^L, s^R)=0
\]
as long as right state $s^R \neq {\bf 123}$.
\end{lem}

\begin{proof}
In 3-COL problems, the state-decoupling condition can be expressed
as
\begin{eqnarray*}
\label{equ:3col_sd_ij}\rho_{c\rightarrow v}({\bf i},{\bf ij}) &=&
\rho_{c\rightarrow v}({\bf ij},{\bf ij})\\
\label{equ:3col_sd_ijk}\rho_{c\rightarrow v}({\bf i},{\bf ijk}) &=&
\rho_{c\rightarrow v}({\bf ij},{\bf ijk}) = \rho_{c\rightarrow
v}({\bf ijk},{\bf ijk}).
\end{eqnarray*}

Note that we only need to prove the Lemma for $s^R$ being a pair of assignments,since when $s^R$ is a singleton, all right messages equal $0$ by the construction of the MRF and Lemma \ref{lem:BP_3COL} describing the BP message-update rule for $3$-COL.

In iteration $l+1$, following 3-COL message-update equations
(\ref{equ:bp_3col_l_i_ij}) to (\ref{equ:bp_3col_r_ijk_ijk}) and using
a superscript to denote the iteration number, we have

\begin{eqnarray}
\rho^{(l+1)}_{c\rightarrow v}({\bf i},{\bf ij}) &=&
\lambda^{(l+1)}_{V(c)\setminus\{v\}\rightarrow c}({\bf k},{\bf
jk}) \nonumber \\
&=& \prod_{b\in
C(V(c)\setminus\{v\})\setminus\{c\}}\left(\rho^{(l)}_{b\rightarrow
V(c)\setminus\{v\}}({\bf k},{\bf jk})+\rho^{(l)}_{b\rightarrow
V(c)\setminus\{v\}}({\bf k},{\bf
ik})+\rho^{(l)}_{b\rightarrow V(c)\setminus\{v\}}({\bf k},{\bf ijk})\right)\nonumber \\
\label{eq:AA}
&& - \prod_{b\in
C(V(c)\setminus\{v\})\setminus\{c\}}\left(\rho^{(l)}_{b\rightarrow
V(c)\setminus\{v\}}({\bf k},{\bf jk})+\rho^{(l)}_{b\rightarrow
V(c)\setminus\{v\}}({\bf k},{\bf
ijk})\right),\\
\rho^{(l+1)}_{c\rightarrow v}({\bf ij},{\bf ij}) &=&
\lambda^{(l+1)}_{V(c)\setminus\{v\}\rightarrow c}({\bf k},{\bf
ijk})\nonumber \\
&=& \prod_{b\in
C(V(c)\setminus\{v\})\setminus\{c\}}\left(\rho^{(l)}_{b\rightarrow
V(c)\setminus\{v\}}({\bf k},{\bf jk})+\rho^{(l)}_{b\rightarrow
V(c)\setminus\{v\}}({\bf k},{\bf
ik})+\rho^{(l)}_{b\rightarrow V(c)\setminus\{v\}}({\bf k},{\bf ijk})\right)\nonumber \\
&& - \prod_{b\in
C(V(c)\setminus\{v\})\setminus\{c\}}\left(\rho^{(l)}_{b\rightarrow
V(c)\setminus\{v\}}({\bf k},{\bf ik})+\rho^{(l)}_{b\rightarrow
V(c)\setminus\{v\}}({\bf k},{\bf
ijk})\right)\nonumber \\
&& - \prod_{b\in
C(V(c)\setminus\{v\})\setminus\{c\}}\left(\rho^{(l)}_{b\rightarrow
V(c)\setminus\{v\}}({\bf k},{\bf jk})+\rho^{(l)}_{b\rightarrow
V(c)\setminus\{v\}}({\bf k},{\bf ijk})\right)\nonumber \\
\label{eq:BB}
&& + \prod_{b\in
C(V(c)\setminus\{v\})\setminus\{c\}}\rho^{(l)}_{b\rightarrow
V(c)\setminus\{v\}}({\bf k},{\bf ijk}).
\end{eqnarray}

Now suppose that the state-decoupling condition as expressed above can be satisfied both in iteration $l$ and
in iteration $l+1$. Then we may equate the right-hand sides of (\ref{eq:AA})
and (\ref{eq:BB}), namely,

\begin{eqnarray*}
&&\prod_{b\in
C(V(c)\setminus\{v\})\setminus\{c\}}\left(\rho^{(l)}_{b\rightarrow
V(c)\setminus\{v\}}({\bf k},{\bf jk})+\rho^{(l)}_{b\rightarrow
V(c)\setminus\{v\}}({\bf k},{\bf ik})+\rho^{(l)}_{b\rightarrow
V(c)\setminus\{v\}}({\bf k},{\bf ijk})\right)\\
&& - \prod_{b\in
C(V(c)\setminus\{v\})\setminus\{c\}}\left(\rho^{(l)}_{b\rightarrow
V(c)\setminus\{v\}}({\bf k},{\bf jk})+\rho^{(l)}_{b\rightarrow
V(c)\setminus\{v\}}({\bf k},{\bf
ijk})\right)\\
&=& \prod_{b\in
C(V(c)\setminus\{v\})\setminus\{c\}}\left(\rho^{(l)}_{b\rightarrow
V(c)\setminus\{v\}}({\bf k},{\bf jk})+\rho^{(l)}_{b\rightarrow
V(c)\setminus\{v\}}({\bf k},{\bf ik})+\rho^{(l)}_{b\rightarrow
V(c)\setminus\{v\}}({\bf k},{\bf ijk})\right)\\
&& - \prod_{b\in
C(V(c)\setminus\{v\})\setminus\{c\}}\left(\rho^{(l)}_{b\rightarrow
V(c)\setminus\{v\}}({\bf k},{\bf ik})+\rho^{(l)}_{b\rightarrow
V(c)\setminus\{v\}}({\bf k},{\bf
ijk})\right)\\
&& - \prod_{b\in
C(V(c)\setminus\{v\})\setminus\{c\}}\left(\rho^{(l)}_{b\rightarrow
V(c)\setminus\{v\}}({\bf k},{\bf jk})+\rho^{(l)}_{b\rightarrow
V(c)\setminus\{v\}}({\bf k},{\bf ijk})\right) + \prod_{b\in
C(V(c)\setminus\{v\})\setminus\{c\}}\rho^{(l)}_{b\rightarrow
V(c)\setminus\{v\}}({\bf
k},{\bf ijk}),
\end{eqnarray*}
which implies
\begin{eqnarray*}
&&\prod_{b\in
C(V(c)\setminus\{v\})\setminus\{c\}}\left(\rho^{(l)}_{b\rightarrow
V(c)\setminus\{v\}}({\bf k},{\bf jk})+\rho^{(l)}_{b\rightarrow
V(c)\setminus\{v\}}({\bf k},{\bf
ijk})\right)\\
&=& \prod_{b\in
C(V(c)\setminus\{v\})\setminus\{c\}}\left(\rho^{(l)}_{b\rightarrow
V(c)\setminus\{v\}}({\bf k},{\bf ik})+\rho^{(l)}_{b\rightarrow
V(c)\setminus\{v\}}({\bf k},{\bf ijk})\right)\\
&& + \prod_{b\in
C(V(c)\setminus\{v\})\setminus\{c\}}\left(\rho^{(l)}_{b\rightarrow
V(c)\setminus\{v\}}({\bf k},{\bf jk})+\rho^{(l)}_{b\rightarrow
V(c)\setminus\{v\}}({\bf k},{\bf ijk})\right) - \prod_{b\in
C(V(c)\setminus\{v\})\setminus\{c\}}\rho^{(l)}_{b\rightarrow
V(c)\setminus\{v\}}({\bf
k},{\bf ijk}).
\end{eqnarray*}
Since every right message must be non-negative, when the state-decoupling
condition is satisfied in iteration $l$, the only way to make the above
equality hold is the case where
\[\rho^{(l)}_{b\rightarrow V(c)\setminus\{v\}}({\bf
k},{\bf ik})=0.
\]

Under the state-decoupling condition, this also means $\rho_{b\rightarrow V(c)
\setminus \{v\}}^{(l)}({\bf ik}, {\bf ik})=0$. Thus we establish this lemma.

\end{proof}

This lemma suggests that when the BP messages satisfy the state-decoupling
condition in two consecutive iterations, then the right messages must take a
trivial form --- equal to $[s^R={\bf 123}]$ up to scale, and contain no
information.

At this point, one is left with either the option of concluding that
PTP (or SP) is {\em not} an instance of BP for $3$-COL problems (and
hence for general CSPs) or the option of doubting the usefulness of
the state-decoupling condition in BP-to-SP reduction. In the remainder of this
subsection, we will clear this doubt and assert the usefulness of the state-decoupling condition by showing that when the state-decoupling condition is {\em manually} imposed on the BP messages in each iteration, BP still reduces to PTP
for $3$-COL problems. That will allow us to conclude that PTP (or SP) is not
a special case of BP.

To force the state-decoupling condition to be satisfied in each BP iteration,
now we modify the  message-passing rule of BP on the Forney graph
representation of general CSPs, and introduce a ``new''
message-passing procedure which we refer to as the {\em
state-decoupled BP} or SDBP.
We note that introducing this ``new'' message-passing procedure is solely
for the purpose of verifying the usefulness of the state-decoupling
condition and hopefully arriving at a unified reduction mechanism
for PTP to reduce from BP (or more precisely from SDBP). Beyond this
purpose, we have no intention to justify the introduction of SDBP.

Identical to BP at local function vertices, SDBP differs from BP in
that messages passed from the right functions need an additional
processing (so that the state-decoupling condition is satisfied)
before they are passed to the left functions. In SDBP, there are
three kinds of messages: {\em right message} $\rho_{c\rightarrow v}$
is computed at right function $f_c$ to pass along the edge to $g_v$;
{\em state-decoupled right message} $\rho_{c\rightarrow v}^*$ is
computed at the edge connecting $f_c$ and $g_v$, which satisfies the
state-decoupling condition, computed only based on the right message
$\rho_{c\rightarrow v}$ on the same edge and to be passed to left
function $g_v$; {\em left message} $\lambda_{v\rightarrow c}$ is
computed at the left function $g_v$ to pass along the edge
connecting to $f_c$. The precise definition of SDBP message-update
rule is given next.

\begin{defn}
The SDBP message-update rule is defined as follows.
\begin{eqnarray}
\label{equ:sd_bp_leftmessage} \lambda_{v\rightarrow
c}(s^L_{v,c},s^R_{v,c}) &:=&
\sum_{s^R_{v,C(v)\setminus\{c\}}}\omega_v\left(s^L_{v,c}\left|\bigcap_{b\in
C(v)}s^R_{v,b}\right.\right)\cdot\prod_{b\in
C(v)\setminus\{c\}}\rho^*_{b\rightarrow
v}(s^R_{v,b})\\
\label{equ:sd_bp_rightmessage}\rho^*_{c\rightarrow v}(s^R_{v,c})
&:=& \delta\cdot\rho_{c\rightarrow
v}(s^R_{v,c},s^R_{v,c})\\
\label{equ:sd_bp_r_mssg}\rho_{c\rightarrow v}(s^L_{v,c},s^R_{v,c})
&:=&
\!\!\!\!\!\!\sum_{s^L_{V(c)\setminus\{v\},c}}\left[s^R_{v,c}={\tt
F}_c(s^L_{V(c)\setminus\{v\},c})\right]\prod_{u\in
V(c)\setminus\{v\}}\lambda_{u\rightarrow c}(s^L_{u,c},{\tt
F}_c(s^L_{V(c)\setminus \{u\},c}))\\
\label{equ:sd_bp_summary}\mu_v(y_v) &:=&
\sum_{s^R_{v,C(v)}}\omega_v\left(y_v\left|\bigcap_{c\in
C(v)}s^R_{v,c}\right.\right)\prod_{c\in C(v)}\rho^*_{c\rightarrow
v}(s^R_{v,c})
\end{eqnarray}
where $\delta=1/\sum_{s^R_{v,c}\in \left(
\chi^*\right)^{\{v\}}}\rho_{c\rightarrow v}(s^R_{v,c},s^R_{v,c})$.
\end{defn}

Comparing  this definition with the BP message-update rule in Lemma \ref{lem:MRF_BP}, the following remarks are in order.
First, the expression of right messages $\rho$ in terms of left
messages $\lambda$ is identical to that in BP. Second, each
state-decoupled message $\rho^*_{c\rightarrow v}$ may be regarded as
a function of $(s^L_{v,c}, s^R_{v,c})$ but the value of the function
only depends the $s^R_{v,c}$ component, namely that the
(state-decoupled) right message satisfies the state-decoupling
condition.
Furthermore, the expression of $\lambda$ in terms of
$\rho^*$  is precisely the same as the expression of $\lambda$ in
terms of $\rho$ in BP\footnote{Although it is possible to formulate
SDBP in more compact form by, for example, suppressing $\rho$ and
expressing the message-update rule only using $\rho^*$ and
$\lambda$, we feel the current way of formulating SDBP makes it
easier to compare SDBP with BP.}.

Following this definition, the next lemma summarizes the SDBP message-update rule for $3$-COL problems.

\begin{lem}
Let $\{\omega_v:v\in V\}$ in 3-COL problems be defined as in
(\ref{equ:sd_bp_3col_weight}). The SDBP message-update rule is then
:

\begin{eqnarray}
\label{equ:sd_bp_3col_l_i_ij}\lambda_{v\rightarrow c}({\bf i},{\bf
ij}) &:=& \prod_{b\in C(v)\setminus\{c\}}\left(\rho^*_{b\rightarrow
v}({\bf ij})+\rho^*_{b\rightarrow v}({\bf ik})+\rho^*_{b\rightarrow
v}({\bf
ijk})\right)\nonumber\\
&&-\prod_{b\in C(v)\setminus\{c\}}\left(\rho^*_{b\rightarrow v}({\bf
ij})+\rho^*_{b\rightarrow v}({\bf
ijk})\right)\\
\label{equ:sd_bp_3col_l_i_ijk}\lambda_{v\rightarrow c}({\bf i},{\bf
ijk}) &:=& \prod_{b\in C(v)\setminus\{c\}}\left(\rho^*_{b\rightarrow
v}({\bf ij})+\rho^*_{b\rightarrow v}({\bf ik})+\rho^*_{b\rightarrow
v}({\bf ijk})\right) -\prod_{b\in
C(v)\setminus\{c\}}\left(\rho^*_{b\rightarrow v}({\bf
ij})+\rho^*_{b\rightarrow v}({\bf
ijk})\right)\nonumber\\
&& -\prod_{b\in C(v)\setminus\{c\}}\left(\rho^*_{b\rightarrow
v}({\bf ik})+\rho^*_{b\rightarrow v}({\bf ijk})\right)+\prod_{b\in
C(v)\setminus\{c\}}\rho^*_{b\rightarrow
v}({\bf ijk})\\
\label{equ:sd_bp_3col_l_ij_ij}\lambda_{v\rightarrow c}({\bf ij},{\bf
ij}) &:=& \prod_{b\in C(v)\setminus\{c\}}\left(\rho^*_{b\rightarrow
v}({\bf ij})+\rho^*_{b\rightarrow v}({\bf
ijk})\right)\\
\label{equ:sd_bp_3col_l_ij_ijk}\lambda_{v\rightarrow c}({\bf
ij},{\bf ijk}) &:=& \prod_{b\in
C(v)\setminus\{c\}}\left(\rho^*_{b\rightarrow v}({\bf
ij})+\rho^*_{b\rightarrow v}({\bf ijk})\right)-\prod_{b\in
C(v)\setminus\{c\}}\rho^*_{b\rightarrow v}({\bf ijk})\\
\label{equ:sd_bp_3col_l_ijk_ijk}\lambda_{v\rightarrow c}({\bf
ijk},{\bf ijk}) &:=& \prod_{b\in
C(v)\setminus\{c\}}\rho^*_{b\rightarrow v}({\bf ijk})
\end{eqnarray}
\begin{eqnarray}
\label{equ:sd_bp_3col_r_*_ij_ij}\rho^*_{c\rightarrow v}({\bf
ij}) &:=& \delta\cdot\lambda_{V(c)\setminus\{v\}\rightarrow c}({\bf k},{\bf ijk})\\
\label{equ:sd_bp_3col_r_*_ijk_ijk}\rho^*_{c\rightarrow v}({\bf ijk})
&:=& \delta\cdot\left(\lambda_{V(c)\setminus\{v\}\rightarrow c}({\bf
ij},{\bf ijk})+\lambda_{V(c)\setminus\{v\}\rightarrow c}({\bf
ik},{\bf ijk}) +\lambda_{V(c)\setminus\{v\}\rightarrow c}({\bf
jk},{\bf
ijk})\right.\nonumber\\
&& \left.+\lambda_{V(c)\setminus\{v\}\rightarrow c}({\bf ijk},{\bf
ijk})\right)
\end{eqnarray}
\begin{eqnarray}
\label{equ:sd_bp_3col_s_i}\mu_v({\bf i}) &:=& \prod_{c\in
C(v)}(\rho^*_{c\rightarrow v}({\bf ij})+\rho^*_{c\rightarrow v}({\bf
ik})+\rho^*_{c\rightarrow v}({\bf ijk})) -\prod_{c\in
C(v)}(\rho^*_{c\rightarrow v}({\bf ij})
+\rho^*_{c\rightarrow v}({\bf ijk}))\nonumber\\
&& -\prod_{c\in C(v)}(\rho^*_{c\rightarrow v}({\bf ik})
+\rho^*_{c\rightarrow v}({\bf ijk}))+\prod_{c\in
C(v)}\rho^*_{c\rightarrow v}({\bf ijk})\\
\label{equ:sd_bp_3col_s_ij}\mu_v({\bf ij}) &:=& \prod_{c\in
C(v)}(\rho^*_{c\rightarrow v}({\bf ij})+\rho^*_{c\rightarrow v}({\bf
ijk}))-\prod_{c\in C(v)}\rho^*_{c\rightarrow v}({\bf
ijk})\\
\label{equ:sd_bp_3col_s_ijk}\mu_v({\bf ijk}) &:=& \prod_{c\in
C(v)}\rho^*_{c\rightarrow v}({\bf ijk}),
\end{eqnarray}
where $\delta$ is such that
\[
\rho^*_{c\rightarrow v}({\bf ijk})+
\sum\limits_{{\bf ij}}\rho^*_{c\rightarrow v}({\bf ij})=1.
\]
\end{lem}

It is now possible to establish a correspondence between PTP and
SDBP messages for $3$-COL problems.

\begin{thm}
For 3-COL problems, the correspondence between PTP and SDBP
message-update rules is
\begin{eqnarray}
\label{equ:ptp_sd_bp_l_i}\lambda^{(\rm PTP)}_{v\rightarrow c}({\bf
i})
&\leftrightarrow& \lambda^{(\rm SDBP)}_{v\rightarrow c}({\bf i},{\bf ijk})\\
\label{equ:ptp_sd_bp_l_ij}\lambda^{(\rm PTP)}_{v\rightarrow c}({\bf
ij})
&\leftrightarrow& \lambda^{(\rm SDBP)}_{v\rightarrow c}({\bf ij},{\bf ijk})\\
\label{equ:ptp_sd_bp_l_ijk}\lambda^{(\rm PTP)}_{v\rightarrow c}({\bf
ijk})
&\leftrightarrow& \lambda^{(\rm SDBP)}_{v\rightarrow c}({\bf ijk},{\bf ijk})\\
\label{equ:ptp_sd_bp_r_ij}\rho^{\rm norm (PTP)}_{c\rightarrow
v}({\bf ij})
&\leftrightarrow& \rho^{*(\rm SDBP)}_{c\rightarrow v}({\bf ij})\\
\label{equ:ptp_sd_bp_r_ijk}\rho^{\rm norm (PTP)}_{c\rightarrow
v}({\bf ijk}) &\leftrightarrow& \rho^{* (\rm SDBP)}_{c\rightarrow
v}({\bf ijk})\\
\label{equ:ptp_sd_bp_s_i}\mu^{(\rm {PTP})}_v({\bf i})
&\leftrightarrow&
\mu^{(\rm SDBP)}_v({\bf i})\\
\label{equ:ptp_sd_bp_s_ij}\mu^{(\rm {PTP})}_v({\bf ij})
&\leftrightarrow&
\mu^{(\rm SDBP)}_v({\bf ij})\\
\label{equ:ptp_sd_bp_s_ijk}\mu^{(\rm {PTP})}_v({\bf ijk})
&\leftrightarrow& \mu^{(\rm SDBP)}_v({\bf ijk}).
\end{eqnarray}
\end{thm}

\begin{proof}
We will first prove that if the ``right correspondence'' (namely
that (\ref{equ:ptp_sd_bp_r_ij}) and (\ref{equ:ptp_sd_bp_r_ijk})) holds, then
the ``left correspondence'' (namely that
(\ref{equ:ptp_sd_bp_l_i}) to (\ref{equ:ptp_sd_bp_l_ijk})) holds.

Suppose that the right correspondence holds (where the symbol $\leftrightarrow$
in (\ref{equ:ptp_sd_bp_r_ij}) and (\ref{equ:ptp_sd_bp_r_ijk}) is understood
as equality). Then

\begin{eqnarray*}
\lambda^{(\rm SDBP)}_{v\rightarrow c}({\bf i},{\bf ijk}) &=&
\prod_{b\in C(v)\setminus\{c\}}\left(\rho^{*(\rm
SDBP)}_{b\rightarrow v}({\bf ij})+\rho^{*(\rm SDBP)}_{b\rightarrow
v}({\bf ik})+\rho^{*(\rm SDBP)}_{b\rightarrow v}({\bf
ijk})\right)\\
&& -\prod_{b\in C(v)\setminus\{c\}}\left(\rho^{*(\rm
SDBP)}_{b\rightarrow v}({\bf ij})+\rho^{*(\rm SDBP)}_{b\rightarrow
v}({\bf
ijk})\right)\\
&& -\prod_{b\in C(v)\setminus\{c\}}\left(\rho^{*(\rm
SDBP)}_{b\rightarrow v}({\bf ik})+\rho^{*(\rm SDBP)}_{b\rightarrow
v}({\bf ijk})\right)+\prod_{b\in C(v)\setminus\{c\}}\rho^{*(\rm
SDBP)}_{b\rightarrow
v}({\bf ijk})\\
&=& \prod_{b\in C(v)\setminus\{c\}}\left(\rho^{\rm norm
(PTP)}_{b\rightarrow v}({\bf ij})+\rho^{\rm norm
(PTP)}_{b\rightarrow v}({\bf ik})+\rho^{\rm norm
(PTP)}_{b\rightarrow v}({\bf ijk})\right)\\
&& -\prod_{b\in C(v)\setminus\{c\}}\left(\rho^{\rm norm
(PTP)}_{b\rightarrow
v}({\bf ij})+\rho^{\rm norm (PTP)}_{b\rightarrow v}({\bf ijk})\right)\\
&& -\prod_{b\in C(v)\setminus\{c\}}\left(\rho^{\rm norm
(PTP)}_{b\rightarrow v}({\bf ik})+\rho^{\rm norm
(PTP)}_{b\rightarrow v}({\bf ijk})\right)+\prod_{b\in
C(v)\setminus\{c\}}\rho^{\rm
norm (PTP)}_{b\rightarrow v}({\bf ijk})\\
&=& \lambda^{(\rm PTP)}_{v\rightarrow c}({\bf i}).
\end{eqnarray*}

Similarly, we can prove that $\lambda^{(\rm SDBP)}_{v\rightarrow
c}({\bf ij},{\bf ijk})=\lambda^{(\rm PTP)}_{v\rightarrow c}({\bf
ij})$ and $\lambda^{(\rm SDBP)}_{v\rightarrow c}({\bf ijk},{\bf
ijk})=\lambda^{(\rm PTP)}_{v\rightarrow c}({\bf ijk})$. It then follows that
the left correspondence holds.

Now we prove that if the left correspondence holds, then the right
correspondence holds.  Suppose that the left correspondence holds,
then we have
\begin{eqnarray*}
\rho^{\rm norm (PTP)}_{c\rightarrow v}({\bf ij}) &=&
\alpha\cdot\rho^{(\rm PTP)}_{c\rightarrow v}({\bf ij})\\
&=& \alpha\cdot\lambda^{\rm norm
(PTP)}_{V(c)\setminus\{v\}\rightarrow c}({\bf
k})\\
&=& \alpha\left(\beta\cdot\lambda^{(\rm
PTP)}_{V(c)\setminus\{v\}\rightarrow c}({\bf
k})\right)\\
&=& \alpha\beta\cdot\lambda^{(\rm
SDBP)}_{V(c)\setminus\{v\}\rightarrow c}({\bf k},{\bf ijk})
\end{eqnarray*}
where $\alpha=1/\sum_{t\in(\chi^*)^v}\rho^{(\rm PTP)}_{c\rightarrow
v}(t)$ and
$\beta=1/\sum_{t\in(\chi^*)^{V(c)\setminus\{v\}}}\lambda^{(\rm
PTP)}_{V(c)\setminus\{v\}\rightarrow c}(t)$. We also have
\begin{eqnarray*}
\rho^{*({\rm SDBP})}_{c\rightarrow v}({\bf ij})
&=& \delta\cdot\lambda^{(\rm SDBP)}_{V(c)\setminus\{v\}\rightarrow
c}({\bf k},{\bf ijk}).
\end{eqnarray*}

Since both $\rho^{*({\rm SDBP})}_{c\rightarrow v}$ and $\rho^{\rm
norm(PTP)}_{c\rightarrow v}$ are normalized, it must hold that
$\alpha\beta=\delta$. This indicates that $\rho^{\rm
norm(PTP)}_{c\rightarrow v}({\bf ij})=\rho^{*({\rm
SDBP})}_{c\rightarrow v}({\bf ij})$. Following a similar procedure, one 
can show that
$\rho^{\rm norm(PTP)}_{c\rightarrow v}({\bf ijk})=\rho^{*({\rm
SDBP})}_{c\rightarrow v}({\bf ijk})$. This implies that the right
correspondence holds.

At this point, we have established the
correspondence between passed messages in PTP and those in SDBP.

Now we will prove the summary correspondence (namely, that
(\ref{equ:ptp_sd_bp_s_i}) to (\ref{equ:ptp_sd_bp_s_ijk})).

\begin{eqnarray*}
\mu_v^{(\rm SDBP)}({\bf i}) &=& \prod_{c\in C(v)}(\rho^{*({\rm
SDBP})}_{c\rightarrow v}({\bf ij})+\rho^{*({\rm
SDBP})}_{c\rightarrow v}({\bf ik})+\rho^{*({\rm
SDBP})}_{c\rightarrow
v}({\bf ijk}))\\
&& -\prod_{c\in C(v)}(\rho^{*({\rm SDBP})}_{c\rightarrow v}({\bf
ij}) +\rho^{*({\rm
SDBP})}_{c\rightarrow v}({\bf ijk}))\\
&& -\prod_{c\in C(v)}(\rho^{*({\rm SDBP})}_{c\rightarrow v}({\bf
ik}) +\rho^{*({\rm SDBP})}_{c\rightarrow v}({\bf ijk}))+\prod_{c\in
C(v)}\rho^{*({\rm
SDBP})}_{c\rightarrow v}({\bf ijk})\\
&=& \prod_{c\in C(v)}(\rho^{\rm norm(PTP)}_{c\rightarrow v}({\bf
ij})+\rho^{\rm norm(PTP)}_{c\rightarrow v}({\bf ik})+\rho^{\rm
norm(PTP)}_{c\rightarrow v}({\bf ijk})\\
&& -\prod_{c\in C(v)}(\rho^{\rm norm(PTP)}_{c\rightarrow v}({\bf
ij})
+\rho^{\rm norm(PTP)}_{c\rightarrow v}({\bf ijk}))\\
&& -\prod_{c\in C(v)}(\rho^{\rm norm(PTP)}_{c\rightarrow v}({\bf
ik}) +\rho^{\rm norm(PTP)}_{c\rightarrow v}({\bf ijk}))+\prod_{c\in
C(v)}\rho^{\rm norm(PTP)}_{c\rightarrow v}({\bf ijk})\\
&=& \mu^{(\rm PTP)}_v({\bf i}).
\end{eqnarray*}

Similarly, we can prove that $\mu_v^{(\rm SDBP)}({\bf ij})=\mu^{(\rm
PTP)}_v({\bf ij})$ and  $\mu_v^{(\rm SDBP)}({\bf ijk})=\mu^{(\rm
PTP)}_v({\bf ijk})$. This proves the summary correspondence.

\end{proof}

At this end, it should be convincing that the state-decoupling condition
is an important ingredient in the reduction of BP to PTP. It is worth noting that in the case of $k$-SAT problems, this condition can be imposed simply by the
initialization of BP messages. However in the case
of $3$-COL problems, one needs to manually impose this condition at each iteration, namely, carrying out SDBP instead of BP, so as to arrive at an equivalence to PTP messages.  This extra complexity involved in $3$-COL problems then suggests that for $3$-COL problems, PTP and hence SP are not a special case of BP. Thus at this end, one may conclude that SP is not BP for general CSPs.

Now it remains to investigate, for general CSPs,
whether the state-decoupling condition is sufficient for PTP or weighted PTP to reduce from BP, or equivalently {\em whether}
and {\em when}
PTP and weighted PTP are SDBP.

\subsection{The Reduction of Weighted PTP from SDBP for General CSPs}

Up to this point, we see that the state-decoupling condition critically governs the reduction of BP to PTP (or weighted PTP) for $k$-SAT problems and
$3$-COL problems. In this subsection, we will however show that the
state-decoupling condition is not sufficient for BP (more precisely SDBP)
to reduce to PTP and that an additional condition is needed in the general context.

\begin{defn}[Forceable Token]
For any $(v-c)$, we say that a token
$t_v\in \left(\chi^*\right)^{\{v\}}$ is {\em forceable} by $\Gamma_c$
if there exists a rectangle $\prod\limits_{u\in V(c)\setminus \{v\}}
t_u$ on $V(c)\setminus \{v\}$ such that
${\tt F}_c\left(
\prod\limits_{u\in V(c)\setminus \{v\}}t_u
\right)=t_v$.
\end{defn}

We will denote by ${\cal F}_c(v)$ the set of all tokens on $v$ that
are forceable by $\Gamma_c$. Let ${\cal A}_c(v):=\bigcup_{t\in {\cal
F}_c(v)}t$. Since ${\cal A}_c(v)={\tt F}_c \left( \prod_{u\in
V(c)\setminus \{v\}} \left(\chi^*\right)^{\{u\}} \right)$, it
follows that ${\cal A}_c(v)$ is always forceable.  In fact, it is
easy to see that
 ${\cal A}_c(v)$ is the ``largest'' forceable token on $v$
by $\Gamma_c$ --- in the sense of containing all other forceable
tokens as its subsets
---
due to the monotonicity of ${\tt F}_c(\cdot)$.

In $k$-SAT problems, for any $(v-c)$, it is easy to see that
${\cal F}_c(v)=\{*, {\bf L}\}$, and ${\cal A}_c(v)= *.$
In $3$-COL problems, for any $(v-c)$, it is easy to see that
$
{\cal F}_c(v)=\{{\bf 123}, {\bf 12}, {\bf 23}, {\bf 13} \}
$,
and
$
{\cal A}_c(v)={\bf 123}.
$


For any $(c-v)$, let ${\cal A}_{\sim c}(v)$ be defined by
\[
{\cal A}_{\sim c}(v):= \bigcap_{b\in C(v)\setminus\{c\}} {\cal
A}_b(v).
\]

\begin{defn} [Locally Compatible Constraint]
A constraint $\Gamma_c$ is said to be locally compatible if for any $v\in V(c)$, any forceable token $t_v\in {\cal F}_c(v)$, any rectangle $t'\in {\tt F}_c^{-1}
\left(t_v\right)$ on $V(c)\setminus \{v\}$
(where ${\tt F}^{-1}_c\left(t_v\right)$ is the set of
all rectangles $y_{V(c)\setminus \{v\}}$ on $V(c)\setminus \{v\}$
such that ${\tt F}_c(y_{V(c)\setminus \{v\}})=t_v$)
and any $u\in V(c)\setminus \{v\}$, it holds that
\[
{\cal A}_{\sim c}(u) \subseteq {\tt F}_c\left(t_v \times
t'_{:V(c)\setminus \{u,v\}}\right).
\]
\end{defn}

We note that the local compatibility of a constraint $\Gamma_c$ as defined
above is not simply a property of $\Gamma_c$ itself.
It also relies on the structure
of all constraints that are distance-2 away from $\Gamma_c$ in
the factor graph.

\begin{thm}
\label{thm:reduce_general} Let the set of obedience conditionals
$\{\omega_v:v\in V\}$ be given, where each $v\in V$ corresponds to a
coordinate of a CSP. Let both the MRF of the CSP (that specified via
(\ref{equ:leftfunction}), (\ref{equ:rightfunction}) and
(\ref{equ:globalfunction})) and the weighted PTP for the CSP be both
parametrized by $\{\omega_v:v\in V\}$. Then if every constraint of
the CSP is locally compatible, the SDBP derived from the MRF is
equivalent to the weighted PTP, where the correspondence is
\[
\rho^{\rm norm (PTP)}_{c\rightarrow v} \leftrightarrow \rho^{* {\rm (SDBP)}}_{c\rightarrow v}.
\]
Conversely, if such an equivalence holds for every 
choice of $\{\omega_v:v\in V\}$, then every constraint of the CSP
must be locally compatible.
\end{thm}

Alternatively phrased, Theorem \ref{thm:reduce_general} suggests
that if the state-decoupling condition is satisfied in every
iteration of BP, the local compatibility condition on all
constraints is the necessary and sufficient condition for weighted
PTP to reduce from BP. --- We note that Theorem \ref{thm:reduce_general}
only refers to the equivalence of right messages. It is however straight-forward to verify (as seen in earlier proofs of equivalent results in this paper)
that right equivalence implies the summary equivalence.

This theorem answers the question {\em when} SP is SDBP
in a general setting.

\begin{proof}

Following the message-update rule of SDBP,
\begin{eqnarray}
\rho^{*{\rm (SDBP)}}_{c\rightarrow v}(s^R_{v,c})
&\propto&\sum_{s^L_{V(c)\setminus\{v\},c}} \left(
\left[s^R_{v,c}={\tt
F}_c\left(s^L_{V(c)\setminus\{v\},c}\right)\right] \prod_{u\in
V(c)\setminus\{v\}} \lambda^{{\rm (SDBP)}}_{u\rightarrow
c}\left(s^L_{u,c}, {\tt F}_c\left( s^L_{V(c)\setminus \{u\}, c}
\right)\right)
\right)\nonumber \\
& = &\sum_{s^L_{V(c)\setminus\{v\},c}} \left( \left[s^R_{v,c}={\tt
F}_c\left(s^L_{V(c)\setminus\{v\},c}\right)\right] \prod_{u\in
V(c)\setminus\{v\}} \sum_{s^R_{u, C(u)\setminus \{c\}}}
\right. \nonumber \\
& & \omega_u\left( s^L_{u,c} \left| \left( \bigcap_{b\in
C(u)\setminus \{c\}} s^R_{u,b} \right)\cap {\tt
F}_c\left(s^L_{V(c)\setminus \{u,v\}, c} \times s^R_{v,c}\right)
\right. \right)
\nonumber \\
& & \left.\cdot \prod_{b\in C(u)\setminus \{c\}} \rho^{*({\rm SDBP})}_{b\rightarrow u}(s^R_{u, b})\right)
\label{eq:SDBP_right2right}
\end{eqnarray}

Similarly following the message-update rule of weighted  PTP, we have
\begin{eqnarray}
\rho^{\rm {norm (PTP)}}_{c\rightarrow v}(t_{c\rightarrow v})&\propto
& \sum_{t_{V(c)\setminus\{v\}\rightarrow c}} \left(
\left[t_{c\rightarrow v} ={\tt F}_c(t_{V(c)\setminus\{v\}\rightarrow
c}) \right] \prod_{u\in V(c)\setminus\{v\}}
\sum_{t_{C(u)\setminus\{c\}\rightarrow u}} \right.
\nonumber \\
& & \left. \omega_u \left( t_{u\rightarrow c} \left| \bigcap_{b\in
C(u)\setminus\{c\}}t_{b\rightarrow u} \right. \right)
\cdot\left(\prod_{b\in C(u)\setminus\{c\}}\rho^{{\rm norm
(PTP)}}_{b\rightarrow u}(t_{b\rightarrow u})\right) \right).
\label{eq:PTP_right2right}
\end{eqnarray}

Identifying
every right state  $s^R_{v,c}$
in (\ref{eq:SDBP_right2right})
with token $t_{c\rightarrow v}$ in (\ref{eq:PTP_right2right}) and
every left state $s^L_{v,c}$
(\ref{eq:SDBP_right2right})
 with token $t_{v\rightarrow c}$ in (\ref{eq:PTP_right2right}), the only difference between
(\ref{eq:SDBP_right2right}) and  (\ref{eq:PTP_right2right}) is the
argument of function $\omega_u$. (We note that since both
$\rho^{*{\rm (SDBP)}}_{c\rightarrow v}$ and $\rho^{\rm
norm(PTP)}_{c\rightarrow v}$ are normalized, the scaling constant in
(\ref{eq:SDBP_right2right}) and  (\ref{eq:PTP_right2right}) are
necessarily the same.) We now prove the sufficiency and necessity of
the local compatibility condition for the equivalence between
$\rho_{c\rightarrow v}^{\rm norm (PTP)}$ and $\rho^{* {\rm
(SDBP)}}_{c\rightarrow v}$ via the following chain of two-way
implications.

\begin{eqnarray*}
& & \rho^{* {\rm (SDBP)}}_{c\rightarrow v} \leftrightarrow \rho^{\rm norm (PTP)}_{c\rightarrow v}, \forall v\in V(c)\\
& \Leftrightarrow & \omega_u\left( s^L_{u,c} \left| \left(
\bigcap_{b\in C(u)\setminus \{c\}} s^R_{u,b} \right)\cap {\tt
F}_c\left(s^L_{V(c)\setminus \{u,v\}, c} \times s^R_{v,c}\right)
\right. \right) = \omega_u\left( s^L_{u,c} \left| \bigcap_{b\in
C(u)\setminus \{c\}} s^R_{u,b} \right.
\right)\\
& & \forall v\in V(c) ~\mbox{and every}~ \left(s^R_{v, c},
s^L_{V(c)\setminus \{v\},c}\right) ~\mbox{in the support of }
~\left[s^R_{v,c}={\tt
F}_c\left(s^L_{V(c)\setminus\{v\},c}\right)\right], \\
& & \forall u\in V(c)\setminus \{v\} ~\mbox{and every choice of}
~|C(u)\setminus \{c\}| ~\mbox{tokens on}~ \{u\},
\left\{s^R_{u,b}:b\in C(u)\setminus \{c\} \right\},\\
& &  ~\mbox{with each}~ s^R_{u,b}
~\mbox{in the support of} ~\rho^{({\rm PTP})}_{b\rightarrow u}.\\
& \Leftrightarrow &
\left(\bigcap_{b\in C(u)\setminus \{c\}} s^R_{u,b}\right)
\cap {\tt F}_c\left(s^L_{V(c)\setminus \{u,v\}, c} \times s^R_{v,c}\right)
=
\bigcap_{b\in C(u)\setminus \{c\}} s^R_{u,b}
\\
& & \forall v\in V(c) ~\mbox{and every}~ \left(s^R_{v, c},
s^L_{V(c)\setminus \{v\},c}\right) ~\mbox{such that}~ s^R_{v,c}\in
{\cal F}_c(v) ~{\rm and} ~s^L_{V(c)\setminus \{v\},c}\in
{\tt F}^{-1}_c(s^R_{v,c}), \\
& & \forall u\in V(c)\setminus \{v\} ~\mbox{and every choice of}
~|C(u)\setminus \{c\}| ~\mbox{tokens on}~ \{u\},
\left\{s^R_{u,b}:b\in C(u)\setminus \{c\} \right\},\\
& &  ~\mbox{with each}~ s^R_{u,b}\in {\cal F}_b(u).
\end{eqnarray*}
\begin{eqnarray*}
& \Leftrightarrow &
\bigcap_{b\in C(u)\setminus \{c\}} s^R_{u,b}
\subseteq {\tt F}_c\left(s^L_{V(c)\setminus \{u,v\}, c} \times s^R_{v,c}\right)
\\
& & \forall v\in V(c) ~\mbox{and every}~ \left(s^R_{v, c},
s^L_{V(c)\setminus \{v\},c}\right) ~\mbox{such that}~ s^R_{v,c}\in
{\cal F}_c(v) ~{\rm and} ~s^L_{V(c)\setminus \{v\},c}\in
{\tt F}^{-1}_c(s^R_{v,c}), \\
& & \forall u\in V(c)\setminus \{v\} ~\mbox{and every choice of}
~|C(u)\setminus \{c\}| ~\mbox{tokens on}~ \{u\},
\left\{s^R_{u,b}:b\in C(u)\setminus \{c\} \right\},\\
& &  ~\mbox{with each}~ s^R_{u,b}\in {\cal F}_b(u).\\
& \Leftrightarrow &
\bigcap_{b\in C(u)\setminus \{c\}} {\cal A}_b(u)
\subseteq {\tt F}_c\left(s^L_{V(c)\setminus \{u,v\}, c} \times s^R_{v,c}\right)
\\
& & \forall v\in V(c) ~\mbox{and every}~ \left(s^R_{v, c},
s^L_{V(c)\setminus \{v\},c}\right) ~\mbox{such that}~ s^R_{v,c}\in
{\cal F}_c(v) ~{\rm and} ~s^L_{V(c)\setminus \{v\},c}\in
{\tt F}^{-1}_c(s^R_{v,c}), \\
& & \mbox{and every} ~u\in V(c)\setminus \{v\}.
\end{eqnarray*}
\begin{eqnarray*}
& \Leftrightarrow &
{\cal A}_{\sim c}(u)
\subseteq {\tt F}_c\left(s^L_{V(c)\setminus \{u,v\}, c} \times s^R_{v,c}\right),
\\
& & \forall v\in V(c) ~\mbox{and every}~ \left(s^R_{v, c},
s^L_{V(c)\setminus \{v\},c}\right) ~\mbox{such that}~ s^R_{v,c}\in
{\cal F}_c(v) ~{\rm and} ~s^L_{V(c)\setminus \{v\},c}\in
{\tt F}^{-1}_c(s^R_{v,c}), \\
& & \mbox{and every} ~u\in V(c)\setminus \{v\}.\\
& \Leftrightarrow &
 \mbox{Constraint} ~\Gamma_c ~\mbox{is locally compatible.}
\end{eqnarray*}

Thus
\begin{eqnarray*}
& & \rho^{\rm norm (PTP)}_{c\rightarrow v} \leftrightarrow \rho^{* {\rm (SDBP)}}_{c\rightarrow v}, ~\mbox{for every} ~(x_v, \Gamma_c)\in E(G)
\\
&\Leftrightarrow&
\mbox{Every constraint} ~\Gamma_c ~\mbox{is locally compatible.}
\end{eqnarray*}
\end{proof}

Now it is easy to verify that for both $k$-SAT and $3$-COL problems, the
fact that PTP or weighted PTP can be reduced from BP with state-decoupling
condition imposed is due to the fact that every constraint is locally
compatible.

For $k$-SAT problems, as noted earlier,
${\cal F}_{c}(v)=\{{\bf L}, *\}$. If we pick $t_v$ to be either token from
${\cal F}_c(v)$, then for any $t'\in {\tt F}_c^{-1}(t_v)$ and
any $u\in V(c)\setminus \{v\}$, it can be verified that
${\tt F}_c\left(t'_{:V(c)\setminus \{u,v\}} \times t_v\right)=*$.
This makes ${\cal A}_{\sim c}(u)
\subseteq
{\tt F}_c\left(t'_{:V(c)\setminus \{u,v\}} \times t_v\right)$
always satisfied, independent of the factor graph structure of the problem instance.

For $3$-COL problems, as noted earlier, we see
${\cal F}_c(v)=\{{\bf 123}, {\bf 12}, {\bf 23}, {\bf 13} \}$.
Suppose that $u$ is the only other coordinate (except $v$) that is involved in
constraint $\Gamma_c$.
If we pick
$t_v$ to be any token from ${\cal F}_c(v)$, then ${\tt F}_c^u\left(t_v\right)=
{\bf 123}$. This again makes ${\cal A}_{\sim c}(u)
\subseteq
{\tt F}^u_c\left(t_v\right)$
always satisfied, independent of the factor graph structure of the problem instance.

That is, in both $k$-SAT and $3$-COL problems, the structure of each
local constraint {\em alone} guarantees the local compatibility
condition satisfied by every constraint, irrespective of how a
constraint interacts with other constraints (that are distant 2
apart) as is generally required in the local compatibility
condition. We generalize this fact in the following corollary --- 
immediately following Theorem \ref{thm:reduce_general} ---  which
provides a sufficient condition for SDBP to reduce to PTP without
relying on the interaction of neighboring constraints.  For CSPs
constructed with generic local constraint by random factor graph
structure, the corollary may turn out to be useful.

\begin{cor}
Let both the MRF of the CSP (specified
via (\ref{equ:leftfunction}), (\ref{equ:rightfunction}) and (\ref{equ:globalfunction})) and the weighted PTP for the CSP be
parametrized by the same $\{\omega_v:v\in V\}$. Suppose that every constraint
$\Gamma_c$ is such that
for any $v\in V(c)$, any forceable token $t_v\in {\cal F}_c(v)$,
any rectangle $t'\in {\tt F}_c^{-1}
\left(t_v\right)$ on $V(c)\setminus \{v\}$, and any $u\in V(c)\setminus \{v\}$, it holds that
\[
{\tt F}_c\left(t_v \times t'_{:V(c)\setminus
\{u,v\}}\right)=\left(\chi^*\right)^{\{v\}}.
\]

Then SDBP derived from the MRF is equivalent to
weighted PTP, where the correspondence is
\[
\rho^{\rm norm (PTP)}_{c\rightarrow v} \leftrightarrow \rho^{* {\rm (SDBP)}}_{c\rightarrow v}.
\]
\end{cor}

For completeness, we conclude this
section by constructing an example of CSP
in which the local compatibility condition is not satisfied by every
constraint.

\begin{figure}
\centerline{ \scalebox{0.6}{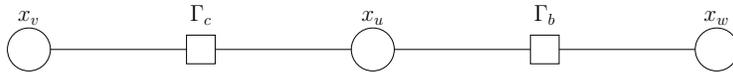} }
\caption{\label{fig:uncompatible} A portion of a factor graph $G$.}
\end{figure}

Suppose that $\Gamma_c$ and $\Gamma_b$ are two of the constraints
defining a CSP, and the factor graph representing the CSP locally
obeys the structure shown in Figure \ref{fig:uncompatible}. Suppose
that each variable of the CSP has alphabet $\chi=\{0, 1, 2\}$ and
that $\Gamma_c$ is defined as 
$\Gamma_c:=\{ (0_v, 0_u), (0_v, 1_u), (1_v, 2_u), (2_v, 2_u)\}$.
Suppose that $\Gamma_b$ is defined as
$
\Gamma_b:=\{ (0_u, 0_w), (1_u, 1_w), (2_u, 1_w) \}.
$
Note that ${\cal F}_c(v)=\{
{\bf 0}_v, {\bf 12}_v, {\bf 012}_v \}$, and it is easy to verify that
$
{\cal A}_{\sim c}(u)={\cal A}_b(u) ={\tt F}_b\left( {\bf 012}_w \right)={\bf 012}_u.
$
Now if we pick $t_v={\bf 0}_v$, then we have
${\cal A}_{\sim c}(u)\not\subseteq {\tt F}_c(t_v)={\bf 01}_u$. Thus constraint $\Gamma_c$ is not locally compatible, and following
Theorem \ref{thm:reduce_general}, PTP or weighted PTP can not be
reduced from SDBP for this CSP.

With this example, we see that it is not always the case that SDBP is SP.

\section{Concluding Remarks}
\label{sec:conclude}

In this paper, we study the question whether SP algorithms (non-weighted and weighted) are special cases
of BP for general constraint satisfaction problems.

The first contribution of this paper is a simple formulation of SP
algorithms for general CSPs as the weighted PTP algorithm. An
advantage of this formulation is that it has a probabilistically
interpretable update rule which allows SP algorithms to be developed
for arbitrary CSPs.

The second and main contribution of this paper is the answer to the
titular question in the most general context. We show that in
general, SP algorithms can not be reduced from the BP algorithm
derived from the MRF formalism in the style of \cite{SP:newlook} and
\cite{rtu:sp_isit_06}.  Such a reduction is only possible for
certain special cases where the notions of state-decoupling
condition and local compatibility condition are both satisfied.

It is worth noting that our answer to whether SP is BP is only
restricted to the MRF formalism in the style of \cite{SP:newlook} or
\cite{rtu:sp_isit_06}. Although this restriction is not completely
satisfactory, it appears to us that such an MRF formalism is the
most natural in light of the natural correspondence between the
states in the MRF and the SP messages (namely that left states
correspond to the ``intentions'' of variables and right states
correspond to the ``commands'' of the constraints). An additional
and perhaps even stronger justification of this MRF is its
combinatorial descriptive power as is elaborated in
\cite{SP:newlook} for $k$-SAT problems, which --- using the terminology
of this paper ---  captures the connectivity of the solution in the
space of all ``rectangles''. In fact, we conjecture that further
investigation of this perspective may provide useful insights into
the algorithm design for solving hard instances of CSPs, whether or
not SP or BP is considered as the choice of algorithms. \footnote{In
\cite{SP:newlook}, under the MRF formalism, Gibbs sampling-based
approach has also been presented as an algorithm for solving random
$k$-SAT problems.}

Further we note that the BP algorithm has been understood as a
special case of Generalized Belief Propagation (GBP) \cite{GBP}. In
that perspective, BP may be derived from iterative minimization of
the Bethe-approximation of the notion of free energy \cite{GBP}. The framework
of GBP allows a variety of ways  (unified under the notion of
``region graphs'') to approximate the free energy whereby leading to
a much richer family of BP-like algorithms. Given the results of
this paper, one may not want to exclude the possibility that certain
choice of free-energy approximation allows the corresponding GBP to
reduce to SP algorithms for general CSPs. Research along that
direction may still be of interest.

As the final remark, however, the authors of this paper would like to raise a philosophical question, in light of the simplicity in the (weighted) PTP formulation of SP and, in contrast, the complexity involved in reducing BP to SP: Should we attempt to seek a complicated explanation for a simple algorithm? Does the simplicity of SP (understood in terms of weighted PTP) imply a more natural, simpler but very different underlying graphical model --- beyond MRF --- that may
better explain SP?

\appendix

We now present some
results concerning the dynamics of SP, based on the formulation of  PTP and weighted PTP. 
These results, although rather elementary, should help provide intuitions
regarding what PTP is doing in solving a CSP.
We will start with the deterministic precursor of PTP, DTP.

\subsection{On the Dynamics of DTP}

We will refer to a subgraph $H$ of factor graph $G$ as a {\em
factor-subgraph} of $G$ if for every constraint vertex $\Gamma_c$ in
$H$, all neighboring variable vertices of $\Gamma_c$ in $G$ are also
in $H$. It is apparent that factor-subgraph $H$ is a factor graph
representing a CSP involving precisely a subset of the constraints
in $G$. We will denote by $C[H]$ the index set of all constraint
vertices in $H$, by $V[H]$  the index set of all variable vertices
in $H$, and by $\Gamma_H$ the set of all assignments on $V[H]$ that
satisfy every constraint $\Gamma_c$, $c\in C[H]$.

If factor-subgraph $H$ is a tree, it is also referred to as a {\em
factor tree} of $G$. For any factor tree $T$ of $G$, we will denote
by $L[T]$ the index set of all leaf vertices of $T$. Since we have
assumed that factor graph $G$ contains no degree-$1$ constraint
vertices, it is necessary that the leaf vertices of any factor tree
$T$ of $G$ are all variable vertices, i.e., that $L[T]$ contains no
index of any constraint vertex.

Suppose that $T$ is a factor tree of factor graph $G$, $U\subset V[T]$,
and $v\in V[T]\setminus U$. For any rectangle $t_U$ on $U$, define
\[
{\tt F}_{T}^{U\rightarrow v}(t_U):=
\left(\left(t_U\times \left(\chi^*\right)^{V[T]\setminus U}\right)\cap \Gamma_T\right)_{:\{v\}}.
\]

It is easy to see that function ${\tt F}_{T}^{U\rightarrow v}(\cdot)$
reduces to ${\tt F}_c^v(\cdot)$ introduced earlier, when $T$ contains a single factor
and $U$ is $V(c)\setminus \{v\}$.

Given a factor tree $T$ of $G$ and two vertices in $T$ indexed by
$a$ and $b$ respectively, we will introduce another notation of
message index,  ${a\stackrel{T}{\longrightarrow} b}$, which indexes
the message sent by the vertex with index $a$ along its only edge
that is on the path (in $T$) leading  to the vertex with index $b$.
For example, suppose that in factor tree $T$, constraint vertex
$\Gamma_c$ has a neighbor of $x_u$ and is on the path from $x_u$ to
$x_v$ in $T$, then message index $u\stackrel{T}{\longrightarrow} v$
is equivalent to $u\rightarrow c$, and
$t_{u\stackrel{T}{\longrightarrow} v}$ is equivalent to
$t_{u\rightarrow c}$.

A factor tree $T$ of $G$ will be referred to as
a $(v, l)$-tree of $G$ if the variable
vertex $x_v$ is in $T$, every leaf vertex in $T$ is distance $2l$
from vertex $x_v$, and
all vertices in $G$ that have distance to $x_v$
no larger than $2l$ are contained
in $T$. It is clear that given $G$, $v\in V$ and a positive integer $l$,
if a $(v, l)$-tree of $G$ exists, it is unique. We therefore denote it
by $T^l_{v}$.

Given $T_v^l$ of factor graph $G$, factor tree $T_{v-c}^l$ of $G$ is
the subgraph of $T_v^l$ induced by vertex $x_v$ and all vertices of
$T_v^l$ whose paths to $x_v$ (in $T_v^l$) traverse through vertex
$\Gamma_c$. On the other hand, factor tree $T_{v\not- c}^l$ is the
subgraph of $T_v^l$ induced by vertex $x_v$ and all vertices of
$T_v^l$ whose paths to $x_v$ (in $T_v^l$) {\em do not} traverse
through vertex $\Gamma_c$.

In what follows, we will use superscript $(l)$ on a message to
refer to the message in the $l^{\rm th}$ iteration.

\begin{prop}
\label{prop:DTP_on_tree_c2v}
Suppose that $l\ge 1$ and that
factor tree $T_v^l$ of factor graph $G$ exists. Then in iteration $l$
of DTP,
\[
t_{c\rightarrow v}^{(l)}= {\tt F}_{T^l_{v-c}}^{L[T_{v-c}^l]\rightarrow v}
\left(
\prod\limits_{u\in L[T_{v-c}^l]} t^{(1)}_{u\stackrel{T_{v-c}^l}{\longrightarrow}v}
\right).
\]

\end{prop}

\begin{proof}
We will prove this result by induction on $l$.

For the base case, we have
\begin{eqnarray*}
t_{c\rightarrow v}^{(1)}&=&{\tt F}_c^v
\left(
\prod\limits_{u\in V(c)\setminus \{v\}} t_{u\rightarrow c}^{(1)}
\right)\\
&=& {\tt F}_{T^1_{v-c}}^{L[T_{v-c}^1]\rightarrow v}
\left(
\prod\limits_{u\in L[T_{v-c}^1]} t^{(1)}_{u\stackrel{T_{v-c}^1}{\longrightarrow}v}
\right).
\end{eqnarray*}

As the inductive hypothesis, suppose that the result of this proposition holds
for a given iteration number $l\ge 1$. This implies specifically
that for every $u\in V(c)\setminus \{v\}$ and every $b\in C(u)\setminus \{c\}$,
\begin{eqnarray*}
t_{b\rightarrow u}^{(l)}
& = &
{\tt F}_{T_{u-b}^l}^{L[T_{u-b}^l]\rightarrow u}
\left(
\prod\limits_{w\in L[T_{u-b}^l]}
t^{(1)}_{w \stackrel{T_{u-b}^l}{\longrightarrow} u
}
\right).
\end{eqnarray*}
Then
\begin{eqnarray*}
t_{u\rightarrow c}^{(l+1)} & = &
\bigcap_{b\in C(u)\setminus \{c\}}
t_{b\rightarrow u}^{(l)}\\
& = &
\bigcap_{b\in C(u)\setminus \{c\}}
{\tt F}_{T_{u-b}^l}^{L[T_{u-b}^l]\rightarrow u}
\left(
\prod\limits_{w\in L[T_{u-b}^l]}
t^{(1)}_{w \stackrel{T_{u-b}^l}{\longrightarrow} u
}
\right)\\
& = & {\tt F}_{T^{l}_{u\not - c}}^ {L[T^{l}_{u\not - c}]\rightarrow
u} \left( \prod\limits_{w\in L[T^{l}_{u\not - c}] }
t^{(1)}_{w\stackrel{T^{l}_{u\not - c}}{\longrightarrow} u} \right).
\end{eqnarray*}

Finally,
\begin{eqnarray*}
t_{c\rightarrow v}^{(l+1)} & = &
{\tt F}_c^v
\left(
\prod\limits_{u\in V(c)\setminus \{v\}}
t_{u\rightarrow c}^{(l+1)}
\right)\\
& = &
{\tt F}_c^v
\left(
\prod\limits_{u\in V(c)\setminus \{v\}}
{\tt F}_{T^{l}_{u\not - c}}^
{L[T^{l}_{u\not - c}]\rightarrow u}
\left(
\prod\limits_{w\in
L[T^{l}_{u\not - c}]
}
t^{(1)}_{w\stackrel{T^{l}_{u\not - c}}{\longrightarrow} u}
\right)
\right)\\
& = &
{\tt F}_{T^{l+1}_{v-c}}^{L[T_{v-c}^{l+1}]\rightarrow v}
\left(
\prod\limits_{w\in L[T_{v-c}^{l+1}]} t^{(1)}_{w\stackrel{T_{v-c}^{l+1}}
{\longrightarrow}v}
\right).
\end{eqnarray*}

This completes the proof.
\end{proof}

Translating this results to summary tokens, the following result can be obtained immediately.

\begin{cor}
\label{cor:DTP_on_tree_summary}
Suppose that $l \ge 1$ and that
factor tree $T_v^l$ of factor graph $G$ exists. Then in iteration $l$ of DTP,
\[
t_{v}^{(l)}= {\tt F}_{T^l_{v}}^{L[T_{v}^l]\rightarrow v}
\left(
\prod\limits_{u\in L[T_{v}^l]} t^{(1)}_{u\stackrel{T_{v}^l}{\longrightarrow}v}
\right).
\]
\end{cor}

The implication of this result is that on factor graph with sufficiently
large girth, DTP is in fact very well-behaved: the summary token at any
variable $x_v$ in iteration $l$ depends precisely on the initial
tokens passed by variables that are $2l$ away from $x_v$. Specifically, one may view those tokens form a rectangle on $L[T_v^l]$, and the
summary token at $x_v$ in iteration $l$ is precisely the set of
all assignments on $\{v\}$
that can make $\Gamma_{T^l_v}$ satisfied, given the assignment on
$L[T_v^l]$ is from that rectangle.

Now we develop some results of DTP that require no ``local cycle-freeness'' in
the factor graph.

\begin{lem}
\label{lem:DTP_cycle_1}

At every $v\in V$ and for any $l$,
\[t_v^{(l)}= \bigcap_{c\in C(v)}t_{v\rightarrow c}^{(l+1)}.\]
\end{lem}
\begin{proof}
Suppose that $x_v\in t_v^{(l)}$. Then $x_v\in t_{c\rightarrow v}^{(l)}$ for
every $c\in C(v)$, by the definition of summary messages.  It follows that
 $x_v\in t_{v\rightarrow c}^{(l+1)}$ for every $c\in C(v)$. Then
$x_v\in \bigcap_{c\in C(v)}t^{(l+1)}_{v\rightarrow c}$. This shows that
$t^{(l)}_v\subseteq \bigcap_{c\in C(v)}t^{(l+1)}_{v\rightarrow c}$.

On the other hand, suppose that $x_v\in \bigcap_{c\in C(v)}t^{(l+1)}_{v\rightarrow c}$. Then $x_v\in t^{(l+1)}_{v\rightarrow c}=
\bigcap_{b\in C(v)\setminus \{c\}}t^{(l)}_{b\rightarrow v}$, for every $c\in C(v)$.  It follows that $x_v\in t^{(l)}_{b\rightarrow v}$ for
every $b\in C(v)$, giving rise to that
$x_v\in \bigcap_{b\in C(v)}t^{(l)}_{b\rightarrow v}=t_v^{(l)}.$  Thus
$\bigcap_{c\in C(v)}t_{v\rightarrow c}^{(l+1)} \subseteq t_v^{(l)}.$

Therefore $t_v^{(l)}= \bigcap_{c\in C(v)}t_{v\rightarrow c}^{(l+1)}$.
\end{proof}

\begin{lem}
\label{lem:DTP_cycle_2}
Suppose that $\hat{x}_V$ is a satisfying assignment on $V$,
 namely that $\hat{x}_V$ satisfies (\ref{equ:csp}). If
$\hat{x}_V\in
\prod\limits_{v\in V}\bigcap\limits_{c\in C(v)} t_{v\rightarrow c}^{(l)}$
in some iteration $l$,
then $\hat{x}_V\in \prod\limits_{v\in V}t_v^{(l)}$.
\end{lem}

\begin{proof}
The fact that $\hat{x}_V \in \prod\limits_{v\in
V}\bigcap\limits_{c\in C(v)} t_{v\rightarrow c}^{(l)}$ implies that
for every $v\in V$ and $c\in C(v)$, $\hat{x}_{V:\{v\}} \in
\bigcap\limits_{c\in C(v)} t_{v\rightarrow c}^{(l)} \subseteq
t_{v\rightarrow c}^{(l)}$, and hence via the ``monotonicity'' of
function ${\tt F}_c$, ${\tt F}_c\left(\{\hat{x}_{V:V(c)\setminus
\{v\}}\}\right) \subseteq {\tt F}_c\left(\prod\limits_{u\in
V(c)\setminus \{v\}} t_{u\rightarrow c}^{(l)} \right) =
t_{c\rightarrow v}^{(l)}.$ Incorporating that $\hat{x}_V$ is a
satisfying assignment, we see that $\hat{x}_{V:\{v\}} \in {\tt
F}_c\left(\{\hat{x}_{V:V(c)\setminus \{v\}}\}\right) \subseteq
 t_{c\rightarrow v}^{(l)}$, for every $v\in V$ and $c\in C(v)$. Thus
$\hat{x}_{V:\{v\}}\in \bigcap\limits_{c\in C(v)}
t_{c\rightarrow v}^{(l)} = t_v^{(l)}.$ It then follows that
$\hat{x}_V\in \prod\limits_{v\in V}t^{(l)}_{v}$.
\end{proof}

\begin{prop}
\label{prop:DTP_on_cycle}
Suppose that $\hat{x}_V$ is a satisfying assignment and that
the initialization of DTP is such that
$\hat{x}_{V:\{v\}} \in t^{(1)}_{v\rightarrow c}$ for every $v\in V$ and
$c\in C(v)$. Then in any iteration $l$, the rectangle
$\prod\limits_{v\in V}t^{(l)}_v$ formed by the summary tokens contains
$\hat{x}_V$.
\end{prop}

\begin{proof}
At iteration 1, the fact that $\hat{x}_{V:\{v\}} \in
t^{(1)}_{v\rightarrow c}$ for every $v\in V$ and $c\in C(v)$ implies
that $\hat{x}_V\in\prod_{v\in V}\bigcap_{c\in
C(v)}t^{(1)}_{v\rightarrow c}$. Followed by Lemma \ref{lem:DTP_cycle_2},
we have
$\hat{x}_V\in\prod_{v\in V}t^{(1)}_{v}$.

As the inductive hypothesis, suppose we have $\hat{x}_{V}\in\prod_{v\in
V}t^{(l)}_{v} $ at iteration $l$. At iteration $l+1$, followed by
Lemma \ref{lem:DTP_cycle_1}, we have $\hat{x}_{V}\in\prod_{v\in V}\bigcap_{c\in
C(v)}t^{(l+1)}_{v\rightarrow c}$. Then by Lemma \ref{lem:DTP_cycle_2},
$\hat{x}_{V}\in\prod_{v\in V}t^{(t+1)}_v$.

Therefore, this proposition is proved by induction.
\end{proof}

At this end, we have shown that if DTP is initialized to
``containing'' a satisfying assignment, then this assignment is
contained in the rectangle formed by the summary tokens in all
iterations. That is, the solution of the CSP will never get ``lost''
during DTP iteration provided that it is contained in the initial
rectangle. This result (Proposition \ref{prop:DTP_on_cycle}) and
Corollary \ref{cor:DTP_on_tree_summary} presented earlier will
become useful when we discuss the dynamics of PTP.

\subsection{On the Dynamics of PTP and Weighted PTP}

We now turn our attention to (non-weighted) PTP.

Denote by $G_v^l$ the factor-subgraph of $G$ which contains all
factors whose messages have propagated to variable $x_v$ by the end
of PTP iteration $l$. That is, $G_v^l$ is the factor-subgraph of $G$
that contains variable vertex $x_v$ and all vertices whose distances
to $x_v$ are no larger than $2l$. It is apparent that if $G_v^l$ is
a tree, then it is the $(v,l)$ factor tree $T_v^l$.

Let $l^*$ be the smallest $l$ such that at least for one $v\in V$,
$T_v^l$ does not exist. Denote $m_v(l):=\left| \left( \Gamma_{G_v^l}
\right)_{:\{v\}}\right|$. That is, $m_v(l)$ is the number of
assignments of variable $x_v$ that can make all constraints in
$G_v^l$ satisfied. Clearly, $m_v(l)$ is a non-increasing function of
$l$.

We will first restrict the CSP to a ``single-solution CSP'', i.e.,
having exactly one satisfying assignment. We will denote this assignment
on $V$ by $\hat{x}_V$.

Let $\hat{l}$ be the smallest $l$ for which
$\min\limits_{v} m_v(l)=1$.  It is worth noting that such $\hat{l}$ exists since
the CSP has precisely one solution. Let $\hat{v}$ satisfy $m_{\hat{v}}(\hat{l})= 1$.

\begin{prop}
\label{prop:PTP_on_tree}
Let factor graph $G$ represent a single-solution CSP.
Suppose that the initialization of PTP is such that every left message
$\lambda_{v\rightarrow c}^{(1)}(t)$ is strictly positive for every
$t\in \left(\chi^*\right)^{\{v\}}$.
If $\hat{l}<l^*$, then

\[
\mu_{\hat{v}}^{{\rm norm}~ (\hat{l})}(t)=[t=\{\hat{x}_{V:\{\hat{v}\}}\}].
\]
\end{prop}

\begin{proof}
This result relies on Corollary \ref{cor:DTP_on_tree_summary}.

First, $\hat{l}<l^*$ implies that $(\hat{v}, \hat{l})$ factor tree
$T_{\hat{v}}^{\hat{l}}$ exists. Then by Corollary
\ref{cor:DTP_on_tree_summary}, if DTP is initialized such that the
tokens sent from the leaves of $T_{\hat{v}}^{\hat{l}}$ form
 $ \prod\limits_{u\in L[T^{\hat{l}}_{\hat{v}}]}t^{(1)}_{u\stackrel{T^{\hat{l}}_{\hat{v}}}{\longrightarrow}\hat{v}}$, then
the summary token at $v$ in the $\hat{l}^{\rm th}$ iteration is
${\tt F}_{T^{\hat{l}}_{\hat{v}}}^{L[T_{\hat{v}}^{\hat{l}}]\rightarrow \hat{v}}
\left(
\prod\limits_{u\in L[T^{\hat{l}}_{\hat{v}}]}t^{(1)}_{u\stackrel{T^{\hat{l}}_{\hat{v}}}{\longrightarrow}\hat{v}}
\right)$.

Since $\hat{v}$ satisfies $m_{\hat{v}}(\hat{l})= 1$, it is necessary that
${\tt F}_{T^{\hat{l}}_{\hat{v}}}^{L[T_{\hat{v}}^{\hat{l}}]\rightarrow \hat{v}}
\left(
\prod\limits_{u\in L[T^{\hat{l}}_{\hat{v}}]}t^{(1)}_{u\stackrel{T^{\hat{l}}_{\hat{v}}}{\longrightarrow}\hat{v}}
\right)$ is either token $\{\hat{x}_{V:\{v\}}\}$ or $\emptyset$, which depends
on the rectangle initialized.

Now PTP on $T_{\hat{v}}^{\hat{l}}$, with respect to $x_{\hat{v}}$,
may be understood as initializing a {\em random} rectangle on
$L[T_{\hat{v}}^{\hat{l}}]$ (the distribution of which is characterized by
the product of the initial messages), transforming the random rectangle
to random token on $\hat{v}$ via a functional mapping
${\tt F}_{T^{\hat{l}}_{\hat{v}}}^{L[T_{\hat{v}}^{\hat{l}}]\rightarrow \hat{v}}
\left(\cdot
\right)$, and conditioning on the resulting token being valid (non-empty set).
The fact that initial messages of PTP are strictly positive assures that
every rectangle on $L[T_{\hat{v}}^{\hat{l}}]$ has non-zero
probability during initialization. After conditioning on the resulting token being valid, the token $\emptyset$ is removed from the allowed realization of
the resulting token and thus
the resulting token equals $\{\hat{x}_{V:\{v\}}\}$ with probability 1. This completes the
proof.
\end{proof}

This result and its proof can be easily extended to a somewhat
larger family of CSPs each containing multiple solutions, as shown in the next proposition.

\begin{prop}
\label{prop:PTP_on_tree_2}
Suppose that in the CSP, there exist a coordinate $\hat{v}\in V$  and
an assignment $\hat{x}_{\hat{v}}\in \left(\chi^*\right)^{\{v\}}$
such that every satisfying configuration $\tilde{x}_V \in \Gamma$
satisfies $\tilde{x}_{V:\{v\}}=\hat{x}_v$. If for some integer $\hat{l}$,
$T_{\hat{v}}^{\hat{l}}$ exists and $m_{\hat{v}}(\hat{l})=1$, then
\[
\mu_{\hat{v}}^{{\rm norm}~ (\hat{l})}(t)=[t=\{\hat{x}_{\hat{v}}\}].
\]
\end{prop}

The proof is similar to that for proposition \ref{prop:PTP_on_tree}, which
essentially relies on Corollary \ref{cor:DTP_on_tree_summary} and that the local tree rooted at $\hat{v}$ is large enough. Skipping the proof, we note that
Proposition \ref{prop:PTP_on_tree} may be viewed as a special case of
Proposition \ref{prop:PTP_on_tree_2}.

Based on the results above, we provide some remarks concerning the dynamics
of PTP and argue intuitively how it solves a CSP.

\begin{enumerate}
\item Similar to what was argued in the proof of Proposition \ref{prop:PTP_on_tree},
 the key insight regarding what PTP is doing is that PTP updates a {\em random}
rectangle whose sides are distributed independently.

At the initialization stage, PTP defines a random rectangle on $V$,
where the sides of the random rectangles are treated as independent random variables.  In every iteration, PTP maps this random rectangle to a new random rectangle in the following steps.
\begin{enumerate}
\item Apply a functional mapping defined by the right-message
update rule and the left-message update rule.
\item
Eliminate  the resulting empty rectangles (via conditioning on that each
side of the resulting random rectangle is not the empty set and
re-normalization).
\item Take the marginal distribution of the resulting random
rectangle on each side variable, and treat all sides as being independent
random variables. This defines a new random rectangle.
\end{enumerate}
PTP iterates over these steps to continuously update the random rectangle.

\item For single-solution CSPs, based on Proposition \ref{prop:PTP_on_tree},
if the girth of the graph is large enough,
at least one side of the new rectangle, after some iterations,
becomes deterministic,
namely the singleton set containing the correct assignment for that variable. This would allow the decimation procedure to fix this variable to the correct
assignment and reduce the problem.
Similar results hold for CSPs having more than one solutions but in which all solutions share a single assignment on some coordinate. By Proposition \ref{prop:PTP_on_tree_2}, in this case, when the local tree rooted
at that variable is sufficiently large, PTP will find that variable and its
correct assignment.
Of course, the condition of Proposition \ref{prop:PTP_on_tree}
and that of Proposition \ref{prop:PTP_on_tree_2}, namely that there is a sufficiently large local tree rooted at a variable and that the variable
only has one correct assignment, may not
hold in reality. As a consequence, no side of the random rectangle is
 deterministically a singleton. In that case, the
decimation procedure must deal with this ambiguity --- resulted from non-ideal
factor graph structure and the complexity of the solution space --- and
make a good guess to fix a variable.

\item \label{arg:PTP_tree} Proposition \ref{prop:PTP_on_tree} and Proposition \ref{cor:DTP_on_tree_summary} also suggest that when the graph has large girth (and when the solutions share one common assignment on some coordinate),
as PTP iterates, the rectangles containing no solutions will be gradually removed from the sample space of the random rectangle.

\item \label{arg:PTP_cycle} Proposition \ref{prop:DTP_on_cycle} implies that regardless of
cycle structure of the graph, all solution-containing rectangles will be
kept (possibly in a form of combining  each other) over PTP iterations.

\item Combining \ref{arg:PTP_tree}) and \ref{arg:PTP_cycle}) above, one may
view  each PTP iteration as performing a ``filtering'' operation on the
distribution of the random rectangle. As the distribution of the
random rectangle evolves,   the probability mass moves gradually
to one biased to some solution-containing rectangles. When the graph has large
girth and some coordinate is in a ``favorable'' position (in a sense combining
its location in the graph and its role in the solution space),
the summary message at this coordinate may become more deterministically biased to a singleton token, making decimation possible.
\end{enumerate}

Finally, we briefly remark on weighted PTP.

Similar to PTP, weighted PTP also updates a random rectangle.
However, instead of
using a functional mapping, in step a) of the above procedure, it uses
a conditional distribution. By examining the form of the
obedience conditionals, it is intuitive that comparing with PTP, weighted
PTP shifts the distribution of each side of the
random rectangle more towards
``smaller''  tokens on each coordinate. (Here $t_v$ is said to be smaller than
$t'_v$ if $t_v\subset t'_v$.) This provides the algorithm better
opportunity to lead to some side of the random rectangle
more deterministically biased to a singleton.




\bibliographystyle{IEEEbib}




\end{document}